\def\notes{0}
\def\lowerbds{0}
\def\oakland{0} %

\ifnum\oakland=1
    \documentclass[conference,compsoc]{IEEEtran}
\else
    \documentclass{article}
\fi

\usepackage[normalem]{ulem}
\usepackage[utf8]{inputenc}
\usepackage{amsmath,amsthm,amsfonts,amscd,amssymb}
\usepackage{fullpage}
\usepackage{hyperref}
\usepackage[capitalize]{cleveref}
\usepackage{mathtools}
\usepackage{microtype}
\usepackage{setspace}
\usepackage{tabularx}
\usepackage{makecell}
\usepackage{bm}
\usepackage{enumitem}
\usepackage{mdframed}
\usepackage{hhline}
\usepackage{xspace}
\usepackage{stackengine}
\usepackage{multirow}

\ifnum\notes=1 
    \ifnum\oakland=0 
    \usepackage{showlabels}
    \fi
\fi

\usepackage[dvipsnames,table,xcdraw]{xcolor}
\hypersetup{colorlinks = true,linkcolor = Periwinkle, anchorcolor = red, citecolor = Periwinkle, filecolor = orange,urlcolor = cyan}

\usepackage{algorithm}
\usepackage[noend]{algpseudocode}

\usepackage{thmtools}
\usepackage{thm-restate}

\usepackage{tikz}
\ifnum\oakland=1
    \usepackage[style=ieee-alphabetic,maxbibnames=99,maxalphanames = 5]{biblatex}
\else
    \usepackage[style=alphabetic,maxbibnames=99,maxalphanames = 5,backref=true]{biblatex}
\fi
\DefineBibliographyStrings{english}{
  backrefpage  = {cited on p.},     %
  backrefpages = {cited on pp.},    %
}

\DeclareFieldFormat{organization}{\itshape{#1}}

\AtEveryBibitem{%
 \clearlist{address}
 \clearfield{date}
 \clearfield{eprint}
 \clearfield{isbn}
 \clearfield{issn}
 \clearlist{location}
 \clearfield{month}
 \clearfield{series}
 \clearlist{publisher}
 \clearname{editor}
 \clearfield{volume}
 \clearfield{pages}
 \clearfield{number}
}

\newcommand{\msf}{\mathsf}

\newtheorem{theorem}{Theorem}[section]

\newtheorem{lemma}[theorem]{Lemma}
\newtheorem{claim}[theorem]{Claim}

\theoremstyle{definition}
\newtheorem{remark}[theorem]{Remark}
\newtheorem{definition}[theorem]{Definition}
\newtheorem{intuition}[theorem]{Intuition}

\crefname{claim}{Claim}{Claims}

\newcommand{\paren}[1]{{\left( {#1} \right)}}
\newcommand{\set}[1]{{\left\{ {#1} \right\}}}
\newcommand{\floor}[1]{{\lfloor {#1} \rfloor}}
\newcommand{\ceil}[1]{{\left\lceil {#1} \right\rceil}}

\newcommand{\Z}{\mathbb{Z}}
\newcommand{\R}{\mathbb{R}}
\newcommand{\N}{\mathbb{N}}
\newcommand{\mA}{\mathcal{A}}
\newcommand{\mB}{\mathcal{B}}
\newcommand{\mC}{\mathcal{C}}
\newcommand{\mD}{\mathcal{D}}

\newcommand{\mG}{\mathcal{G}}

\newcommand{\mM}{\mathcal{M}}

\newcommand{\mS}{\mathcal{S}}
\newcommand{\mX}{\mathcal{X}}
\newcommand{\mY}{\mathcal{Y}}
\newcommand{\mZ}{\mathcal{Z}}

\newcommand{\Otilde}{\widetilde{O}}
\newcommand{\Thetatilde}{\widetilde{\Theta}}

\newcommand{\zo}{\{0,1\}}

\newcommand{\linf}{\ell_\infty}
\newcommand{\lone}{\ell_1}

\newcommand{\hf}{\widehat{f}}

\newcommand{\gs}{\mathit{S}}

\newcommand{\mGS}{\mathcal{S}}

\newcommand{\dg}{\mathsf{DG}}
\newcommand{\pdg}{\mathsf{PDG}}
\newcommand{\satstep}{\mathsf{SatStage}} %
\newcommand{\projstage}{\mathsf{ProjStage}} %
\newcommand{\projstagetext}{{projection stage}}
\newcommand{\projstagealttext}{{is processed at \projstagetext}}
\newcommand{\projstagealttexttwo}{{is processed at a \projstagetext}}
\newcommand{\projstagealttextthree}{{is processed prior to this \projstagetext}}
\newcommand{\projstagealttextfour}[1]{{at the \projstagetext{} when #1 is processed}}
\newcommand{\satsteptext}{{saturation stage}}
\newcommand{\inccrit}{inclusion criterion}
\newcommand{\Ered}{E_\mathit{red}}
\newcommand{\Eblue}{E_\mathit{blue}}
\newcommand{\ecolor}{\mathsf{color}}
\newcommand{\gdiff}{\Delta}

\newcommand{\cut}{\text{order-induced cut}}

\newcommand{\eps}{\varepsilon}
\newcommand{\del}{\delta}
\renewcommand{\phi}{\varphi}

\newcommand{\ord}[2][th]{\ensuremath{{#2}^{\mathrm{#1}}}}
\newcommand{\simed}{\approx_{\eps,\delta}}

\ifnum\notes = 1
    \newcommand{\mystrikeout}[2]{{\color{#1} \sout{#2}}} %

    \newcommand{\pj}[1]{{\color{blue}[P: #1]}}
    \newcommand{\pjnote}[1]{{\color{blue}\footnote{{\color{blue} {\bf P:} #1}}}}
    \newcommand{\pjtodo}[1]{{\color{violet}[to do]\footnote{{\color{violet} {\bf P:} #1}}}}
    \newcommand{\oldtext}[1]{}
    \newcommand{\adaminline}[1]{{\color{brown} (adam)~#1}}
    \newcommand{\adamnote}[1]{{\color{ForestGreen}\footnote{{\color{brown} {\bf A:} #1}}}}
    
    \newcommand{\connor}[1]{{\color{Mulberry}[Connor] #1}}

    \newcommand{\connornote}[1]{{\color{red}\footnote{{\color{Mulberry} {\bf C:} #1}}}}
    
    \newcommand{\connorq}[1]{{\color{red}\footnote{{\color{red} {\bf C:(question)} #1}}}}
    \newcommand{\connorqd}[1]{{\color{gray}\footnote{{\color{gray} {\bf C:(discussed question)} #1}}}}
    \newcommand{\note}[1]{{\color{blue}(note)~#1}}
    \newcommand{\question}[1]{{\color{red}(question)~#1}}
    \newcommand{\todo}[1]{{\color{red}{(TO-DO ITEM)~#1}}}
    \newcommand{\status}[1]{{\color{blue}{(STATUS)~#1}}}
\else
    \newcommand{\mystrikeout}[2]{}
    \newcommand{\pj}[1]{}
    \newcommand{\pjnote}[1]{}
    
    \newcommand{\pjtodo}[1]{}
    \newcommand{\oldtext}[1]{}
    \newcommand{\adaminline}[1]{}
    \newcommand{\adamnote}[1]{}

    \newcommand{\connor}[1]{}

    \newcommand{\connorq}[1]{}
    \newcommand{\connornote}[1]{}
    \newcommand{\connorqd}[1]{}
    \newcommand{\note}[1]{}
    \newcommand{\question}[1]{}
    \newcommand{\todo}[1]{}
    \newcommand{\status}[1]{}
    \renewcommand{\sout}{}
\fi

\ifnum\oakland=1
    \newcommand{\noakland}[1]{}
    \newcommand{\yoakland}[1]{#1}
\else
    \newcommand{\noakland}[1]{#1}
    \newcommand{\yoakland}[1]{}
\fi

\newcommand{\Ex}{\mathop{\mathbb{E}}}
\newcommand{\pr}[1]{\Pr\left[{#1}\right]}

\newcommand{\err}{\textsf{ERR}}

\newcommand{\margs}{\textsf{Marginals}}

\newcommand{\edges}{\ensuremath{f_{\msf{edges}}}}

\newcommand{\triangles}{\ensuremath{f_{\msf{triangles}}}}
\newcommand{\connectedcomps}{\ensuremath{f_{\msf{CC}}}}
\newcommand{\kstars}{\ensuremath{f_{\msf{k\text{-}stars}}}}
\newcommand{\deghist}{\ensuremath{f_{\msf{degree\text{-}hist}}}}

\newcommand{\subgraphedges}{\mathsf{SubgraphEdges}}
\newcommand{\maxflowedges}{\mathsf{MaxflowEdges}}
\newcommand{\svt}{\mathsf{SVT}}
\newcommand{\DistToGraph}{\mathsf{DistToGraph}}
\newcommand{\dtg}{\DistToGraph}
\newcommand{\style}{\mathsf{c}}
\newcommand{\addedge}{\mathsf{add\_edge}}

\newcommand{\original}{\mathit{original}}
\newcommand{\projection}{\mathit{projected}}

\newcommand{\chist}[2]{{\mathsf{ch}_{#1}\paren{#2}}}

\newcommand{\pluseq}{\mathrel{+}=}

\newcommand{\mech}{\text{\sf RestrictedPrivAlg}}

\newcommand{\bben}{\textsf{BBRestrictedToNodePriv}}

\newcommand{\countsf}{\mathsf{count}}

\newcommand{\true}{\mathit{True}}
\newcommand{\false}{\mathit{False}}

\newcommand{\test}{\mathsf{Test}}
\newcommand{\base}{\mathsf{Base}}
\newcommand{\failtest}{\beta_\test}

\newcommand{\ptr}{\mathsf{PTR}}

\newcommand{\Above}{\bot}
\newcommand{\Below}{\top}
\newcommand{\verdict}{\mathsf{verdict}}

\newcommand{\passed}{\mathsf{passed}}

\newcommand{\nodeel}{\partial V}
\newcommand{\edgeel}{\partial E}

\newcommand{\dedge}{d_\mathit{edge}}
\newcommand{\dnode}{d_\mathit{node}}

\newcommand{\outcolor}{\mathsf{out\text{-}color}}

\newcommand{\nnode}{\simeq_\mathit{node}}
\newcommand{\nedge}{\simeq_\mathit{edge}}

\newcommand{\vfl}{v^\text{fl}}

\newcommand{\lap}{\mathit{Lap}}

\def\polylog{\operatorname{polylog}}
\def\poly{\operatorname{poly}}

\newcommand{\collapse}[1]{\mathsf{flatten}\paren{#1}}
\newcommand{\flatten}[1]{\collapse{#1}}
\newcommand{\flatt}[1]{\mathsf{flatten}(#1)}

\newcommand{\projgen}[1][D]{\Pi_{#1}}
\newcommand{\projo}[1][D]{\Pi^{\text{\tiny BBDS}}_{#1}}
\newcommand{\projp}[1][D]{\Pi^{\text{\tiny DLL}}_{#1}}

\newcommand{\vplus}{v^+}
\newcommand{\eplus}{e^+}

\newcommand{\id}{\mathsf{id}}

\newcommand{\red}{\mathit{red}}
\newcommand{\blue}{\mathit{blue}}

\newcommand{\tildeptr}{\ensuremath{\widetilde{\ptr}}}
\newcommand{\ttildeptr}{\ensuremath{\widetilde{\tildeptr}}}

\algdef{SE}[SUBALG]{Indent}{EndIndent}{}{\algorithmicend\ }%
\algtext*{Indent}
\algtext*{EndIndent}

\newlength\myindent %
\newcommand{\commentstyle}[1]{{\footnotesize \textcolor{CadetBlue}{#1}}}

\newcommand{\myparagraph}[1]{\smallskip \noindent {\bf{#1}}}

\makeatletter
\newcommand{\thickhline}{%
    \noalign {\ifnum 0=`}\fi \hrule height 1pt
    \futurelet \reserved@a \@xhline
}
\newcolumntype{"}{@{\hskip\tabcolsep\vrule width 1pt\hskip\tabcolsep}}
\makeatother

\usepackage{fancyhdr}

\yoakland{\pagenumbering{gobble}}

\begin{document}

\title{Time-Aware Projections:\\ Truly Node-Private Graph Statistics under Continual Observation\yoakland{*}
}
\ifnum\oakland=1
    \definecolor{emailcolor}{rgb}{0, 0, 1}
    \author{
        \IEEEauthorblockN{Palak Jain}
        \IEEEauthorblockA{Boston University \\ \href{mailto:palakj@bu.edu}{\color{black}palakj@bu.edu}}
        \and
        \IEEEauthorblockN{Adam Smith}
        \IEEEauthorblockA{Boston University \\ \href{mailto:ads22@bu.edu}{\color{black}ads22@bu.edu}}
        \and
        \IEEEauthorblockN{Connor Wagaman}
        \IEEEauthorblockA{Boston University \\ \href{mailto:wagaman@bu.edu}{\color{black}wagaman@bu.edu}}
    }
\else
    \author{Palak Jain\thanks{Boston University, {\tt \{palakj,ads22,wagaman\}@bu.edu}.
    Supported in part by NSF awards CCF-1763786 and CNS-2120667, Faculty Awards from Google and Apple, and Cooperative Agreement CB16ADR0160001 with the Census Bureau. The views expressed in this paper are those of the authors and not those of the U.S. Census
    Bureau or any other sponsor.
    } 
    	\and
    	Adam Smith\footnotemark[1]
    	\and
    	Connor Wagaman\footnotemark[1]
    }
\fi
\date{ November 5, 2025}

\maketitle

\begin{abstract} 

    Releasing differentially private statistics about social network data is challenging: one individual's data consists of a node and all of its connections, and typical analyses are sensitive to the insertion of a single unusual node in the network. This challenge is further complicated in the \textit{continual release} setting, where the network varies over time and one wants to release information at many time points as the network grows. 
    Previous work addresses node-private continual release by assuming an unenforced promise on the maximum degree in a graph, but leaves open whether such a bound can be verified or enforced privately.

    In this work, we describe the first algorithms that satisfy the standard notion of node-differential privacy in the continual release setting (i.e., without an assumed promise on the input streams). These algorithms are accurate on sparse graphs, for several fundamental graph problems: counting edges, triangles,  other subgraphs,  and connected components; and releasing degree histograms.
    Our unconditionally private algorithms generally have optimal error, up to polylogarithmic factors and lower-order terms.
    
    We provide general transformations that take a base algorithm 
    for the continual release setting, which need only be private  for streams  satisfying a promised degree bound, and produce an 
    algorithm that is unconditionally private 
    yet mimics the base algorithm when the stream meets the degree bound (and adds only linear overhead to the time and space complexity of the base algorithm).
    To do so, we design new projection algorithms for graph streams, based on the batch-model techniques of  \cite{BlockiBDS13,DayLL16}, which modify the stream to limit its degree. 
    Our main technical innovation is to show that   
    the projections are \emph{stable}---meaning that similar input graphs have similar projections---when the input stream satisfies a privately testable safety condition. Our transformation then follows a novel online variant of the Propose-Test-Release framework \cite{DL09}, 
    privately testing the safety condition before releasing output at each step.
\end{abstract}

\ifnum\oakland=0 %
    \newpage  
    {
    \setstretch{1}
    \hypersetup{linkcolor=black}
    \tableofcontents
    }
    \newpage
\fi

\section{Introduction}\label{sec:intro}

\fancyhf{} %
\renewcommand{\headrulewidth}{0pt}

\thispagestyle{fancy}

\yoakland{\lfoot{\footnotesize{*~A full version of this paper is available on arXiv \cite{JainSW24}.}}}

Graphs provide a flexible and powerful way to model and represent relational data, such as social networks, epidemiological contact-tracing data, and employer-employee relationships. 
Counts of substructures---edges, nodes of a given degree, triangles, connected components---are fundamental statistics that shed light on a network's structure and are the focus of extensive algorithmic study.
For example, edge counts can quantify relationships in a social network; a function of triangle and 2-star counts called the ``correlation coefficient'' or ``transitivity'' of a network is of interest to epidemiologists for understanding disease spread \cite{BadhamS10,YoungJMHH13}; and connected component counts have been used for determining the number of classes in a population \cite{Goo49} and estimating fatalities in the Syrian civil war \cite{ChenSS18}.

When the graph %
contains sensitive information about individuals, 
one must balance the accuracy of released statistics 
with those individuals' privacy. 
Differential privacy \cite{DworkMNS16} is a widely studied and deployed framework for quantifying
such a trade-off. It requires that the output of an algorithm reveal little about any single individual's record (even hiding its presence or absence in the data set).

In this work, we study differentially private algorithms that  
continually monitor several fundamental statistics about a graph that evolves over time.
We consider the \emph{continual release} (or \emph{continual observation}) model of differential privacy~\cite{DNPR10,CSS11} 
in which the input data is updated over time and statistics about it must be released continuously.
(In contrast, in the \emph{batch} model, input arrives in one shot and output is produced only once.)

There are two standard notions of differential privacy (DP) for algorithms that operate on graph data: (1) \emph{edge DP} \cite{NRS07}, for which the algorithm must effectively obscure 
the information revealed by any individual edge (including its mere presence or absence), 
and (2) \emph{node DP} \cite{HLMJ09,BlockiBDS13,KasiviswanathanNRS13,ChenZ13}, for which the algorithm must obscure 
the information revealed by a node's entire set of connections (including even the node's presence or absence).

Edge privacy is typically easier to achieve and more widely studied. However, in 
social networks and similar settings,
nodes---rather than edges---correspond to individuals and so
node privacy is more directly relevant.
Indeed, existing attacks infer sensitive information about a person from aggregate information about their neighborhood in the network
(e.g., sexuality can be inferred from an individual's Facebook friends \cite{JerniganM09gaydar}),
showing that privacy at the node level rather than only the edge level is important.

\myparagraph{Background: Node-differential Privacy.}
To understand the challenges of designing node-private algorithms, consider the task of estimating the number of edges in a graph. For every graph $G$ with $n$ vertices, there is a node-neighboring graph $G'$ with $n$ more edges (obtained by adding a new, high-degree node connected to all existing nodes). A node private algorithm, however, must hide 
the difference between $G$ and $G'$. Therefore, every node-private algorithm must have %
additive error $\Omega(n)$
on  either $G$ or $G'$, which means large relative error when $G$ and $G'$ are sparse. %
It is thus impossible to get a useful worst-case accuracy guarantee for %
counting the %
edges in a graph or, for similar reasons, many other basic statistics. 

As a result, node private algorithms are often tailored to specific %
families of inputs. In the batch model, instead of aiming for universal accuracy across all graph types, algorithms are designed to provide privacy for %
all possible graphs while providing accurate estimates for a select subset of ``nice" graphs---for example, graphs that satisfy a degree bound---on which the statistic of interest is well behaved. 
There are now several techniques for achieving this type of guarantee in the batch model, notably projections~\cite{BlockiBDS13,KasiviswanathanNRS13,DayLL16} and Lipschitz extensions \cite{BlockiBDS13,KasiviswanathanNRS13,ChenZ13,RaskhodnikovaS16,RaskhodnikovaS16-E,DayLL16,BorgsCSZ18,CummingsD20,KalemajRST23}. Broadly, these techniques start from an algorithm which is both private and accurate %
when restricted to a set of ``nice" graphs, and find a new algorithm that mimics the base algorithm on the set of ``nice" graphs while providing privacy for all possible graphs; such extensions exist under very general conditions~\cite{BorgsCSZ18ext,BorgsCSZ18}.

The continual release setting complicates these approaches, causing tools for the batch setting to break down.
Projections and Lipschitz extensions are harder to design: in the batch setting, the decision to remove an edge may propagate only across nodes; in the continual release setting, the change can also propagate through time, rendering existing batch-model solutions ineffective. 
Even the general existential result of \cite{BorgsCSZ18ext} applies, at best,  only to the offline version of continual release (in which the algorithm can inspect the entire input stream before starting to produce output). 
In a nutshell, ensuring low sensitivity separately at each point in time does not guarantee the type of stability that is needed to get low error with continual release (e.g., $\ell_1$ stability of difference vectors~\cite{SLMVC18}).
For this reason, the only straightforward way to get node-private algorithms from existing batch model work is to use advanced composition \cite{DworkRV10} and compose over $T$ time steps.
This potentially explains why prior works on node-private continual release of graph statistics assume restrictions on the input graphs to their private algorithms, providing privacy only in the case where the restriction is satisfied.

\subsection{Our Results}
\label{sec:our results}

We consider an insertion-only model of graph streams, where an arbitrary subset of new nodes and edges arrives at each of $T$ time steps (\Cref{defn:graph stream}). We do not assume any relationship between the size of the graph and $T$. The degree of a node $u$ in the stream is the total number of edges adjacent to $u$ in the stream (equivalently, in the final graph). The stream's maximum degree is the largest of its nodes' degrees; if this maximum is at most $D$, we say the stream is \emph{$D$-bounded}.

\myparagraph{Truly Node-private Algorithms.}
For several fundamental graph statistics, we obtain (the first) algorithms that satisfy the usual notion of node-differential privacy in the continual release setting---that is, they require no assumption on the input streams.

In contrast, %
previous work on node-private continual release of graph statistics~\cite{SLMVC18,FHO21} develops
algorithms for basic graph statistics
with strong (sometimes near-optimal) accuracy bounds
but that are only guaranteed to be private when the input graph stream has maximum degree at most a user-specified bound $D$. Unfortunately, their algorithms
can exhibit blatant privacy violations if the input graph stream 
violates the bound.\footnote{For example, suppose edges in the graph denote transmissions of a stigmatized disease like HIV and suppose the analyst knows that all the edges associated with one individual, Bob, arrive at a given time step $t$ (and only those edges arrive). The outputs of the \cite{FHO21} edge-counting algorithm at times $t-1$ and $t$ would together reveal how many disease transmissions Bob is involved in, up to error  $D/\eps$, for  privacy parameter $\eps$. When Bob's degree is much larger than $D$,
this is a clear violation of node privacy.
One can argue using this example that the algorithms of \cite{SLMVC18,FHO21} do not satisfy $(\eps,\delta)$-node differential privacy (\Cref{defn:dp crt}) for any finite $\eps$ with $\delta<1$.
}
This conditional node-privacy is not standard in the literature: prior work on node-privacy in the batch model gives unconditional privacy guarantees (e.g., \cite{BlockiBDS13,KasiviswanathanNRS13,ChenZ13,RaskhodnikovaS16,RaskhodnikovaS16-E,DayLL16,BorgsCSZ18,CummingsD20,KalemajRST23}).
To emphasize the conditional nature of the privacy guarantees for \cite{SLMVC18,FHO21}, 
we say they satisfy \emph{$D$-restricted node-DP}. 
Similarly, algorithms whose edge-privacy depends on such an assumption satisfy \emph{$D$-restricted edge-DP}.

In \cref{sec:related} we describe some ways to modify the algorithms of prior work to offer (unconditional) node privacy in the continual release model. However, these modified algorithms have error guarantees that scale polynomially with either the time horizon $T$ (for the batch-model algorithms of, e.g., \cite{BlockiBDS13,ChenZ13,KasiviswanathanNRS13,DayLL16,KalemajRST23}) or the size of the graph (for the $D$-restricted algorithms of, e.g., \cite{SLMVC18,FHO21}).

Our main contribution is a black-box transformation that can take any $D$-restricted-DP ``base'' algorithm 
and transform it into an algorithm that is private on all graphs while  maintaining the original algorithm's accuracy on $D$-bounded graphs.
We use this transformation to produce specific node-private algorithms for estimating several fundamental graph statistics 
and (in all but one case) show that the error incurred by these private algorithms is near optimal. 
We generally use algorithms of Fichtenberger et al.~\cite{FHO21} as our base, though for connected components the restricted-DP algorithm is new. Importantly for real-world applications, our algorithms are efficient: they add only linear overhead to the time and space complexity of the base algorithm.

 \cref{table:tpdp} summarizes 
the bounds we obtain on additive error---worst-case over $D$-bounded graph streams of length $T$---for releasing
the counts of edges ($\edges$), triangles ($\triangles$),
$k$-stars\footnote{A $k$-star is a set of $k$ nodes, each with an edge to a single common neighbor (which can be thought of as the $k$-star's center).} ($\kstars$),
and connected components ($\connectedcomps$), as well as degree histograms
($\deghist$); it also compares with previous results. 
The parameters $\eps,\delta$ specify the privacy guarantee (\Cref{defn:dp crt}).

\newcommand{\Nsensitivity}{\ensuremath{\Delta_{\msf{N},D}}}
\newcommand{\Esensitivity}{\ensuremath{\Delta_{\msf{E}}}}
\newcommand{\epsDguarantee}{$(\eps,\,$\small{$0,D$}$)$\small{-node-restricted-DP}}
\newcommand{\darktint}{\cellcolor[HTML]{EFEFEF}}
\newcommand{\headtint}{\cellcolor[HTML]{E0E0E0}}
\newcommand{\notint}{\cellcolor[HTML]{FFEDED}}
\newcommand{\darknotint}{\cellcolor[HTML]{EFDDDD}}
\newcommand{\yestint}{\cellcolor[HTML]{EDFFEC}}
\newcommand{\darkyestint}{\cellcolor[HTML]{DEEFDC}}
\renewcommand{\arraystretch}{1.3} %

\begin{table*}[t]
\begin{center}
\begin{tabular}{l||l|l l c|}
\hhline{~||-|- - -|}
& \headtint \begin{tabular}[c]{@{}l@{}} Lower bounds for \\ $\eps \geq \frac{\log T}{T}, \del = O\paren{\frac{1}{T}}$ \end{tabular}
& \headtint Reference & \headtint 
    \begin{tabular}[c]{@{}l@{}}
        Additive $\linf$ error for \\
        $\eps \leq 1$, 
        $\del = \Omega\paren{\frac{1}{\poly(T)}}$
    \end{tabular} 
    &\headtint 
    \begin{tabular}[c]{@{}l@{}}Node-DP\\guarantee\end{tabular} 
\\ \hline \hline
\multicolumn{1}{|l||}{}
      &    &    
   \textsc{bbds/cz/knrs}  &   $\Otilde(D\sqrt{T}/\eps)^*$  & \yestint $(\eps,\del)$ \\ 
\multicolumn{1}{|l||}{}
   &    &    
   \cite{FHO21}  &   $O(D\,\;\log^{5/2} T /\eps)$  & \notint \small{$D$-restricted} $(\eps,0)$ \\ 
\multicolumn{1}{|l||}{\multirow{-3}{*}{$\edges$}}
    & \multirow{-3}{*}{$\Omega(D\,\;\log T/\eps)$} & 
   Our work  &   $O(D\,\;\log^{5/2} T /\eps \;+$ \small{$\log^{7/2}T/\eps^2$}$)$    &  \yestint $(\eps,\delta)$ \\ \hline
\rowcolor[HTML]{EFEFEF}
\multicolumn{1}{|l||}{}
    &    &     
    \textsc{bbds/cz/knrs} & $\Otilde(D^2\sqrt{T}/\eps)^*$    &  \darkyestint $(\eps,\del)$\\ 
\rowcolor[HTML]{EFEFEF}
\multicolumn{1}{|l||}{}
    &    &     
    \cite{FHO21} & $O(D^{ 2}\log^{5/2} T /\eps)$    &  \darknotint \small{$D$-restricted} $(\eps,0)$\\ 
\rowcolor[HTML]{EFEFEF}
\multicolumn{1}{|l||}{\multirow{-3}{*}{$\triangles$}}
    & \multirow{-3}{*}{$\Omega(D^{ 2}\log T/\eps)$} & 
    Our work &  $O(D^{ 2}\log^{5/2} T /\eps \;+$ \small{$\log^{9/2}T / \eps^3$}$)$    & \darkyestint   $(\eps,\delta)$\\ \hline 
\multicolumn{1}{|l||}{}
    &    &     
    \textsc{bbds/cz/knrs} & $\Otilde(D^{ k}\sqrt{T} /\eps)^*$    & \yestint $(\eps,\del)$\\ 
\multicolumn{1}{|l||}{}
    &    &     
    \cite{FHO21} & $O(D^{ k}\log^{5/2} T /\eps)$    & \notint \small{$D$-restricted} $(\eps,0)$\\ 
\multicolumn{1}{|l||}{\multirow{-3}{*}{$\kstars$}}
    & \multirow{-3}{*}{$\Omega(D^{ k}\log T/\eps)$} & 
    Our work &  $O(D^{ k}\log^{5/2} T /\eps \;+$ \small{$\log^{k+5/2}T / \eps^{k+1}$}$)$    &   \yestint $(\eps,\delta)$ \\ \hline 
\rowcolor[HTML]{EFEFEF}
\multicolumn{1}{|l||}{}
    &    &     
    \cite{DayLL16} & $\widetilde{O}(D^{ 2}\sqrt{T} /\eps)^*$    &  \darkyestint $(\eps,\del)$ \\
\rowcolor[HTML]{EFEFEF}
\multicolumn{1}{|l||}{}
    &    &     
    \cite{FHO21} & $\widetilde{O}(D^{ 2}\log^{5/2} T /\eps)$    &  \darknotint \small{$D$-restricted} $(\eps,0)$ \\
\rowcolor[HTML]{EFEFEF}
\multicolumn{1}{|l||}{\multirow{-3}{*}{$\deghist$}}
    & \multirow{-3}{*}{$\Omega(D\,\;\log T/\eps)$} & 
    Our work &  $\widetilde{O}(D^{ 2}\log^{5/2} T /\eps \;+$ \small{$\log^{9/2}T / \eps^{3}$}$)$    & \darkyestint   $(\eps,\delta)$\\ \hline 
\multicolumn{1}{|l||}{}
    &    &     
    \cite{KalemajRST23} & $\Otilde(D \sqrt{T} /\eps)^*$    &  \yestint $(\eps,\del)$ \\
\multicolumn{1}{|l||}{\multirow{-2}{*}{$\connectedcomps$}}
    &  \multirow{-2}{*}{ $\Omega(D\,\;\log T/\eps)$}  &   
    Our work   &  $O(D\,\;\log^{5/2} T /\eps \;+$ \small{$  \log^{7/2}T/\eps^2$}$)$   &  \yestint $(\eps,\delta)$\\ \hline 
\end{tabular}
\end{center}
\caption{Accuracy of our node-private algorithms, previously known \emph{restricted} node-private algorithms, and node-private batch model algorithms on insertion-only, $D$-bounded graph streams of length $T$. ``\textsc{bbds/cz/knrs}'' refers to \cite{BlockiBDS13,ChenZ13,KasiviswanathanNRS13}; the error bounds with $^*$ were computed by applying advanced composition \cite{DworkRV10} to batch model algorithms.
Error lower bounds for these problems are for sufficiently large $T$.
\noakland{See \cref{sec:optimal algs} for more detailed bounds.}}
\label{table:tpdp}
\end{table*}

\myparagraph{Stable, Time-aware Projections.}
Central to our approach is the design of \emph{time-aware projection} algorithms that take as input an arbitrary graph stream and produce a new graph stream, in real time, that satisfies a user-specified degree bound $D$. ``Time-aware'' here refers to the fact that the projection acts on a stream, as opposed to a single graph; we drop this term when the context is clear. ``Projection'' comes from the additional requirement that the output stream be identical to the input on every prefix of the input that is $D$-bounded.
Ideally, we would  simply run 
the restricted-DP algorithm (e.g., from \cite{FHO21}) on the projected graph stream and thus preserve
the original algorithm's accuracy on $D$-bounded
streams.

The challenge is that the resulting process is only private if the projection algorithm is \textit{stable}, meaning that neighboring input streams map to nearby projected streams. 
Specifically, the \emph{node distance} between two streams $\gs$ and $\gs'$ is the minimum number of nodes that must be added to and/or removed from $\gs$ to obtain $\gs'$. \emph{Edge distance} is defined similarly (\Cref{defn:node dist}). \emph{Node-neighboring} streams are at node distance 1.
The \emph{node-to-node stability} of a projection is the largest node distance between the projections of any two node-neighboring streams; the \emph{node-to-edge stability} is the largest edge distance among such pairs.

If we had a projection with good (that is, low) node-to-node stability, then 
running the restricted-DP algorithm on the projected graph stream would satisfy node privacy, and we would be done. Alas, we do not know if such a projection exists. (We show that such a transformation does exist for edge-DP---see the end of this section.) 
Instead, we give two simple, greedy projection algorithms that have good node-to-node and node-to-edge stability
when the input graph stream satisfies a privately testable ``safety'' condition. The safety condition is that the stream has few  large-degree vertices (\Cref{defn:d ell bdd}). Specifically, a graph (or stream) is $(D,\ell)$-bounded if it has at most $\ell$ nodes of degree larger than $D$.\footnote{$(D,\ell)$-boundedness is a computationally efficient proxy for requiring that the stream be close in node-distance to a $D$-bounded stream. Testing the latter condition directly is NP-hard (by reduction from vertex cover); we instead efficiently compute the distance to the nearest \textit{not} $(D,\ell)$-bounded stream---see \Cref{sec:always-private}.}

We obtain a general transformations from $D$-restricted-DP algorithms to truly private ones by testing the safety condition using a novel online variant of the Propose-Test-Release framework of \cite{DL09}.

We explore two natural methods for time-aware projection, each based on a batch-model projection algorithm from the literature.
Both time-aware projections
greedily add edges while maintaining an upper bound on each node's degree. 
One of these methods bases its greedy choices on the degree of the nodes in the original graph stream (``BBDS'', \cite{BlockiBDS13}), while the other bases its choices on the degree of the nodes in the projection that it produces~(``DLL'', \cite{DayLL16}). The results in \Cref{table:tpdp} are obtained using the BBDS-based projection and our general transformation.

We give tight bounds on three measures of stability, summarized  in \cref{tab:projection-stabilities}. The table lists upper bounds; the lower bounds for the BBDS projection are identical up to small additive constants (and the edge-to-edge stability is identical), while the bounds for DLL are tight up to a constant multiplicative factor\noakland{ (see \cref{sec:stab-tight})}.
A graph in the batch model can be represented as a length-$1$ graph stream, so these projections' stability properties also hold for graphs in the batch model.

\begin{table}[t]
\begin{center}
\begin{tabular}{l|c|c|}
\hhline{~|-|-|}
\multicolumn{1}{c|}{ }              &\cellcolor[HTML]{E0E0E0}   $\projo$    & \cellcolor[HTML]{E0E0E0} $\projp$    \\ \hline
\multicolumn{1}{|l|}{edge-to-edge} & 3         & $2\ell+1$ \\ \hline
\rowcolor[HTML]{EFEFEF} 
\multicolumn{1}{|l|}{node-to-edge} & $D+\ell$ & $D + 2\ell\sqrt{\min\{D,\ell \} }$ \\ \hline
\multicolumn{1}{|l|}{ node-to-node} & $2\ell+1$ & $2\ell+1$ \\ \hline
\end{tabular}%
\end{center}
\caption{Stability  of 
$\projo$ and $\projp$ on $(D,\ell)$-bounded input graph streams,  %
from \cref{thrm:combined-stab}.}
\label{tab:projection-stabilities}
\end{table}

The DLL projection preserves more edges than the BBDS projection when the input has some high-degree vertices (the graph returned by BBDS is a subgraph of that returned by DLL), which initially suggests that the DLL projection could be more useful.
Indeed, in the batch setting, the authors of~\cite{DayLL16} show that the projected degree distribution (and number of edges) has low sensitivity. This allows for the DLL projection to provide a better privacy-utility trade-off for these tasks in the batch model.
However, this projection actually has worse stability when we measure node- or edge-distance between output graphs.\footnote{This distinction is crucial in the continual-release setting. For example, even though the degree distribution of the DLL projection has node sensitivity $O(D)$ in the batch setting, the sequence of degree distributions one gets when projecting a stream has unbounded sensitivity.}
Therefore, more noise must be added when using the DLL projection for generic applications, as compared to the BBDS projection.
In our uses, this ultimately means that the projection of~\cite{BlockiBDS13} provides the better privacy-utility trade-off.

\myparagraph{Truly Edge-private Algorithms.} 
Although our focus is on node-privacy, we show along the way that the BBDS-based time-aware projection has edge-sensitivity 3, uniformly over all graphs. (This follows from a batch-model argument of \cite{BlockiBDS13} and a general ``Flattening Lemma'' (\cref{lem:greedy-flatten}) that we establish for greedy, time-aware projections.) As a result, one can make $D$-restricted edge-private algorithms into truly private ones at almost no cost in accuracy. 
Some consequences are summarized in \Cref{thrm:acc of applications edge}.

\myparagraph{Experiments.} We provide experiments on synthetic graphs, which show that our transformation adds little run time overhead and results in truly node-private algorithms that improve considerably over the batch-model baseline.

\subsection{Techniques}
\label{sec:techniques}

\oldtext{Old Section 1.2 here: 

The first key idea in our work is in the definition and analysis of time-aware variants of the batch model projection algorithms from ~\cite{DayLL16} and~\cite{BlockiBDS13} respectively. Even in the batch setting~\cite{BlockiBDS13} only consider edge-to-edge stability while~\cite{DayLL16} analyze only the stability of particular functions on the projected graphs. Our work, on the other hand, considers robust stability guarantees of the overall projection in the continual release model. This allows us to obtain a general transformation in the continual release setting from restricted-DP algorithms to those that provide unconditional privacy. We showcase the wide applicability of this transformation by using it to get nearly optimal algorithms for several problems.

The second key idea is in the identification and use of an efficiently and privately testable ``safety" codition ($(D,\ell)$-boundedness) of a graph stream under which the time-aware projections have robust stability properties. In particular, since this condition corresponds to a low-sensitivity property on node-neighboring graph streams, we are able to privately test this value using an online variant of the \emph{Propose-Test-Release (PTR)} framework of \cite{DL09}. Notice that it is not apriori obvious that the privacy properties of PTR hold in the continual release setting (since we cannot simply reduce the privacy analysis of PTR in the continual release model to that in the batch model.)

Finally, the main technical innovation in our work is in the proof of \cref{thrm:combined-stab} which presents the robust stability guarantees of the time aware projections we consider in this work. There are a few key ideas in the proofs of these results. (1) We observe\cref{lem:greedy-flatten} that due to the greedy nature of the projections, the stability analysis need not carefully analyse the projected stream items over time and can instead restrict its attention to the overall graphs corresponding to the projected stream. }

\myparagraph{Stability Analyses.}
The main technical contribution lies in defining time-aware versions of the two greedy projection algorithms (\Cref{alg:time aware projection}), and leveraging that structure to analyze the sensitivity of the entire projected graph sequence (\Cref{thrm:combined-stab}). Our analyses differ substantially from existing batch-model analyses, both because of the sequential nature of our problem and the stronger notions of stability we consider.

\Cref{sec:stability-thm-overview} contains a detailed overview of the arguments; we highlight here a few simple but useful ideas. 
The time-aware projections share two key features: 
\begin{itemize}
    \item \textit{Shortsightedness:} the algorithm includes all nodes and makes a final decision about each edge at the time it arrives; 
    \item \textit{Opportunism:} if an edge connects vertices with degree at most $D$ in the original graph stream, it will necessarily be included in the projection.
\end{itemize}

Such greedy structure is computationally convenient but also helps us analyze stability. To see why, consider two graph streams $\gs,\gs'$ that differ in the presence of a node $\vplus$ and its edges, and let $\Pi_D(\gs)$ and $\Pi_D(\gs')$ denote the projected streams (where $\Pi_D$ could be either of our two projections). We consider for each time step $t$ the \emph{difference graph} $\gdiff_t$ consisting of edges that  have been added (at or before time $t$) to one projected stream but not the other. 

The first feature, shortsightedness, implies that this difference graph \emph{grows monotonically}. (Such a statement need not hold for arbitrary projections.) This allows us to show a ``Flattening Lemma'' (\Cref{lem:greedy-flatten}), which states that the edge- and node-distance between $\Pi_D(\gs)$ and $\Pi_D(\gs')$ depend only on the final difference graph $\gdiff_T$ (or an intermediate graph $\gdiff_t$ in the case that we are considering only a prefix of the streams). Thus, shortsightedness allows us to ignore the sequential structure and reduce to a batch-model version of $\Pi_D$ in which arrival times affect only the order in which edges are greedily considered.

The second feature, opportunism, allows us to take advantage of $(D,\ell)$-boundedness. If the larger stream has at most $\ell$ vertices of degree more than $D$, we can show that $\gdiff_t$ will have a vertex cover of size at most $\ell+1$ (\Cref{lem:edges dont change}). 

These structural results suffice to bound the node-to-node stability of both projections by $2\ell+1$.

From this point, the analyses of the two projections diverge. Each inclusion rule leads to different structure in the difference graph $\gdiff_T$.
The most involved of these analyses proves a (tight) bound of $D+2\ell\sqrt{\min\set{D,\ell}}$ on the node-to-edge stability of the DLL-based projection. At a high level, that analysis proceeds by orienting the edges of $\gdiff_T$ to show that it is close in edge distance to a large DAG which is covered by at most $\ell$ edge-disjoint paths and then bounding the possible size of such a DAG.

The node-to-edge analysis of the BBDS-based algorithm is also subtle, but different. The key point there is that all of the edges connected to $\vplus$ can potentially cause changes in the projected graph, even the edges which are not selected for inclusion in the projection themselves. We refer to \Cref{sec:BBDS node edge} for further detail.

\myparagraph{Testing Distance to Unsafe Streams.}
A second, less involved insight is that,
although it is NP-hard to compute the node distance to the nearest stream that is not $D$-bounded,
$(D,\ell)$-boundedness gives us a proxy that is much easier to work with. Specifically, we observe that $D$-bounded streams are always distance at least $\ell+1$ from the nearest non-$(D+ \ell,\ell)$-bounded stream; furthermore, this distance can be computed in linear time~(\Cref{lem:dtg-properties}). 
Since the distance to non-$(D+\ell, \ell)$-boundedness at any given time step has low node-sensitivity, we can use a novel (to our knowledge) online variant of the PTR framework~\cite{DL09} based on the sparse vector technique~\cite{DNRRV09,RothR10,HardtR10} to monitor the distance and stop releasing outputs when the distance becomes too small.
The privacy analysis of this part follows the argument of \cite{DL09} but differs
because, rather than making a binary decision to either release or not release an output, the testing process dynamically chooses to release outputs at up to $T$ time steps (see \cref{thrm:ptr crt simple}).
The resulting general transformations are summarized in \cref{thrm:bb priv params}.

\subsection{Related Work}
\label{sec:related}

Our contributions draw most heavily from the literature on batch-model node-differentially private algorithms. Node privacy was first formulated by \cite{HLMJ09}. The first nontrivial node-private algorithms emerged in three concurrent works~\cite{BlockiBDS13,KasiviswanathanNRS13,ChenZ13} that collectively identified two major families of (overlapping) approaches based on Lipschitz extensions \cite{BlockiBDS13,KasiviswanathanNRS13,ChenZ13,RaskhodnikovaS16-E,RaskhodnikovaS16,DayLL16,BorgsCSZ18,CummingsD20,KalemajRST23} on one hand, and projections~\cite{BlockiBDS13,KasiviswanathanNRS13,DayLL16} on the other.  These works provide algorithms with (tight) accuracy guarantees for $D$-bounded graphs for the statistics we consider here as well as families that arise in the estimation of stochastic block models and graph neural networks. They also consider other families of graphs on which their specific statistics are well behaved. Most relevant here is the batch-model projection of BBDS \cite{BlockiBDS13} with low edge-to-edge sensitivity, and the Lipschitz extension for degree distributions of DLL \cite{DayLL16}. This latter extension can be viewed algorithmically as a greedy projection for which the degree histogram is stable. We use their projection idea, and then analyze the stability of the graph as a whole. 

There is also an extensive literature on batch-model edge-private algorithms; we do not attempt to survey it here.

A second major tool we draw on is the $D$-restricted node- and edge-private algorithms of \cite{SLMVC18,FHO21} for continual release of graph statistics. These in turn use the widely-studied tree mechanism, whose use in the continual-release setting (for numerical data) dates back to the model's introduction~\cite{DNPR10,CSS11}. Also relevant are the edge-private streaming algorithms of \cite{Upadhyay13,UpadhyayUA21} for cuts and spectral clustering. (To the best of our understanding, the application of our transformations to their algorithms does not yield non-trivial utility guarantees.)

Finally, our work draws on the Propose-Test-Release framework of \cite{DL09}, combining it with the sparse vector mechanism \cite{RothR10,HardtR10,LSL17} to monitor the distance from the stream to the nearest non-$(D,\ell)$-bounded stream.

\myparagraph{Modifying existing algorithms.} By modifying algorithms from prior work, we can obtain algorithms with unconditional privacy in the continual observation model. However, the error of these modified algorithms scales polynomially with the time horizon $T$ (when starting from batch-model algorithms) or the size of the graph (when starting from $D$-restricted algorithms). This means the relative error of the resulting algorithms is large when run, respectively, for a large number of time steps or on sparse graphs.

First, any batch-model algorithm with unconditional $(\eps,\delta)$-node-DP (e.g., \cite{BlockiBDS13,ChenZ13,KasiviswanathanNRS13,DayLL16,KalemajRST23}) can be extended to the continual-release model in the following way: fix some time horizon $T$, run the algorithm with (roughly) $\eps \approx \eps'/\sqrt{T}$ and $\delta \approx \delta'/T$ at each time step. 
By composition \cite{DworkRV10} over the $T$ time steps, the resulting algorithm is $(\eps',\del')$-DP. The error guarantee of the modified algorithm scales linearly with $\sqrt{T}$ in the typical case that the error of the batch algorithm scales as $1/\eps$.

Second, the $D$-restricted node-DP algorithms of \cite{SLMVC18,FHO21} can be modified to offer privacy for all graphs as follows: given a public estimate $\tilde{n}$ for the number of nodes in the graph stream, one can run a modified form of the original algorithm with $D =\tilde{n}$ that ignores any nodes beyond the first $\tilde{n}$ nodes in the graph stream.\footnote{Restricting the input graph stream to the first $\tilde n$ nodes that arrive will amplify node distances by at most a factor of 2; one can adjust for this by roughly doubling $\eps$ and $\delta$.}
If an estimate $\tilde{n}$ is not known ahead of time, the algorithm can be further adapted based on a standard trick:
Part of the privacy budget can be used to separately maintain a differentially private estimate $\tilde{N}$ of the number of nodes in the graph. We can initialize the bound $\tilde n$ to a default value and then double it and restart the main algorithm whenever $\tilde N$ gets sufficiently large. This process will increase the privacy budget by
the logarithm of the final size of the graph.
The resulting algorithm is truly node-DP, but incurs large relative error on sparse graphs since it uses a degree bound $D$ that is close to the actual size of the graph.

The approach we present in this paper achieves error bounds that are independent of the number of nodes in the stream, and scale only poly-logarithmically with $T$.

\subsection{Organization of This Manuscript}
\label{sec:org}

\Cref{sec:prelims} lays out the model and basic definitions used in the remainder of the paper. 
\Cref{sec:projection} presents the time-aware projections and the results on their stability.
\Cref{sec:always-private} explains the general black-box transformation from $D$-restricted edge (or node) privacy to true node privacy. 
\Cref{sec:opt algs} develops the applications to basic graph statistics. 
\Cref{sec:experiments} presents our experimental results.
\yoakland{Because of space constraints, many proofs are deferred to the full version, which can be found at \cite{JainSW24}.}

\section{Preliminaries}\label{sec:prelims}

\begin{definition}[Graph]
    A \emph{graph} $G = (V,E)$ consists of a set of vertices $V$ (also known as nodes), and a set of edges $E$, where edge $\{v_1, v_2\}\in E$ if and only if there is an edge between nodes $v_1\in V$ and $v_2 \in V$.
\end{definition}

\noakland{We use $V(G)$ and $E(G)$ to denote the vertex set and edge set of graph $G$, respectively. We use $\deg_v(G)$ to denote the degree of a node $v\in V(G)$; we drop $G$ when the graph being referenced is clear.
}

\begin{definition}[Graph stream]
\label{defn:graph stream}
    Given a time horizon $T$, a \emph{graph stream} $\gs\in \mGS^T$ is a $T$-element vector, where each element of the vector contains some set of nodes and edges, or the symbol $\bot$ if no nodes or edges arrive in that time step.
    At each time $t\in[T]$, either $\bot$ arrives or some set of nodes and edges arrives. By convention, an edge's endpoints arrive no later than the edge.
\end{definition}

We denote by $\gs_t$ the set of added nodes and edges which arrive in time step $t$.
We use $\gs_{[t]}$ to denote the sequence $\gs_1,\ldots, \gs_{t}$. 

\begin{definition}[Flattened graph]
\label{defn:flattened graph}
    Let $\gs\in\mGS^T$ be a graph stream of length $T$.
    The \emph{flattened graph} of the first $t$ terms $S_{[t]}$ of a graph stream, denoted $\flatt{\gs_{[t]}}$, is the graph that can be formed by all of the nodes and edges which arrive at or before time $t$.
\end{definition}

When the meaning is clear, we may refer to the graph stream through time $t$ when stating a property of the flattened graph of the graph stream through time $t$. For example, when we say that $\gs_{[t]}$ has maximum degree at most $D$, we mean that $\flatt{\gs_{[t]}}$ has maximum degree at most $D$.

We next define neighboring graphs and graph streams. In privacy-preserving data analysis, the notion of neighboring datasets is important since privacy requires that evaluating a function on similar datasets produces indistinguishable outputs. There are two natural notions of neighboring graphs and graph streams: node neighbors differ on a node (and its associated edges), while edge neighbors differ on one edge. We denote node and edge neighbors with the relations $\nnode$ and $\nedge$ respectively.

\begin{definition}[Neighboring graph streams]
\label{defn:nn graph streams}
    Two graphs (respectively, graph streams) are {\em node-neighbors} if one can be obtained from the other by removing a vertex and all of its adjacent edges. (For graph streams, the adjacent edges for the removed node may have been spread over many time steps.)

    Similarly, two graphs (respectively, graphs streams) are \emph{edge neighbors} if one can be obtained from the other by either removing one edge, removing an isolated node, or removing a node of degree 1 and its adjacent edge.\footnote{Another way to define edge neighbors would be to take the set of nodes as fixed and public, and only consider changes to one edge. We adopt the more general definition since it simplifies our results on node-to-edge stability.} 
\end{definition}

A generalization of node- and edge-neighboring graphs and graph streams is the notion of node and edge distance between graphs and graph streams. We note that node- and edge-neighboring datasets are at node and edge distance $1$, respectively.

\begin{definition}[Node and edge distance]
\label{defn:node dist}
    The {\em node-distance} $\dnode(G, G')$ is defined as the length $d$ of the shortest chain of graphs (respectively, graph streams) $G_0,G_1,\ldots,G_d$ where $G_0=G$, $G_d=G'$, and every adjacent pair in the sequence is node neighboring. 

    The {\em edge-distance} $\dedge(G, G')$ is defined as the length $d$ of the shortest chain of graphs (respectively, graph streams) $G_0,G_1,\ldots,G_d$ where $G_0=G$, $G_d=G'$, and every adjacent pair in the sequence is edge neighboring. 
\end{definition}

Given two  graphs with no isolated vertices $G=(V,E)$ and $G'=(V',E')$ (where $V$ and $V'$ may overlap), the edge distance between $G$ and $G'$ is exactly the size of the set $E \triangle E'$. (Isolated vertices that are not in both graphs add to the distance.)

\myparagraph{Differential Privacy in the Batch Model.}
To define differential privacy in the batch model, we introduce the notion of $(\eps, \delta)$-indistinguishability.

\begin{definition}[$(\eps,\delta)$-indistinguishability]
We say that two random variables $R_1,R_2$ over outcome space $\mY$ are \emph{$(\eps,\delta)$-indistinguishable} (denoted $R_1\simed R_2$) if for all $Y\subseteq \mY$, we have
\begin{align*}
\pr{R_1\in Y} &\leq e^{\eps} \pr{R_2 \in Y} + \delta; \\
\pr{R_2\in Y} &\leq e^{\eps} \pr{R_1 \in Y} + \delta.
\end{align*}
\end{definition}

Informally, a function is differentially private if applying the function to inputs which differ in the data of one individual results in outputs from similar distributions---more specifically, from distributions which are $(\eps, \delta)$-indistinguishable. We use the definitions of node and edge neighbors presented above to formalize the notion of what it means for graphs or graph streams to differ in the data of one individual.

\begin{definition}[Differential privacy (DP) in the batch model \cite{DworkMNS16}]
\label{defn:dp batch}
A randomized algorithm $\mM:\mGS^T\to \mY$ is $(\eps,\del)$-\emph{node-DP} (respectively, \emph{edge-DP}), if for all pairs of node-neighboring (respectively, edge-neighboring) graph streams $\gs$ and $\gs'$, the distributions $\mM(\gs)$ and $\mM(\gs')$  are $(\eps,\del)$-indistinguishable:
\[
\mM(\gs)\simed\mM(\gs').
\]

The term \emph{pure DP} refers to the case where $\delta = 0$, and \emph{approximate DP} refers to the case where $\delta > 0$.
\end{definition}

\myparagraph{Privacy under Continual Observation.}
We now define privacy of graph statistics under continual observation; the general definition is borrowed from \cite{JRSS21}. In the continual release setting, first explored by \cite{CSS11,DNPR10}, an algorithm receives a stream of inputs $S = (S_1,\ldots, S_T) \in \mS^T$. The definition of privacy requires indistinguishability on the distribution of the entire sequence, not just one output.
For simplicity, in this version, we consider only the simpler, non-adaptive concept of differential privacy. We conjecture that all our algorithms and results extend verbatim to the adaptive version~\cite{JRSS21} (since the main components of our algorithm, the tree mechanism and sparse vector technique, are known to be adaptively private).

\begin{definition}[Privacy of a mechanism under continual observation]
\label{defn:dp crt}
    Define $\mA_{\mM}$ as the batch-model algorithm that receives a dataset $x$ as input, runs $\mM$ on stream $x$, and returns the output stream $y$ of $\mM$. We say that $\mM$ is \emph{$(\eps,\del)$-DP in the non-adaptive setting under continual observation} if $\mA_{\mM}$ is $(\eps,\del)$-DP in the batch model.
\end{definition}

We borrow the definition of accuracy from \cite{JRSS21}, which bounds the error of a mechanism with respect to a target function $f$. The definition takes the maximum error over both time steps and the coordinates of the output of $f$. (Although most of the functions we approximate return a single real value, the degree histogram $\deghist$ returns a vector at each step and requires the extra generality.)

\begin{definition}[Accuracy of a mechanism]
\label{defn:accuracy crm}
Given a set of allowable streams $\mS \subseteq \mX^*$, a mechanism $\mM$ is \emph{$(\alpha, T)$-accurate with respect to $\mS$} for a function $f:\mX^*\to \R^k$ if, for all fixed (i.e., non-adaptively chosen) input streams $S = (S_1,\ldots, S_T) \in \mS$, the maximum $\linf$ error over the outputs $a_1, \ldots, a_T$ of mechanism $\mM$ is bounded by $\alpha$ with probability at least $0.99$; that is,
\[
\Pr_{\text{coins of $\mM$}} \left[ \max_{t\in[T]}  \left\Vert f\left(\gs_{[t]}\right) - a_t \right\Vert_\infty \leq \alpha \right] \geq 0.99.
\]
\end{definition}

If the indistinguishability property of differential privacy holds conditioned on the promise that both node-neighbors (respectively, edge-neighbors) lie in the set of graph streams with maximum degree at most $D$, we say that the algorithm offers \emph{$D$-restricted $(\eps,\del)$-node-DP} (respectively, \emph{$D$-restricted $(\eps,\del)$-edge-DP}). This is the notion of privacy explored by all prior work on node-private graph statistics under continual observation \cite{FHO21,SLMVC18}. 

\begin{definition}[$D$-restricted DP]
A randomized algorithm $\mM:\mG\to \mY$ is \emph{$D$-restricted $(\eps,\del)$-node-DP} (respectively, \emph{edge-DP}) if \cref{defn:dp batch} holds when restricted to the set of node-neighboring (respectively, edge-neighboring) graph streams with maximum degree at most $D$.

Likewise, $\mM$ is \emph{$D$-restricted $(\eps,\del)$-node-DP under continual observation} (respectively, \emph{edge-DP}) if \cref{defn:dp crt} holds when restricted to the set of node-neighboring (respectively, edge-neighboring) graph streams with maximum degree at most $D$.
\end{definition}

\noakland{
    \yoakland{\subsection{Properties of Differential Privacy}}
\noakland{\myparagraph{Properties of Differential Privacy.}}
The definition of differential privacy extends to groups of individuals. If an algorithm is DP for one individual, it also offers a (differently parameterized) privacy guarantee for a collection of individuals. We borrow the formulation of this property from \cite{Vadhan17}.

\begin{lemma}[DP offers group privacy \cite{DworkMNS16}]
\label{lemma:group priv}
Let $\mM:\mX\to \mY$ be a randomized algorithm that is $(\eps,\del)$-DP. Then, where $x,x'\in \mX$ differ in the data of $k$ individuals, $\mA(x)$ and $\mA(x')$ are $(k\cdot \eps, k\cdot e^{k\eps} \cdot \del)$-indistinguishable. That is,
\[
    \mA(x) \approx_{k\cdot \eps, k\cdot e^{k\eps} \cdot \del} \mA(x').
\]
\end{lemma}

This follows from a well-known ``weak triangle inequality'' for $(\eps,\delta)$-indistinguishability: 

\begin{lemma}[Weak triangle inequality]
\label{lem:weak triangle ineq}
    For all $\eps_1, \eps_2, \delta_1 ,\delta_2 \geq 0$: If random variables $A,B,C$ satisfy $A\approx_{\eps_1,\delta_1} B$ and $B \approx_{\eps_2,\delta_2} C$, then $A\approx_{\eps',\delta'} C$ for  $\eps' = \eps_1 + \eps_2$ and $\delta' = \max(\delta_1  + e^{\eps_1}\delta_2,\ \delta_2+ e^{\eps_2}\delta_1) \leq  e^{\eps_2}\delta_1+ e^{\eps_1}\delta_2$.
\end{lemma}

Differential privacy is robust to post-processing.

\begin{lemma}[DP is robust to post-processing \cite{DworkMNS16}]
\label{lemma:post proc}
Let $\mM:\mX\to \mY$ be a randomized algorithm that is $(\eps,\del)$-DP. Let $f:\mY\to \mZ$ be an arbitrary, randomized mapping. Then $f\circ \mM:\mX\to \mZ$ is $(\eps,\del)$-DP.
\end{lemma}

}

\ifnum\lowerbds=1
    \section{Graph Statistics}

\ifnum\lowerbds=1
The lower bounds that we show in \cref{sec:lower bounds} hold for insertion-only graph streams with no solo edges. Note that the lower bounds, then, also hold for the insertion-only setting and hold for the setting in which nodes and edges can be inserted and deleted, although tighter bounds may be possible in these settings.
\fi

We now define notation for releasing graph statistics.

\begin{definition}[Edges in a graph]
    Let $\edges:\mG\to \R$ return to the number of edges in the input graph $G\in\mG$.
\end{definition}

Let $n$ be the number of nodes in a graph. When working with sparse graphs, computing the true number of edges in the graph and releasing a noisy answer with the Laplace mechanism often introduces too much noise for the answer to be useful. If no degree bound is imposed, rewiring one node can result in the addition or deletion of up to $n-1$ edges, so the edge count has sensitivity $\Theta(n)$. Many graphs are sparse and do not have maximum degree $n$; for this reason, we want to be able to release edge counts with noise closer to $O(D/\eps)$, where $D$ is the maximum degree of the graph.

\cite{BlockiBDS13,KasiviswanathanNRS13,ChenZ13} introduced the first accurate, node-private algorithms for graph statistics. \cite{KasiviswanathanNRS13,BlockiBDS13},\connorq{Does \cite{ChenZ13} do projection? I looked at it for a little while, and I was confused. I can look at it more, but I thought I'd ask since you might just know.} among other more recent works like \cite{DayLL16}, present \emph{projection} algorithms which take a graph $G\in\mG$ and a degree bound $D$ as input, and essentially map $G$ to a (not necessarily induced) subgraph with degree bound $D$. Importantly, neighboring graphs are projected to similar graphs. As a result, adding noise with scale $O(D/\eps)$ preserves privacy for all graphs (regardless of maximum degree), and, when working with sparse graphs with maximum degree $\leq D\ll n$, these projection mechanisms which use noise distributions with scale $O(D/\eps)$ are much more accurate than mechanisms which add noise from a distribution with scale $O(n/\eps)$.

Roughly, these papers provide extensions of the true edge count that provide answers close to maximizing the edge count over all (not necessarily induced) subgraphs of the input graph which have degree $\leq D$. \connorq{Is there a source to support this idea? I think it's true, but idk.} For graphs of degree $\leq D$, then, the output is close to the true output of $\edges$. We provide a formal definition of this idea below.

\begin{definition}[Edge count over degree-bounded subgraphs]
\label{defn:degree bound func}
    Let $\mG_D$ denote the set of all graphs with maximum degree at most $D$. Additionally, for $G,G'\in \mG$, we write $G'\preceq G$ if $G'$ is a subgraph of $G$.

    We define the edge count over degree-bounded subgraphs $\subgraphedges_D:\mG\to \R$ as the edge count, maximized over all degree $D$ bounded subgraphs of the input graph. That is,
    \[
        \subgraphedges_D(G) = 
        \max_{
        \substack{
        G'\in \mG_D  \\
        \text{s.t. } G'\preceq G
        }
        }
        \edges(G').
    \]
\end{definition}

In addition to showing a lower bound on the error needed to privately approximate $\subgraphedges_D$, we also show a lower bound on the error needed to privately approximate the extension for the number of edges described in \cite[Alg. 1]{KasiviswanathanNRS13} and defined below in \cref{defn:KasiviswanathanNRS13 alg}.

\begin{definition}[Flow graph \cite{KasiviswanathanNRS13}]
\label{defn:flow-graph}
    Given an undirected graph $G = (V,E)$, let $V_\ell = \{v | v\in V\}$ and $V_r = \{v | v\in V\}$ be two copies of $V$, called the \emph{left} and \emph{right} copies, respectively. Let $D\in \N$ be less than $n$. The \emph{flow graph} of $G$ with parameter $D$, a source $s$, and sink $t$ is a directed graph on nodes $V_\ell \cup V_r \cup \{s,t\}$ with the following capacitated edges: edges of capacity $D$ from the source to all nodes in $V_\ell$ and from all nodes in $V_r$ to the sink $t$, and unit-capacity edges $(u_\ell, v_r)$, where $u_\ell\in V_\ell$ and $v_r\in V_r$, for all edges $\{u,v\}$ of $G$. Let $\vfl_D(G)$ denote the value of the maximimum flow in the flow graph of $G$ with parameter $D$.
\end{definition}

\begin{definition}[Flow-based edge count from \cite{KasiviswanathanNRS13}]
\label{defn:KasiviswanathanNRS13 alg}
    Let $\maxflowedges_D : \mG\to \R$ be defined as $\maxflowedges_D:G \mapsto \vfl_D(G)/2$, where $\vfl_D(G)$ denotes the value of the maximum flow in the flow graph of $G$ with parameter $D$ as defined in \cref{defn:flow-graph}.
\end{definition}

\section{Lower bounds for graph statistics}
\label{sec:lower bounds}

In this section, we use a reduction from \emph{one-way marginals} (OWMs) --- defined in \cref{sec:owms} --- to prove the following lower bound on the error needed to privately approximate $\subgraphedges_D$ and $\maxflowedges_D$. The reduction is presented in \cref{sec:reduc from owms}, and the proof of \cref{thrm:lower bound edge count} is presented in \cref{sec:pf of lower bounds}.

\begin{theorem}[Lower bounds for edge counts]
\label{thrm:lower bound edge count} 
Let $\eps \in (0,1]$, $\del\in [0,1]$, $D\in \N$, and $T\in \N$. Let $\mS$ be the set of all insertion-only graph streams with no solo edges (\cref{defn:graph stream no solo}). Let $\mM$ be a mechanism for approximating $\subgraphedges_D$ or $\maxflowedges_D$. If $\mM$ is $(\eps, \del)$-DP and $(\alpha, T)$-accurate with respect to $\mS$, the following statements hold.
\begin{enumerate}
    \item If $\del = 0$, then $\alpha = \Omega\left( \dfrac{DT}{\eps} \right)$.
    \item If $\del > 0$ and $\del = o\left( \dfrac{\eps \log T}{\sqrt{T}} \right)$, then $\alpha = \Omega\left( \dfrac{D\sqrt{T}}{\eps\log T} \right).$
\end{enumerate}
    \connorq{I feel pretty good about item (1), since it matches the bound I got through the packing argument. I'm fairly confident the math goes through for the approx-DP bound. The proof of \cref{thrm:lower bound edge count} is at the end of \cref{sec:reduc from owms}. I've identified one question about the proof of (1) and two questions about the proof of (2) that are provided in-line there. Could you take a look at those questions?}
    
    \connorqd{The corresponding theorem in \cite{JRSS21} says, ``... and sufficiently large $T\in \N$.'' I think this theorem goes through for all $T\in\N$ for me, but I'm not sure. Does it go through for all $T\in\N$, or should I also have the ``sufficiently large'' modifier?}\connor{Discussed 5/31/2023: Adam says ``it depends on what I prove''. Makes sense. Recall discussion on asymptotic bounds for expressions with several variables. $f(x,y) = x/ \log y$ is only $O(x)$ for sufficiently large $y$; for $y\in(1,2)$, the function gets arbitrarily large. By saying ``arbitrarily large'', we express the idea that there may exist a value of $y$ for which the asymptotic expression doesn't hold, but that once $y$ exceeds some constant value (which we don't need to specify), then the asymptotic bound holds.}
\end{theorem}

Note that the bounds in \cref{thrm:lower bound edge count} also hold when $\mS$ is defined as the set of all insertion-only graph streams (\cref{defn:graph stream}), or when $\mS$ is defined as the set of all graphs streams with insertions and deletions, though stronger lower bounds may also hold in these settings.

Put the main lower bound theorem(s) here. I'm thinking separate ones for maxflow-edge-count and subgraph-edge-count. We can jointly state one or both of the reduction results in the next section though if it makes the proof easier.

probably use restatable theorems so that they can be mentioned in the introduction as well as in the main body of the paper. necessary packages are in the comment below.

Mention the structure of the proof here

\subsection{One-way marginals (OWMs)}
\label{sec:owms}

Lower bounds on the error needed to privately release OWMs are well studied \cite{HT10,Vadhan17,BUV18}, and reducing from one-way marginals gives useful lower bounds for both pure and approximate DP. Before proceeding with the reduction, we provide additional information on one-way marginals.

\begin{definition}[One-way marginals]
    Let $\mX = \{0,1\}^k$. We define $\margs_k:\mX^*\to [0,1]^k$ as the function which, for each column $i\in[k]$, computes the fraction of rows in the dataset which have property $i$.
    
    Let $x_i$ denote the $\ord{i}$ row of vector $x$, and let $y[j]$ denote the $\ord{j}$ entry in $y\in\mX$. More formally, $\margs_k$ outputs the vector $(\phi_1(x),\ldots,\phi_k(x))$, where
    \begin{equation}
    \label{eq:jth marginal}
        \phi_j(x) = \frac{1}{n}\sum_{i\in[n]}x_i[j].
    \end{equation}
\end{definition}

To discuss lower bounds on accuracy, we need a notion of accuracy for the batch model. We borrow the formulation of accuracy from \cite{JRSS21}.

\begin{definition}[Accuracy in the batch model]
\label{defn:accuracy batch}
    Recall the definition of $\err$ from \cref{defn:error}, which measures the accuracy of an approximation to a function by returning the $\linf$ distance between the outputs of the function and the answers produced by the approximation. Additionally, let $\gamma \in [0,1]$ and $n\in\N$. 

    An algorithm $\mA$ is \emph{$(\gamma, n)$-accurate} for a function $f:\mX^n \to \R^k$ if, for all input datasets $x\in\mX^n$, the $\linf$ distance between $f(x)$ and the vector $a$ of outputs from $\mA$ is bounded by $\gamma$ with probability $\geq 0.99$; that is,
    \[
        \Pr_{\text{coins of $\mA$}} \left[ \err_f(x, a) \leq \gamma \right] \geq 0.99.
    \]
\end{definition}

Note that this definition applies to the batch model but is otherwise similar to \cref{defn:accuracy crm}. One-way marginals have the following lower bounds. The presentation of the lower bounds in \cref{lemma:owm lower bounds} is borrowed from \cite{JRSS21}.

\begin{lemma}[Lower bounds for one-way marginals]
\label{lemma:owm lower bounds}
For all $\eps\in(0,1]$, $\del\in [0,1]$, $\gamma\in(0,1)$, and $k\in\N$, if algorithm $\mA$ is $(\eps,\del)$-DP and $(\gamma, n)$-accurate for $\margs_k$, we have the following lower bounds.
\begin{enumerate}
    \item If $\del = 0$, then $n = \Omega\left( \dfrac{k}{\gamma \eps} \right)$ \cite{HT10}.
    \item If $\del > 0$ and $\del = o(1/n)$, then $n = \Omega\left(\dfrac{\sqrt{k}}{\gamma \eps \log k}\right)$ \cite{BUV18}.
\end{enumerate}
\end{lemma}

\subsection{Reduction from one-way marginals}
\label{sec:reduc from owms}

We now present our reduction from one-ways marginals in \cref{alg:reduction from owms}. The accuracy and privacy bounds of the reduction are provided in \cref{thrm:reduction priv acc}; the proof of \cref{thrm:reduction priv acc} is provided at the end of \cref{sec:reduc from owms}.

\begin{algorithm}[ht!]
    \caption{Algorithm $\mA_G$ for constructing the graph stream}
    \label{alg:make graph}
\connor{I'll make this nice later. For now, I only have the algorithm typed out in the ``enumerate'' environment. I plan to type it as an ``algorithmic'' once it's finalized.}

\begin{itemize}
    \item Input: a dataset $x = (x_1, \ldots, x_n)\in \mX^n$, where $\mX = \zo^k$.
    \item Output: a graph stream $\gs \in \mGS^{2k}$%
\end{itemize}
\begin{enumerate}
    \item %
    \sout{The first element of the graph stream contains no edges, and $n$ nodes $v_1,\ldots, v_n$; each individual $x_i$ in the input dataset corresponds to node $v_i$.}
    \item \sout{Let $\gs$ be a graph stream, represented as a vector with $2k$ entries.}
    \item For all time steps $t\in [k]$:
    \begin{enumerate}
        \item If $t < k$: \connor{Only need this to keep stream of length $2k$ instead of $2k+1$}
        \begin{enumerate}
            \item Add $D$ copies of $s^{(D,2)}$ to $\gs_{2t+1}$, as defined in \cref{defn:stream elts reduction}.
        \end{enumerate}
        \item For all  \sout{$i\in [n]$:}
        \begin{enumerate}
            \item If $x_i[t] = 1$, add $D$ copies of $s^{(D,1,1)}_i$ to $\gs_{2t}$, as defined in \cref{defn:stream elts reduction}.%
            \item Otherwise, add $D$ copies of $s^{(D,1,0)}_i$ to $\gs_{2t}$, as defined in \cref{defn:stream elts reduction}.
        \end{enumerate}
    \end{enumerate}
    \item Return $\gs$.
\end{enumerate}
\end{algorithm}

\begin{algorithm}[ht!]
    \caption{Algorithm $\mA$ for one-way marginals}
    \label{alg:reduction from owms}
    \connor{I'll make this nice later. For now, I only have the algorithm typed out in the ``enumerate'' environment. I plan to type it as an ``algorithmic'' once it's finalized.}
    
    \begin{itemize}
    \item Input: a dataset $x = (x_1, \ldots, x_n)\in \mX^n$, where $\mX = \zo^k$, and black-box access to mechanism $\mM$.
    \item Output: $a = (a_1,\ldots, a_k)\in [0,1]^k$.

\end{itemize}
\begin{enumerate}
    \item Set $\gs = \mA_G(x)$ (where $\mA_G$ is defined in \cref{alg:make graph}).
    \item Compute $a' = \mM(\gs)$. 
    \item For all $i\in[k]$, set $a_i = 1 - \dfrac{a'_{2i+1} - a'_{2i}}{Dn}$.
    \item Return $a$.
\end{enumerate}
\end{algorithm}

\begin{theorem}
\label{thrm:reduction priv acc}
    Let $\mA$ be \cref{alg:reduction from owms}, let $\mS$ be the set of all insertion-only graph streams with no solo edges (\cref{defn:graph stream no solo}). 
    For $\eps>0$, $\del\in[0,1]$, $\alpha\in \R^+$, and $D,k,n,T\in \N$ where $T \geq 2k$, 
     if the continual release mechanism $\mM$ is $(\eps,\del)$-DP and $(\alpha, T)$-accurate for $\maxflowedges_D$ (or $\subgraphedges_D$) with respect to $\mS$, then batch algorithm $\mA$ is $(\eps,\del)$-DP and $\left(\dfrac{\alpha}{Dn},n\right)$-accurate for $\margs_k$.
\end{theorem}

\connorq{\cref{thrm:reduction priv acc} used to appear at the end of \cref{sec:reduc from owms}, but the discussion on 5/31/2023 indicated that it should be moved to the start of \cref{sec:reduc from owms}. Is this is a good location, or should it be presented somewhere else?}

For this reduction, we describe a way to encode an input dataset as an insertion-only graph stream with no solo edges (\cref{defn:graph stream no solo}), in a manner such that neighboring datasets map to node-neighboring graph streams. We will see in the reduction that releasing the vector which approximates the edge count at every time step is tantamount to releasing an approximation to the one-way marginals of the input dataset.

Before proceeding with the formal reduction, we describe the reduction at a high level.

\begin{intuition}[Intuition for the reduction from one-way marginals]
\connor{Skip this if you're reading. It will change as the proofs change, and the helpfulness of this note can be improved.}
    Recall that each row of the input dataset corresponds to one individual. In the reduction, the graph will begin with $n$ nodes $v_1,\ldots,v_n$, where $v_i$ corresponds to row $x_i$; and the graph will begin with no edges. The algorithm will release outputs at $2k$ time steps. For all $t\in[k]$, $D$ elements will arrive for each node $v_i$ in time step $2t$. The elements which arrive in time $2t$ will encode whether $x_i[t] = 1$ or $x_i[t] = 0$: namely, each of the $D$ elements associated with $v_i$ which arrives in time $2t$ will have an edge to $v_i$ if and only if $x_i[t] = 1$.

    The rough motivation for this encoding is that the number of edges will increase by $D\cdot \phi_t(x)$ in time step $2t$. However, then any node $v_i$ for which $x_i[t] = 1$ will have $D$ edges, which would seem to indicate that $v_i$ cannot admit additional edges in subsequent time steps. This is rectified by the $D$ stream elements per node $v_i$ which arrive at time step $2t+1$, each of which is associated with, and has an edge to, a (distinct) stream element which arrived in time step $2t$. The arrival of these elements in time $2t+1$ can be thought of as enabling a max flow where none of the flow is going to or coming from $v_i$ (or enabling an edge-maximal, $D$-bounded subgraph where none of the edges from $v_i$ are included in the subgraph). In time step $2t+2$, then, the arrival of nodes with edges to $v_i$ will again increase the edge count, because $v_i$ will effectively have degree 0 and can therefore admit new edges.

    \connor{There are a few things to fix here, as discussed on 5/31/2023. I think Adam indicated that he may have more suggestions on how to present this intuition after he goes through the actual theorems and proofs, so for now I'm holding off on making any changes for now. Basically, I should pretend that I'm a teacher explaining the idea for the proof, which makes sense. Two major changes: (1) use the term ``gadgets'' instead of ``elements'' (and maybe define them prior to the note on intuition); and (2) mention early on that time steps are paired up, and that the gadgets which arrive in the even time steps result in edge counts that give information about one of the marginals, and that the gadgets which arrive in odd time steps result in edge counts that are the same regardless of the value of the one-way marginals.}
    
\end{intuition}

\begin{remark}
We remark that we do not need to bring in $D^2$ nodes for each $v_i$ in every time step but, from the perspective of exposition and analysis, this technique is simpler.

\connor{I'm pretty sure we could kind of share nodes in a clever way, which would give the same lower bound but would require fewer nodes. I think it's (1) harder to explain and (2) I'm not 100\% sure whether the max flow stuff works out (I think it will?), though I think the edge-maximal, degree-$D$ subgraph will work. Is there a benefit to doing this analysis too, or should I just go with what I have?}
\end{remark}

We formally define the two types of stream elements that we use to encode a dataset as a graph stream. The reduction, which uses these stream elements, is provided in \cref{alg:reduction from owms}.

we should probably put the following definition in the relevant algorithm unless we use it in other places as well.

\begin{definition}[Graph stream elements for reduction]
\label{defn:stream elts reduction}
    We use two types of graph stream elements in our reduction. Let $s^{(D,1)}$ be a $(D-1)$-star, and let $s^{(D,2)}$ be a node with a single edge to the center node of $s^{(D,1)}$.

    There are two variants of $s^{(D,1)}$. We use $s^{(D,1,1)}_i$ to denote the $(D-1)$-star which also has an edge to its associated node $v_i$ (so $s^{(D,1,1)}_i$ is actually a $D$-star). We use $s^{(D,1,0)}_i$ to denote the $(D-1)$-star which does not have an edge to its associated node $v_i$ (so $s^{(D,1,0)}_i$ is in fact a $(D-1)$-star).
    
    \connor{An explanation of the notation choice: $s$ because they're stream elements, $D$ is a parameter for the degree bound, and $1,2$ to distinguish stream element 1 from stream element 2. Please let me know if you think of a better notation.}
\end{definition}

To prove \cref{thrm:reduction priv acc}, which describes the privacy and accuracy properties of this reduction, we first need the following lemma about the reduction.

\begin{lemma}
\label{lemma:graph behavior}
    In \cref{alg:make graph}, the graph stream $\gs$ returned by $\mA_G$ has the property that for $j\in[k]$ and $D\in\N$, if $f\in \{\hf^e_D, f^e_D \}$, then
    \[
    f(\gs_{[2j+1]}) - f(\gs_{[2j]})
    \]
    is equal to the number of occurrences of $s^{(D,1,0)}$ in time step $2j$.
\end{lemma}

\begin{proof}[Proof of \cref{lemma:graph behavior}]

\question{This proof feels verbose, but I don't have any concrete ideas for where it should be tightened. Are there any ``obvious'' adjustments in style that I should make to shorten things up? (I recognize that this question is fairly open-ended, but I don't know how to make it more specific since I'm unsure how to shorten the proof --- sorry.)}

Note that $f(\gs_{[2j+1]}) - f(\gs_{[2j]})$ is the number of additional edges counted for $\gs_{[2j+1]}$ (the graph at time step $2j+1$) as compared to $\gs_{[2j]}$ (the graph at time step $2j$). By construction, each node $v_i$ receives $D$ new edges in time step $2j+1$. However, not all of these edges will necessarily be counted by $f$.

We first consider the case where we have $f = \hf^e_D$. We claim that, prior to the arrival of $s^{(D,2)}$, the flow into and out of the center node of an element $s^{(D,1,1)}$ is already $D$. We now show why the center node of $s^{(D,1,1)}$ has an in- and out-flow of $D$. We can see that the max flow gives the center node of $s^{(D,1,1)}$ an in- and out-flow of at least $D-1$, since, aside from $v_i$, the nodes which are adjacent to the center node of $s^{(D,1,1)}$ are not adjacent to any other nodes, and the center node of $s^{(D,1,1)}$ has degree $\leq D$, so edges to each of these nodes will contribute a total of $D-1$ units of flow. We note that there is also one unit of flow between $v_i$ and the center of $s^{(D,1,1)}$ in each direction. This is because each copy of $s^{(D,1)}$ which arrived in a previous time step can have in- and out-flow between itself and (1) the $D-1$ other nodes in the star and (2) the node in $s^{(D,2)}$, which frees up $v_i$ to send and receive a unit of flow between itself and the center of $s^{(D,1,1)}$. The flow into and out of the center node of $s^{(D,1,1)}$ is therefore $D$ before the arrival of $s^{(D,2)}$.

When $s^{(D,2)}$ is added, only one node and one edge are added. For the flow to increase, there must be flow on this edge, and on all of the edges that already had flow. However, the node in $s^{(D,2)}$ is only adjacent to the center node of $s^{(D,1,1)}$, and this node already has an in- and out-flow of $D$. This means that any additional unit of in- or out-flow between the node in $s^{(D,2)}$ and the center node of $s^{(D,1,1)}$ must come at the cost of losing a unit of in- or out-flow between some other node and the center node of $s^{(D,1,1)}$, so the overall flow (and therefore the overall edge count returned by $f = \hf^e_D$) will remain the same.

By contrast, for the nodes $v_i$ where $x_i[j] = 0$, the nodes at the centers of the $D$ copies of $s^{(D,1,0)}$ only have degree $D-1$ prior to the addition of $s^{(D,2)}$, so the max flow-based edge count increases by 1 for each of the $D$ copies of $s^{(D,2)}$ since there will now be one unit of in- and out-flow between the center node of $s^{(D,1,0)}$ and the node in $s^{(D,2)}$. 

We see, then, that the change in the output of $f = \hf^e_D$ between time step $2j$ and $2j+1$ will be equal to the number of occurrences of $s^{(D,1,0)}$ in time step $2j$, which is what we wanted to show for $f = \hf^e_D$.

Now consider the case where we have $f = f^e_D$. We claim that when $s^{(D,1,1)}$ arrives, the center node of $s^{(D,1,1)}$ will have degree $D$ in the edge-maximal, degree-$D$ subgraph. We now prove this claim. Each copy of $s^{(D,1)}$ which arrived in a previous time step can have an edge between itself and (1) the $D-1$ other nodes in its star and (2) the node in $s^{(D,2)}$, so each of these centers of $s^{(D,1)}$ elements has $D$ edges, which is edge-count maximizing. None of these edges are to $v_i$, so there can be an edge between $v_i$ and the center node for each of the $D$ new copies of $s^{(D,1,1)}$. Therefore, because there are $D-1$ edges to the center node of $s^{(D,1,1)}$ from the other nodes in the star, and there is 1 edge from $v_i$, a total of $D$ edges to the center node of $s^{(D,1,1)}$ are being counted before the arrival of $s^{(D,2)}$.

We now consider the change in edge count when $s^{(D,2)}$ arrives. The node in $s^{(D,2)}$ only has an edge to the center of $s^{(D,1,1)}$, so the edge count of the edge-maximal, degree-$D$ subgraph can only increase by having an additional edge to the center of $s^{(D,1,1)}$. However, $s^{(D,1,1)}$ already has $D$ edges, so it cannot admit a new edge, and the edge count returned by $f = f^e_D$ will not change.

On the other hand, for the nodes $v_i$ where $x_i[j] = 0$, the nodes at the center of the $D$ copies of $s^{(D,1,0)}$ only have degree $D-1$. Therefore, the edge count returned by $f = f^e_D$ increases by 1 for each of the $D$ copies of $s^{(D,2)}$ because the edge between the center node of $s^{(D,1,0)}$ and the node in $s^{(D,2)}$ will be included in the degree-$D$, edge-maximal subgraph.

We see, then, that the change in the output of $f^e_D$ between time step $2j$ and $2j+1$ will be equal to the number of occurrences of $s^{(D,1,0)}$ in time step $2j$, which is what we wanted to show for $f = f^e_D$.

\end{proof}

\begin{proof}[Proof of \cref{thrm:reduction priv acc}]

We first evaluate the privacy of $\mA$. Consider any two neighboring datasets $x\simeq x'\in\mX^n$, where $\mX = \zo^k$, which are provided as input to $\mA$. By definition, there exists at most one index $i\in[n]$ such that $x_i \neq x'_i$. Let $\gs$ and $\gs'$ be the graph streams produced by $\mA_G$ on inputs $x$ and $x'$ respectively. We see, by construction, that for all $t\in [2k]$, $\gs_{[t]}$ and $\gs'_{[t]}$ are node neighbors. This is because $\gs_{[t]}$ and $\gs'_{[t]}$ only differ in the edges to nodes $v_i$ and $v'_i$; with the exception of these edges, the stream elements are all identical at all time steps.

After applying $\mM$ to the graph stream output by $\mA_G$, the algorithm $\mA$ only post-processes the output stream of $\mM$. Since $\mM$ is $(\eps,\del)$-DP and $\mA$ only post-processes the output stream of $\mM$, by \cref{lemma:post proc} we see that $\mA$ is also $(\eps, \del)$-DP, which completes the proof of item (1).

We next prove item (2) by evaluating the accuracy of $\mA$. We do this by showing a relationship between $F$ and $\margs_k$. Recall from \cref{eq:jth marginal} that, for all $j\in[k]$, the $\ord{j}$ marginal is defined as
\[
\phi_j(x) = \frac{1}{n}\cdot \sum_{i\in[n]} x_i[j].
\]
Additionally, \cref{lemma:graph behavior} tells us that $f(\gs_{[2j+1]}) - f(\gs_{[2j]})$ is equal to the number of occurrences of $s^{(D,1,0)}$ in time step $2j$. By the construction of the graph stream returned by \cref{alg:make graph}, the number of copies of $s^{(D,1,0)}$ in time step $2j$ is equal to $D\cdot \sum_{i\in[n]}\left(1 - x_i[j] \right)$, so 
\[
f(\gs_{[2j+1]}) - f(\gs_{[2j]}) = D\cdot \sum_{i\in[n]}\left(1 - x_i[j] \right).
\]

Therefore, for all $j\in[k]$ we have 
\begin{equation}
\label{eq:margs and edges}
    \phi_j(x) = \frac{1}{n}\cdot \sum_{i\in[n]} x_i[j] = 1 - \dfrac{D\cdot \sum_{i\in[n]}\left(1 - x_i[j] \right)}{D}\cdot \frac{1}{n} = 1 - \frac{f(\gs_{[2j+1]}) - f(\gs_{[2j]})}{D\cdot n}.
\end{equation}
The relation shown in \cref{eq:margs and edges} is precisely the transformation that \cref{alg:reduction from owms} does on the output stream from $\mM$ to obtain its estimates of the one-way marginals. The mechanism $\mM$ returns a randomized output that is $(\alpha,T)$-accurate for $F$, so the output stream from $\mM$ is not necessarily equal to $F(\gs)$, which means we need to consider the accuracy of $\mM$ to evaluate the accuracy of $\mA$. The post-processing which the algorithm $\mA$ does on the output of $\mM$ is deterministic, so the coins used by $\mA$ are the same as the coins used by $\mM$. As defined in \cref{alg:reduction from owms}, let $\gs = \mA_g(x)$, $a' = \mM(\gs)$, and $a = \mA(x)$. By \cref{eq:margs and edges} and the transformation in \cref{alg:reduction from owms} to compute estimates for one-way marginals from the output of $\mM$, we have
\begin{align*}
    \MoveEqLeft \Pr_{\text{coins of $\mA$}}\left[ \err_{\margs_k}\left( x, \mA(x) \right)  \leq \frac{\alpha}{Dn} \right] \\
    &= \Pr_{\text{coins of $\mA$}}\left[ \max_{j\in[k]} \left|\phi_j(x) - a_j \right| \leq \frac{\alpha}{Dn} \right] \\
    &= \Pr_{\text{coins of $\mM$}}\left[ \max_{t\in[k]} \left| \left(f(\gs_{[2t+1]}) - f(\gs_{[2t]})\right) - (a'_{2t+1} - a'_{2t}) \right| \leq 2\alpha \right] \\
    &\geq \Pr_{\text{coins of $\mM$}}\left[ \max_{t\in[T/2]} \left| \left(f(\gs_{[2t+1]}) - f(\gs_{[2t]})\right) - (a'_{2t+1} - a'_{2t}) \right| \leq 2\alpha \right] \tag{note that we now take $\max_{t\in[T/2]}$} \\
    &\geq \Pr_{\text{coins of $\mM$}}\left[ \max_{t\in[T]} \err_F\left( \gs_{[t]}, a'_t \right) \leq \alpha \right] \tag{triangle inequality} \\
    &\geq \frac{2}{3},
\end{align*}
with the final inequality following from the fact that $\mM$ is $(\alpha, T)$-accurate for $F$. Therefore, by \cref{defn:accuracy batch} we see that $\mA$ is $\left(\dfrac{\alpha}{Dn},n\right)$-accurate for $\margs_k$, which is what we wanted to show.
\end{proof}

\subsection{Proof of \texorpdfstring{\cref{thrm:lower bound edge count}}{lower bounds for edge counts}}
\label{sec:pf of lower bounds}

We now use \cref{thrm:reduction priv acc} to prove the lower bounds in \cref{thrm:lower bound edge count}.

\begin{proof}[Proof of \cref{thrm:lower bound edge count}]

By \cref{thrm:reduction priv acc}, we see that if $\mM$ is $(\eps,\del)$-DP and $(\alpha, T)$-accurate for $F$, then $\mA$ as defined in \cref{alg:reduction from owms} is $(\eps,\del)$-DP and $\left(\dfrac{\alpha}{Dn},n\right)$-accurate for $\margs_k$, for all $k\leq T/2$. 

We begin by proving item (1). By item (1) of \cref{lemma:owm lower bounds}, if $\del = 0$, then (where we set $\gamma = \frac{\alpha}{Dn}$) we must have
\[
n = \Omega\left( \dfrac{k}{\frac{\alpha}{Dn} \eps} \right),
\]
so equivalently we have $\alpha = \Omega\left( kD/\eps \right)$.
\cref{thrm:reduction priv acc} holds for all $k\leq T/2$, so by maximizing $k$, we see that we must have $\alpha = \Omega\left(DT/\eps \right)$, which is what we wanted to show. Note that, although \cref{lemma:owm lower bounds} only applies for $\frac{\alpha}{Dn} = \gamma \in (0,1)$, the value of $n$ in our reduction is completely independent of all other parameters, so we can set $n$ as large as necessary to satisfy the requirement in \cref{lemma:owm lower bounds} that $\frac{\alpha}{Dn}\in (0,1)$, \connorq{Is it fine to ``set $n$ as large as necessary to satisfy the requirement in \cref{lemma:owm lower bounds} that $\frac{\alpha}{Dn}\in (0,1)$''? I think yes, but I feel uncertain and I'd like a second opinion.} so the lower bound of $\alpha = \Omega\left(DT/\eps \right)$ holds for all settings of parameters permitted by \cref{thrm:lower bound edge count}.

We now move to a proof of item (2). We first prove a lower bound on $\alpha$, and then we show an upper bound on the values of $\del$ for which this bound on $\alpha$ holds. By item (2) of \cref{lemma:owm lower bounds}, if $\del > 0$ and $\del = o(1/n)$, then (where we set $\gamma = \frac{\alpha}{Dn}$) we must have
\begin{equation}
\label{eq:bound on n for delta}
    n = \Omega\left(\dfrac{\sqrt{k}}{\frac{\alpha}{Dn} \eps \log k}\right),
\end{equation}
so equivalently we have
\[
\alpha = \Omega\left( \dfrac{D\sqrt{k}}{\eps\log k} \right).
\]
Note that \cref{thrm:reduction priv acc} holds for all $k\leq T/2$, so by maximizing $k$, we see that we must have
\[
\alpha = \Omega\left( \dfrac{D\sqrt{T}}{\eps\log T} \right),
\]
which is what we wanted to show. As in the proof of item (1), note that, although \cref{lemma:owm lower bounds} only applies for $\frac{\alpha}{Dn} = \gamma \in (0,1)$, the value of $n$ in our reduction is completely independent of all other parameters, though we need to be more careful about selecting $n$ here than we were in proving item (1). Because $\delta$ is related inversely to $n$ and we want to show that the theorem statement holds for $\del$ as large as possible, for fixed $D$ and $\alpha$, we minimize $n$ such that we still have $\frac{\alpha}{Dn} \in (0,1)$. \connorq{This is related to the question above, so the answer there may affect the answer here. Are we allowed to ``we minimize $n$ such that we still have $\frac{\alpha}{Dn} \in (0,1)$''?} 

We now use our lower bound on $n$ from \cref{eq:bound on n for delta} to write a more usable bound on $\delta$. Substituting our bound from \cref{eq:bound on n for delta} into the requirement that $\del = o(1/n)$, we have
\[
\del = o\left( \dfrac{\alpha \eps \log k}{Dn\sqrt{k}} \right).
\]
Since we used the minimum $n$ such that $\frac{\alpha}{Dn} = \gamma \in (0,1)$, we can ignore the term $\frac{\alpha}{Dn}$ since $\frac{\alpha}{Dn}\approx 1$. \connorq{Can we in fact ``ignore the term $\frac{\alpha}{Dn}$ since $\frac{\alpha}{Dn}\approx 1$''? I think yes, but I'd like a second opinion.} This gives
\[
\del = o\left( \dfrac{\eps \log k}{\sqrt{k}} \right).
\]
Recall that, for our analysis of $\alpha$ above, we maximized $k\in\Z$ s.t. $k\leq T/2$ (so $k\approx T/2$), which gives us 
\[
\del = o\left( \dfrac{\eps \log T}{\sqrt{T}} \right).
\]
This is the bound on $\alpha$ and $\del$ that we wanted to show, which completes the proof of item (2).
    
\end{proof}

Note that, in the proof of item (2), we could potentially tune parameters more carefully and get a statement for even larger $\del$, but the value of $\del$ above is sufficiently large (i.e., shrinks with roughly $\sqrt{T}$) and the lower bound on $\alpha$ is also sufficiently large as to offer a compelling lower bound on accuracy.

\fi

\section{Stable and Time-Aware Projections}
\label{sec:projection}
\label{sec:tap alg}

\connornote{For a future version: improve the presentation of proofs in this section.}
In this section we present two \emph{time-aware projection} algorithms $\projo$ and $\projp$, and prove \cref{thrm:combined-stab} that presents their robust stability guarantees when run on $(D,\ell)$-bounded graph streams.
The two projection algorithms follow very similar strategies at a high level and are presented together in \cref{alg:time aware projection}. 
Both 
take as input a graph stream $\gs$ of length $T\in \N$ and some user-specified value $D\in\N$ and 
return a projected graph stream with maximum degree at most $D$. If the graph stream $\gs$ is already $D$-bounded, then both algorithms output it unchanged. At each time step, both algorithms 
greedily choose and output 
some subset of the arriving edges to include in the projection.

\begin{algorithm}[ht!]
    \caption{$\Pi_D$ for time-aware graph projection by edge addition.}
    \label{alg:time aware projection}
    
    \begin{algorithmic}[1]
        \Statex \textbf{Input:} Graph stream $(\gs_1, \ldots, \gs_T) = \gs \in \mGS^T$, time horizon $T\in\N$, degree bound $D\in \N$, and inclusion criterion $\style \in \{ \original, \projection\}$.
        \yoakland{\Statex \Comment{\commentstyle{$\style =\original$ yields BBDS-based projection  $\projo$}}
        \Statex \Comment{\commentstyle{$\style =\projection$ yields DLL-based projection  $\projp$}}}
        \noakland{\Statex \Comment{\commentstyle{$\style =\original$ yields BBDS-based projection  $\projo$; $\style =\projection$ yields DLL-based projection  $\projp$}}}

        \Statex \textbf{Output:} Graph stream $(\gs^*_1, \ldots, \gs^*_T) = \gs^* \in \mGS^T$

        \For{$t = 1$ to $T$}
            \State Parse $\gs_t$ as $(\nodeel_t,\edgeel_t)$ 
            \For{$v$ in $\nodeel_t$} $d(v) = 0$ \label{line:init counters} \EndFor
            \State $\edgeel^\mathit{proj}_t = \emptyset$
            \For{$e = \{u,v\}$ in $\edgeel_t$, in 
            consistent
            order,}
            \label{line:for loop proj}
                \State $\addedge \gets (d(u) < D) \land (d(v) < D)$\label{line:add edge condition}
                \If{$\addedge$}
                    \State set $\edgeel^\mathit{proj}_t = \edgeel^\mathit{proj}_t \cup \{e\}$ \label{line:check edge}
                    \yoakland{\State} \Comment{\commentstyle{add edge $e = \{u,v\}$ to the projection}}\label{line:add edge to proj}
                \Else{\ ignore $e=\{u,v\}$}
                \EndIf
                \If{$\style = \original$ or $\addedge$}
                    \State $d(u) \mathrel{+}= 1$, $d(v) \mathrel{+}= 1$ \label{line:increment counter} 
                    \yoakland{\State} \Comment{\commentstyle{ increment  degree counters for $u,v$}}
                \EndIf
            \EndFor
            \State Output $S^*_t = (\nodeel_t, \edgeel^\mathit{proj}_t)$ 
        \EndFor
    \end{algorithmic}
\end{algorithm}

\cref{alg:time aware projection} takes parameter $\style$, called the \emph{\inccrit{}}, that determines which of the two projections it executes. 
Let $\projo$ denote \cref{alg:time aware projection} with \inccrit{} $\style = \original$ and let $\projp$ denote the version with $\style = \projection$.
(We use the author initials of \cite{BlockiBDS13,DayLL16} to denote the algorithms inspired by their respective projections.)

The two algorithms differ from each other in terms of how they decide whether an edge should be added to the projection so far. The first algorithm $\projo$
adds edge $e = \{u,v\}$ if the degree of both end points $u$ and $v$ is less than $D$ in the \emph{original} graph stream so far (i.e., $\gs$ restricted to all edges considered before $e$). The second algorithm $\projp$
adds edge $e = \{u,v\}$ if nodes $u$ and $v$ both have degree less than $D$ in the \emph{projection} so far (i.e. $\Pi_D(\gs)$ restricted to all edges considered before $e$).

\myparagraph{Consistent Ordering.}
When multiple edges arrive in a time step, the projections must decide on the order in which to consider these edges. While the exact ordering does not matter, we assume a
\emph{consistent ordering} of the edges in the input graph stream.
Consistency means that any pair of edges in neighboring graph streams should be considered for addition to the projection in the same relative order. 
A similar ordering assumption is made by \cite{BlockiBDS13,DayLL16}.

A simple implementation of such an ordering assumes that each node $u$ has a unique string identifier $\id_u$---a user name, for example---and orders edges  according to their endpoints (so $(u,v)$ gets mapped to $(\id_u,\id_v)$, where $\id_u<\id_v$ and pairs are ordered lexicographically). 

Since both projections process edges at the time they arrive, they end up considering edges according to a \emph{time-aware version} of the ordering: edges end up being considered in the lexicographic order given by the triples $(t,\id_u,\id_v)$, where $t$ is the edge's arrival time.

Ordering the edges uniformly randomly within each time step would also suffice since one can couple the random orderings on two neighboring streams so they are consistent with each other. We omit a proof of this, and assume lexicographic ordering in the rest of this manuscript.

\begin{remark}[Running \cref{alg:time aware projection} on static graphs]
\label{rmk:tap-graph-input}
    \cref{alg:time aware projection} can also take a (static, not streamed) graph as input by interpreting the graph as a length-$1$ graph stream, where the first element of the graph stream is equal to the graph itself.
\end{remark}

\subsection{Stability of the Time-Aware Projection Algorithms}
\label{sec:stability-thm-overview}

Our analysis of the projection algorithms
differs significantly from the batch-model analyses of \cite{BlockiBDS13,DayLL16}.
First, we consider the stability of the entire projected sequence, and not a single graph. Second,
Blocki et al. \cite{BlockiBDS13} consider only the edge-to-edge stability of their projection algorithm, 
while Day et al. \cite{DayLL16} only analyze the stability with respect to a particular function of the projected graph (namely, its degree distribution).
We analyze several stronger notions of stability for the entire sequence produced by our projections. All but one of these stability guarantees hold for streams that are $(D,\ell)$-bounded.

\oldtext{--- but are especially useful in the continual release model since they allow us to leverage known restricted edge-DP and restricted node-DP algorithms such as those of \cite{FHO21}. (An alternative explanation is that, when looking at functions of the projected graph 
stream, edge stability translates to the type of stability that is sufficient to apply the tree mechanism of \cite{DNPR10,CSS11}, a workhorse of accurate continual-release algorithms.) While \cref{alg:time aware projection} with $\style = \projection$ requires a private test for edge-to-edge stability, the edge-to-edge stability property of \cref{alg:time aware projection} with $\style = \original$ admits edge-private algorithms, and it does not require any condition on the input graph streams, so this property can be leveraged to make pure-DP edge-private algorithms.

As described above, some of the stability properties rely on the input stream satisfying a condition that the input is \emph{$(D,\ell)$-bounded}.}

\begin{definition}[$(D,\ell)$-bounded]
\label{defn:d ell bdd}
    We say that a graph $G$ is \emph{$(D,\ell)$-bounded} if it has at most $\ell$ nodes of degree greater than $D$. 
    Similarly, a graph stream $\gs$ of length $T$ is \emph{$(D,\ell)$-bounded through time $t\in[T]$} if the flattened graph $\collapse{\gs_{[t]}}$ is $(D,\ell)$-bounded (i.e., has at most $\ell$ nodes of degree greater than $D$).
\end{definition}

We now present our theorem on the stability of \cref{alg:time aware projection}. These stabilities are summarized in \cref{tab:projection-stabilities}.

\begin{theorem}[Stability of projections]
\label{thrm:combined-stab} 
    Let $T\in\N$, $D\in \N$, $\ell\in \N\cup\{0\}$, and let $\projo,\projp$ be \cref{alg:time aware projection} with \inccrit{} $\style = \original$ and $\style=\projection$, respectively.
    \begin{enumerate}[leftmargin=*]
        \item \textbf{(Edge-to-edge stability.)} If $\gs\nedge \gs'$ are edge-neighboring graph streams of length $T$, then for all time steps $t\in[T]$, the edge distances between the projections through time $t$ satisfy the following:
        \begin{enumerate}
            \item
            \( 
                \dedge\Big(\projo(\gs)_{[t]}\ ,\ 
                \projo(\gs')_{[t]} \Big)
                \leq 3.
            \)
        
            \item If $\gs,\gs'$ are $(D,\ell)$-bounded through time $t$, then 
            \(
                \dedge\Big(\projp(\gs)_{[t]}\ ,\ 
                \projp(\gs')_{[t]} \Big) \leq 2 \ell + 1.
            \)
        \end{enumerate}
        \item \textbf{(Node-to-edge stability.)} If $\gs\nnode \gs'$ are node-neighboring graph streams of length $T$, then for all time steps $t\in[T]$, the edge distances between the projections through time $t$ satisfy the following:
        \begin{enumerate}
            \item If $\gs,\gs'$ are $(D,\ell)$-bounded through time $t$, then
            \( 
                \dedge\Big(\projo(\gs)_{[t]}\ ,\ 
                \projo(\gs')_{[t]} \Big)
                \leq
                D + \ell.
            \)
                
            \item If $\gs,\gs'$ are $(D,\ell)$-bounded through time $t$, then
            \ifnum\oakland=1
            \end{enumerate}
            \vspace{-6pt}
            \begin{eqnarray*}
                \dedge\Big(\projp(\gs)_{[t]},
                \projp(\gs')_{[t]} \Big)
                \leq
                D + 2\ell\sqrt{\min\{D,\ell\}}.
            \end{eqnarray*}
            \fi
            \ifnum\oakland=0
            \(
                \dedge\Big(\projp(\gs)_{[t]}\ ,\ 
                \projp(\gs')_{[t]} \Big)
                \leq 
                \begin{cases}
                    D + 2\ell^{3/2} & \text{if } D\geq \ell,\ \text{and}\\
                    D + 2\ell \sqrt{D} & \text{if } D< \ell.\\
                \end{cases}
            \)
            \end{enumerate}
            \fi

        \item \textbf{(Node-to-node stability.)} If $\gs\nnode \gs'$ are node-neighboring graph streams of length $T$, then for all time steps $t\in[T]$, the node distances between the projections through time $t$ satisfy the following:
        \begin{enumerate}
            \item If $\gs,\gs'$ are $(D,\ell)$-bounded through time $t$, then
            \( 
                \dnode\Big(\projo(\gs)_{[t]}\ ,\ 
                \projo(\gs')_{[t]} \Big)
                \leq
                2\ell + 1.
            \)
            
            \item If $\gs,\gs'$ are $(D,\ell)$-bounded through time $t$, then
            \(
                \dnode\Big(\projp(\gs)_{[t]}\ ,\ 
                \projp(\gs')_{[t]} \Big)
                \leq
                2\ell + 1.
            \)
        \end{enumerate}
    \end{enumerate} 
    Furthermore, the bounds above are all tight in the worst case, either exactly (bound 1(a)), up to an additive constant of 2 (bounds 2(a) and 3(a)), or up to multiplicative constants.
\end{theorem}

\noakland{The proofs that the bounds are tight are collected in \Cref{sec:stab-tight} (see \Cref{thrm:tightness-stab}).} \yoakland{The proofs that the bounds are tight appear in the full version.} Here, we focus on proving the upper bounds. 

\myparagraph{Proof Sketch for Stability of \cref{alg:time aware projection} (\cref{thrm:combined-stab}).}
To analyze stability, we consider a pair of graphs streams $\gs,\gs'$ that are either edge (part (1)) or node neighbors (parts (2) and (3)). Assume without loss of generality that 
$\gs'$ is the larger of the two streams. 
When $\gs'$ is larger by virtue of including an additional node, we use $\vplus$ to denote the additional node in $\gs'$.

Our proofs of stability rely heavily on two observations. First, the greedy nature of \Cref{alg:time aware projection} ensures that once an edge is added to the projection of a graph stream, that edge will not be removed from the projection at any future time step. Moreover, an edge may only be added to the projection at the time it arrives---it will not be added to the projection at any later time. 
This greedy behavior simplifies the analysis dramatically. It allows us to reason only about the distances between the flattened projected graphs at time $t$, rather than about the entire projected sequence: an edge that appears at some time $t'<t$ in one projected sequence but not the other will still differ between the flattened graphs at time $t$. This fact is captured in the ``Flattening Lemma'' (\cref{lem:greedy-flatten}), which we use
throughout the remainder of the argument.
(A different projection algorithm, for example one that recomputes a projection from scratch at each time step, would require us to more explicitly analyze the entire projected sequence.)

The other important observation for our analysis is the following. Consider two neighboring graphs, where one graph contains (at most) one additional node $\vplus$. In the larger graph, if an edge is between two nodes of degree at most $D$ and is not incident to the added node $\vplus$, then it will be in both projections. In other words, only edges that do not satisfy this condition may differ between projections. This idea is captured in \cref{lem:edges dont change}.

\myparagraph{(Stability of \boldmath{$\projo$}.)} To prove edge-to-edge stability, we apply \cref{lem:greedy-flatten} and largely borrow the analysis of \cite[Proof of Claim 13]{BlockiBDS13}.
To prove node-to-node stability, we only need to consider the projections through times $t\in[T]$ for which the graph streams $\gs$ and $\gs'$ are both $(D,\ell)$-bounded. We then use \cref{lem:greedy-flatten}, in addition to the fact that only an edge with an endpoint node of degree greater than $D$ in one of the original graphs may differ between projections of neighboring streams (\cref{lem:edges dont change}) to see that all nodes and edges that differ between graph streams belong to a vertex cover of size at most $\ell + 1$. Therefore, we can obtain $\gs_{[t]}$ from $\gs'_{[t]}$ by removing $\vplus$ and changing the remaining $\ell$ nodes in the vertex cover.

The node-to-edge stability also applies only through times $t\in[T]$ for which the graph streams $\gs$ and $\gs'$ are both $(D,\ell)$-bounded. Its proof requires more careful analysis of exactly how many edges incident to nodes of degree greater than $D$ may change, in addition to using \cref{lem:greedy-flatten,lem:edges dont change}. We first show that at most $D$ edges incident to the added node $\vplus$ may appear in the projection of $\gs'$; these edges cannot appear in $\gs$ or its projection.

We next consider whether any of the other edges differ between projections.
First, consider an added edge $\eplus$ incident to $\vplus$ and some ``high-degree'' node $u$ (i.e., with degree greater than $D$ in $\gs'$). The presence of $\eplus$ may mean exactly one edge incident to $u$ that was included in the projection of $\gs$ will now be dropped, since $u$ may already have degree $D$ when that edge is considered for addition. Now, suppose instead that $\eplus$ is incident to $\vplus$ and some ``low-degree'' node (i.e., with degree greater than $D$ in $\gs'$). None of the edges incident to $u$ will be dropped from the projection due to the inclusion of $\eplus$, because although the degree of $u$ is now larger, it is still safely at or below the threshold of $D$. Of the remaining edges, then, only edges incident to high-degree nodes may change.
Since there are at most $\ell$ nodes with degree greater than $D$, there are at most $\ell$ additional edges that differ between projections.
Therefore, by combining this with the above observation that at most $D$ edges incident to $\vplus$ appear in the projection of $\gs'$, we see that at most $D+\ell$ edges differ between projections through time $t$ for node-neighboring graph streams.

\myparagraph{(Stability of \boldmath{$\projp$}.)} The analysis of this projection, especially its node-to-edge stability, is generally more complex. One exception is the proof of node-to-node stability, which follows from the same argument used to prove node-to-node stability for $\projo$. All bounds on the stability of this projection only apply through times $t\in[T]$ for which the graph streams $\gs$ and $\gs'$ are both $(D,\ell)$-bounded (note that the edge-to-edge stability for $\projo$ does not rely on this assumption).

The node-to-edge stability analysis of this algorithm is more involved, though it also relies on \cref{lem:greedy-flatten}.
We consider a pair of arbitrary node neighbors and leverage the greedy nature of our algorithm to iteratively construct a \emph{difference graph} that tracks which edges differ between the projections of each graph. We show that the edge distance between the projections of neighboring graphs is exactly the number of edges in the difference graph\noakland{~(\cref{lem:dg has all})}.
If the difference graph on any two node neighbors were to form a DAG on $\ell+1$ nodes with at most $k$ paths of length at most $\ell$, then we would be able to bound its edge count by $2\ell\sqrt{k}$ and prove, by setting $k = \min\{D,\ell\}$, that the projections are at edge distance at most $2\ell\sqrt{\min\set{D,\ell}}$\noakland{ (\cref{lem:dag ell 32})}\yoakland{ (in the full version, we prove this upper bound on the number of edges in such a DAG)}.

In reality, we show that the difference graph is edge distance at most $D$ from a DAG with this special structure\noakland{ (\cref{lemma:pdg is dag etc})}.
This introduces an additive $D$ term in the edge distance between the projections (which is to be expected since the projections of two node neighbors could differ in $D$ edges that are incident to the differing node).
The ``pruning'' argument, which shows how to remove edges from the difference graph to obtain the special DAG, and the construction of the resulting DAG form the most involved part of our analysis. \noakland{In \cref{sec:node edge projected} we formally define a difference graph and prove \cref{lem:dag ell 32,lem:dg has all,lemma:pdg is dag etc}.}

The edge-to-edge stability of $\projp$ also uses \cref{lem:greedy-flatten}. If both graphs have maximum degree at most $D$, \cref{alg:time aware projection} acts as the identity, and the projections differ in at most one edge. For graphs with at most $\ell$ nodes of degree greater than $D$, we observe that all of the edges in the corresponding difference graph form two paths of length at most $\ell$, where all edges in each path are incident either to nodes with degree greater than $D$ in $\gs'$ or to the node with the added edge. Since there are at most $\ell$ of these high-degree nodes, the path has length at most $2 \ell + 1$.

\myparagraph{Useful Lemmas for Proving Stability.}
As described in the proof sketch, \cref{lem:greedy-flatten,lem:edges dont change} are used for proving many of the stability statements in \cref{thrm:combined-stab}. Before presenting the proofs of stability, we present those lemmas, along with definitions for terms that are used in the lemmas and their proofs.

The first definition is motivated by the observation that, if we run the algorithm on edge- or node-neighboring graph streams, edges are considered for inclusion in the output stream $\gs^*$ in the same relative order for both graph streams (i.e., edge $e_1$ is considered before edge $e_2$ in stream $\gs$ if and only if $e_1$ is considered before $e_2$ in stream $\gs'$).

\begin{definition}[Projection stage of an edge]
\label{defn:arrstep}
Let $T\in\N$ be a time horizon, $D\in \N$ be a degree bound, $\gs$ be a graph stream of length $T$, and let $\Pi_D\in \set{\projo,\projp}$ denote one of the variants of \cref{alg:time aware projection}.
An edge $e$ in $\gs$ \emph{\projstagealttext}{} $i$ of algorithm $\Pi_D(\gs)$ (denoted $\projstage(e)=i$) if $e$ is the $\ord{i}$ edge to be considered by $\Pi_D(\gs)$. %
\end{definition}

In the proofs that follow, we are often interested in the value of a counter $d(\cdot)$ (see Line~\ref{line:init counters} of \cref{alg:time aware projection}) in relation to a projection stage; we say that a counter $d(u)$ has value $j$ at projection stage $i$ if, when the $\ord{j}$ edge is considered for inclusion in the projection, $d(u) = j$ on Line~\ref{line:add edge condition}.

Multiple edges may arrive in a time step, so the \emph{\projstagetext} of an edge is related to but distinct from the \emph{arrival time} of the edge.
\noakland{Knowing the \projstagetext{} of an edge is useful since it captures the order in which \cref{alg:time aware projection} considers edges for inclusion in the output stream, which is a function of both an edge's arrival time and its placement in the consistent ordering on edges.}

The second definition comes from the following observation. Consider a node $u$ and counter $d(u)$ in \cref{alg:time aware projection}. If an edge $e$ incident to $u$ \projstagealttexttwo{} where \cref{alg:time aware projection} has $d(u) \geq D$, then $e$ will not be included in the output stream.
Since no more edges incident to $u$ will be included, node $u$ can be thought of as being \emph{saturated}. \cref{defn:satstep} allows us to talk about the order in which this saturation occurs.

\begin{definition}[Saturation stage of a node]
\label{defn:satstep}
    Let $T\in\N$ be a time horizon, $D\in \N$ be a degree bound, $\gs$ be a graph stream of length $T$, and let $\Pi_D\in \set{\projo,\projp}$ denote one of the variants of \cref{alg:time aware projection}.
    A node $u$ in $\gs$ has \emph{\satsteptext{} $b$} (denoted $\satstep_{\gs}(u)=b$) if $b$ is the first \projstagetext{} such that $d(u) \geq D$ by the end of Line~\ref{line:increment counter} in $\Pi_D(\gs)$. We define $\satstep_{\gs}(u) = \infty$ if there is no $b$ for which the described condition holds.
\end{definition}

\noakland{
\begin{definition}[Shortsighted algorithm]
\label{defn:node preserv greedy}
    An algorithm $\Pi:\mGS^T\to \mGS^T$ on graph streams of length $T$ is a \emph{shortsighted} algorithm if it
    \begin{enumerate}
        \item adds all input nodes to the output stream in the time when they first arrive and
        \item adds some subset of input edges arriving in that time to the output stream (and does not add any other edges).
    \end{enumerate}
\end{definition}

Observe that \cref{alg:time aware projection} is a shortsighted algorithm for \inccrit{} $\style = \original$ and $\style = \projection$.
}
\noakland{%
    \begin{definition}[Differing nodes and edges]
    \label{defn:differ}
        An edge $e$  (respectively, node $v$) \emph{differs} between a pair of graphs $G$ and $G'$ if $e$ (resp., $v$) appears in one graph but not the other.
        \noakland{We similarly say that an edge $e$ (resp., node $v$) differs between a pair of graph streams $\gs$ and $\gs'$ through time $t$ if $e$ (resp., $v$) arrives at time step $i\in[t]$ in one graph stream, but either
        \begin{enumerate}[label=(\alph*)]
            \item fails to appear in the other graph stream through time $t$ or
            \item appears at time step $j\neq i$ (where $j\in[t]$) in the other graph stream.
        \end{enumerate}}
    \end{definition}
}

\begin{lemma}[Flattening Lemma]
\label{lem:greedy-flatten}
    \yoakland{Let $\Pi_D\in \set{\projo,\projp}$ denote one of the variants of \cref{alg:time aware projection}.}
    \noakland{Let $\Pi_D\in \set{\projo,\projp}$ be a shortsighted algorithm.}
    For every edge- or node-neighboring pair of graph streams $\gs$ and $\gs'$ of length $T$, the edge- and node-distance between the \emph{projected streams} through time $t$ is the same as the edge- and node-distance between the \emph{flattened graphs} through time $t\in[T]$:
    \begin{eqnarray}
    \label{eq:edge dist}
        \lefteqn{\dedge\paren{\Pi_D(\gs)_{[t]},\Pi_D(\gs')_{[t]}}} \nonumber &&\\
        && = \dedge\paren{\flatten{\Pi_D(\gs)_{[t]}}, \flatten{\Pi_D(\gs')_{[t]}}}
    \end{eqnarray}
    and 
    \begin{eqnarray}
    \label{eq:node dist}
        \lefteqn{\dnode\paren{\Pi_D(\gs)_{[t]},\Pi_D(\gs')_{[t]}}} \nonumber &&\\
        && = \dnode\paren{\flatten{\Pi_D(\gs)_{[t]}}, \flatten{\Pi_D(\gs')_{[t]}}} .
    \end{eqnarray}
\end{lemma}

\ifnum\oakland=0
\begin{proof}[Proof of \cref{lem:greedy-flatten}]
    Let $\gs$ and $\gs'$ be graph streams of length $T$, and let $\Pi_D$ be a shortsighted algorithm. Let $P_{[t]},P'_{[t]}$ denote projected graph streams $\Pi(\gs)_{[t]}$ and $\Pi(\gs')_{[t]}$. Also, let $F_{[t]},F'_{[t]}$ denote the flattened graphs $\flatten{\Pi(\gs)_{[t]}}$ and $\flatten{\Pi(\gs')_{[t]}}$. 

    \textbf{(Proof of Expression~\ref{eq:edge dist}.)} We first prove Expression~\ref{eq:edge dist}, where $\gs$ and $\gs'$ are node-neighbors (and let $\gs'$ be the larger graph stream---it contains an additional node $\vplus$ and associated edges). We see $\dedge(F_{[t]},F'_{[t]}) \leq \dedge(P_{[t]},P'_{[t]})$ because if an edge differs between flattened graphs, then it must differ between the projections that were flattened to obtain the graphs (this same logic applies for isolated nodes that differ between flattened graphs). Note that this inequality holds for arbitrary (e.g., non-neighboring) graph streams.

    We now need to show $\dedge(P_{[t]},P'_{[t]}) \leq \dedge(F_{[t]},F'_{[t]})$. This expression does not necessarily hold for arbitrary graph streams. For node-neighbors, though, the edges in the corresponding projections $P_{[t]}$ or $P'_{[t]}$ can only differ according to (a) and not (b) in \cref{defn:differ} (and similarly for isolated nodes that differ between projections).
    This holds for two reasons. First, all edges that arrive in $\gs$ at time $i$ must also arrive in $\gs'$ at time $i$. 
    Second, $\projgen$ is a shortsighted algorithm, so if an edge arrives in one projected stream at time $i$, it either also appears in the other projected stream at time $i$ or it never appears at all. Therefore, all edges that differ between graph streams must differ according to (a). If an edge differs between graph streams through time $t$ according to (a), then it will also differ between the flattened versions of those graphs. This gives us $\dedge(P_{[t]},P'_{[t]}) \leq \dedge(F_{[t]},F'_{[t]})$, which yields the desired equality.

    If $\gs$ and $\gs'$ are edge-neighbors (and let $\gs'$ be the larger graph stream---it contains an additional edge), edges in the corresponding projections $P_{[t]}$ or $P'_{[t]}$ can also only differ according to (a) and not (b) since all edges that arrive in $\gs$ at time $i$ must also arrive in $\gs'$ at time $i$, and since $\projgen$ is a shortsighted algorithm (a similar statement applies to isolated nodes that differ between projections). This gives us $\dedge(P_{[t]},P'_{[t]}) \leq \dedge(F_{[t]},F'_{[t]})$. The node-neighboring argument for $\dedge(F_{[t]},F'_{[t]}) \leq \dedge(P_{[t]},P'_{[t]})$ also applies to edge neighbors, completing the proof.

    \textbf{(Proof of Expression~\ref{eq:node dist}.)} To show $\dnode(F_{[t]},F'_{[t]}) \leq \dnode(P_{[t]},P'_{[t]})$, the same argument as above applies (with the additional observation that nodes differing between graphs also differ between graph streams).
    We now show $\dnode(P_{[t]},P'_{[t]}) \leq \dnode(F_{[t]},F'_{[t]})$. Let $V_\mathit{remove}$ and $V_\mathit{add}$ be the sets of nodes (and edges incident to these nodes) that are removed from and added to $F$ to obtain $F'$, such that $|V_\mathit{remove}| + |V_\mathit{add}|$ is minimized.
    
    Recall from the argument for Expression~\ref{eq:edge dist} that nodes and edges differ between edge- and node-neighboring graph streams according only to (a) and not (b) in \cref{defn:differ}.
    Therefore, removing $V_\mathit{remove}$ from $P$ and adding $V_\mathit{add}$ to $P$ at the appropriate times (and also adding the edges at the appropriate times) gives us $P'$. This gives us $\dnode(P_{[t]},P'_{[t]}) \leq |V_\mathit{remove}| + |V_\mathit{add}|$, so we have $\dnode(P_{[t]},P'_{[t]}) \leq \dnode(F_{[t]},F'_{[t]})$, which completes the proof of Expression~\ref{eq:node dist}.
\end{proof}
\fi

\begin{lemma}[Edges between low-degree nodes remain in both projections]
\label{lem:edges dont change}
    Let $\Pi_D\in \set{\projo,\projp}$ denote one of the variants of \cref{alg:time aware projection}.
    Consider a pair of edge- or node-neighboring graph streams $\gs,\gs'$ of length $T$, where $\gs'$ contains (at most) one additional node as compared to $\gs$.
    For all edges $e = \{u,v\}$ that arrive in $\gs$ at (or before) time step $t\in[T]$, if $u$ and $v$ have degree at most $D$ in $\gs'_{[t]}$, then $e$ is in both $\Pi_D(\gs)$ and $\Pi_D(\gs')$.
\end{lemma}

\noakland{\begin{proof}[Proof of \cref{lem:edges dont change}]
    Without loss of generality let $\gs'$ be the larger graph stream (i.e., for node-neighbors, $\gs'$ can be formed by adding a node $\vplus$ and associated edges to $\gs$; and for edge-neighbors, $\gs'$ contains an additional edge as compared to $\gs$). Before proceeding with the proof, we note that all nodes and edges in $\gs_{[t]}$ are in $\gs'_{[t]}$. Additionally, we see that for all nodes $v$ in $\gs'_{[t]}$, the degree of $v$ is at least as large in $\gs'_{[t]}$ as in $\gs_{[t]}$---that is, $\deg_v(\gs'_{[t]}) \geq \deg_v(\gs_{[t]})$. 

    Let $e = \{u,v\}$ be an edge that arrives in $\gs$ at time $t$, and let $u$ and $v$ both have degree at most $D$ in $\gs'_{[t]}$.
    We now consider \cref{alg:time aware projection}. We see from Lines~\ref{line:add edge condition}--\ref{line:add edge to proj} that, if the counters $d(u)$ and $d(v)$ are both less than $D$ when we consider adding edge $e$ to the projection, then $e$ will be added. The counters are initialized to $0$, and we see on Line~\ref{line:increment counter} that, for all nodes $w$ in the input, $d(w)$ is incremented only if \cref{alg:time aware projection} considers adding an edge incident to $w$ to the projection.

    This means that, prior to considering $e=\{u,v\}$ in the for loop (Line~\ref{line:for loop proj}), the counters $d(u)$ and $d(v)$ will have each been incremented fewer than $D$ times when running $\Pi_D$ on $\gs$ or on $\gs'$. Therefore, \cref{alg:time aware projection} will have $d(u)<D$ and $d(v)<D$ when the for loop considers $e = \{u,v\}$ for addition to the projection, so $e$ will be added to the projection. We see that this argument also holds if $e$ arrives before time $t$.
\end{proof}}

\cref{lem:edges dont change} has an important consequence: if $U$ is the set of nodes that have degree more than $D$ in $\gs'_{[t]}$ (where $\gs'$ is the larger of two neighboring graph streams), then $U\cup \{v_+\}$ forms a vertex cover for the edges that differ between the projections $\Pi_D(\gs)_{[t]}$ and $\Pi_D(\gs')_{[t]}$. 
We use this in the proof of \Cref{thrm:combined-stab} in a few ways. Most directly, if we have an upper bound of $\ell+1$ on the size of $U$, then we immediately obtain an upper bound of $2\ell+1$ on the node distance between the projections $\Pi_D(\gs)_{[t]}$ and $\Pi_D(\gs')_{[t]}$, which gives us the proof of node-to-node stability. It is also the first step in the proofs of other stability statements. 

\ifnum\oakland=0

\subsection{Proof of Edge-to-Edge Stability for \texorpdfstring{\boldmath{$\projo$}}{BBDS}}

We begin by proving item (1a) of \cref{thrm:combined-stab}, which we repeat below for convenience. The proof follows immediately from \cref{lem:greedy-flatten} and the ideas from the proof of \cite[Claim 13]{BlockiBDS13}.

\begin{theorem}[Item (1a) of \cref{thrm:combined-stab}]
\label{thrm:edge-edge-stab}
Let $T\in\N$, $D\in \N$, $\ell\in \N\cup\{0\}$, and let $\projo$ be \cref{alg:time aware projection} with \inccrit{} $\style = \original$. If $\gs\nedge \gs'$ are edge-neighboring graph streams of length $T$, then for all time steps $t\in[T]$, the edge distances between the projections through time $t$ satisfy
\[
    \dedge\Big(\projo(\gs)_{[t]}\ ,\ 
    \projo(\gs')_{[t]} \Big)
    \leq 3.
\] 
\end{theorem}

\begin{proof}[Proof of \cref{thrm:edge-edge-stab}] 
    Let $\gs,\gs'$ be edge-neighboring graph streams (wlog, let $\gs'$ be formed by adding one edge $e^+=\{u,v\}$ to $\gs$). To simplify notation, let $F_{[t]}, F'_{[t]}$ denote $\flatten{\projo(\gs)_{[t]}}$ and $\flatten{\projo(\gs')_{[t]}}$ respectively. By \cref{lem:greedy-flatten}, we only need to consider how $F$ and $F'$ differ. 

    We consider three types of edges: edge $e^+ = \{u,v\}$, edges $e$ incident to neither $u$ nor $v$, and edges $e_u$ and $e_v$ incident to either $u$ or $v$ respectively. We first note that $e^+$ may appear in $F'_{[t]}$ and will not appear in $F_{[t]}$.
    
    Next consider edges $e=\{w,x\}$ incident to neither $u$ nor $v$. An edge of this form will either appear in both $F_{[t]}$ and $F'_{[t]}$, or it will appear in neither. This is because $e^+$ will not affect the value of $d(w)$ or $d(x)$, so the decision to add $e$ to the output stream will be the same, regardless of whether $\projo$ is running on $\gs$ or $\gs'$.

    Now consider edges incident to either $u$ or $v$. If $e^+$ \projstagealttexttwo{} where $d(u)\geq D$ and $d(v) \geq D$ in $\projo$ on $\gs'$, then the projections of both streams will be identical. However, if $e^+$ \projstagealttextthree, then there may be one edge $e_u$ incident to $u$ that appears in $F_{[t]}$ but does not appear in $F'_{[t]}$, due to having $d(u) = D$ \projstagealttextfour{$e_u$} for the case of running $\projo$ on $\gs'$, instead of having $d(u) = D-1 < D$ as is the case on $\gs$.
    Note that, if an edge incident to $u$ \projstagealttexttwo{} following that of $e_u$, then that edge will appear in neither $F_{[t]}$ nor $F'_{[t]}$. This is because $u$ is already saturated in $\gs$ and $\gs'$ when those edges are considered---that is, $d(u) \geq D$. Likewise, an edge incident to $u$ that \projstagealttexttwo{} prior to $e_u$ will either appear in both $F_{[t]}$ and $F'_{[t]}$ or will appear in neither. Similarly, there may be one edge $e_v$ incident to $v$ that appears in $F_{[t]}$ but does not appear in $F'_{[t]}$. Graphs $F'_{[t]}$ and $F_{[t]}$ differ on at most three edges---namely $e^+$, $e_u$, and $e_v$---so $\dedge(F_{[t]},F'_{[t]})\leq 3.$
\end{proof}

\fi

\ifnum\oakland=0
\subsection{Proof of Node-to-Node Stability for \texorpdfstring{\boldmath{$\projo$ and $\projp$}}{BBDS and DLL}}

We next prove item (3) of \cref{thrm:combined-stab}, which we repeat below for convenience. The proof follows from \cref{lem:greedy-flatten,lem:edges dont change}, with the same approach applying for \inccrit{} $\style = \original$ and $\style = \projection$.

\begin{theorem}[Item (3) of \cref{thrm:combined-stab}]
\label{thrm:node-node stab}
    Let $T\in\N$, $D\in \N$, $\ell\in \N\cup\{0\}$, and let $\projo,\projp$ be \cref{alg:time aware projection} with \inccrit{} $\style = \original$ and $\style=\projection$, respectively. If $\gs\nnode \gs'$ are node-neighboring graph streams of length $T$, then for all time steps $t\in[T]$, the node distances between the projections through time $t$ satisfy the following:
    \begin{enumerate}[label=(\alph*)]
        \item If $\gs,\gs'$ are $(D,\ell)$-bounded through time $t$, then
        \( 
            \dnode\Big(\projo(\gs)_{[t]}\ ,\ 
            \projo(\gs')_{[t]} \Big)
            \leq
            2\ell + 1.
        \)
        
        \item If $\gs,\gs'$ are $(D,\ell)$-bounded through time $t$, then
        \(
            \dnode\Big(\projp(\gs)_{[t]}\ ,\ 
            \projp(\gs')_{[t]} \Big)
            \leq
            2\ell + 1.
        \)
    \end{enumerate}
\end{theorem}

\begin{proof}[Proof of \cref{thrm:node-node stab}]
    Without loss of generality, let $\gs'$ be the larger graph stream---that is, it contains an additional node $\vplus$ and associated edges.
    To simplify notation, let $F_{[t]}, F'_{[t]}$ denote $\flatten{\projo(\gs)_{[t]}}$ and $\flatten{\projo(\gs')_{[t]}}$ respectively. Note that we only care to bound the node distance between projections for times $t\in[T]$ where $\gs_{[t]},\gs'_{[t]}$ are $(D,\ell)$-bounded. By \cref{lem:greedy-flatten}, we only need to bound $\dnode(F_{[t]}, F'_{[t]})$.

    Let $U$ be the set of nodes with degree greater than $D$ in $\gs'_{[t]}$. By \cref{lem:edges dont change}, only edges incident to $\vplus$ or to nodes with degree greater than $D$ in $\gs'_{[t]}$ will differ between flattened graphs $F_{[t]}$ and $F'_{[t]}$, so $U\cup \{ \vplus \}$ forms a vertex cover for these differing edges. This means we can use the following process to obtain $F_{[t]}$ from $F'_{[t]}$: first, delete $\vplus$ and all nodes in $U$ from $F'_{[t]}$; then, add all of the nodes in $U$ to this graph, but instead of adding the edges incident to these nodes in $F'_{[t]}$, add the edges incident to these nodes in $F_{[t]}$. We removed $\ell + 1$ nodes and added $\ell$ nodes, which gives us the desired bound of $\dnode(F_{[t]}, F'_{[t]})\leq 2\ell + 1$.    
\end{proof}

\fi

\subsection{Proof of Node-to-Edge Stability for \texorpdfstring{\boldmath{$\projo$}}{BBDS}}
\label{sec:BBDS node edge}

Here we prove item (2a) of \cref{thrm:combined-stab}, which we repeat below for convenience. The proof uses \cref{lem:greedy-flatten,lem:edges dont change}, though it requires a more careful analysis of the flattened graphs than\yoakland{ the proofs of items (1a) and (3) of \cref{thrm:combined-stab}.}\noakland{ the other stability proofs we have provided so far.}

\begin{theorem}[Item (2a) of \cref{thrm:combined-stab}]
\label{thrm:node-edge stab orig}
    Let $T\in\N$, $D\in \N$, $\ell\in \N\cup\{0\}$, and let $\projo$ be \cref{alg:time aware projection} with \inccrit{} $\style = \original$. If $\gs\nnode \gs'$ are node-neighboring graph streams of length $T$, then for all time steps $t\in[T]$ such that $S$ and $S'$ are $(D, \ell)$-bounded through time $t$, the edge distances between the projections through time $t$ satisfy
    \[
                \dedge\Big(\projo(\gs)_{[t]}\ ,\ 
                \projo(\gs')_{[t]} \Big)
                \leq
                D + \ell.
    \]   
\end{theorem}

\begin{proof}[Proof of \cref{thrm:node-edge stab orig}]
    Without loss of generality, let $\gs'$ be the larger graph stream---that is, it contains an additional node $\vplus$ and associated edges, which we represent with the set $E^+$.
    To simplify notation, let $F_{[t]}, F'_{[t]}$ denote $\flatten{\projo(\gs)_{[t]}}$ and $\flatten{\projo(\gs')_{[t]}}$ respectively. Note that we only care to bound the node distance between projections for times $t\in[T]$ where $\gs_{[t]},\gs'_{[t]}$ are $(D,\ell)$-bounded. By \cref{lem:greedy-flatten}, we only need to bound $\dedge(F_{[t]}, F'_{[t]})$.

    By \cref{lem:edges dont change}, only edges incident (1) to $\vplus$ or (2) to nodes with degree greater than $D$ in $\gs'_{[t]}$ will differ between flattened graphs $F_{[t]}$ and $F'_{[t]}$. We now count the number of edges in these categories that differ between $F_{[t]}$ and $F'_{[t]}$.

    We first bound the number of edges in (1). There are at most $D$ edges incident to $\vplus$ in $F'_{[t]}$ because $F'_{[t]}$ has maximum degree at most $D$, and none of these edges show up in $F_{[t]}$ since $\vplus$ is not in $F_{[t]}$. Therefore, there are at most $D$ edges in category (1).

    We next bound the number of edges in (2). Each flattened graph $F_{[t]}$ and $F'_{[t]}$ contains at most $D\cdot \ell$ edges incident to nodes with degree greater than $D$ in $\gs'_{[t]}$. However, we show that many of these edges will be the same in both flattened graphs. Let $U$ denote the set of nodes, excluding $\vplus$, that both have degree greater than $D$ in $S'_{[t]}$ and are incident to edges in $E^+$. Consider an edge $e=\{u,v\}$, where neither $u$ nor $v$ is in $U$. The added edges in $E^+$ will not affect the values of $d(u)$ or $d(v)$, so either $e$ will be in both $F_{[t]}$ and $F'_{[t]}$, or it will be in neither flattened graph.

    Now consider edges $e'$ incident to nodes in $U$. All edges in $E^+$ are of the form $e^+ = \{\vplus, w\}$. If $e^+$ \projstagealttexttwo{} where $d(w)\geq D$ in $\projo$ on $\gs'$, then $e^+$ will not cause any edges $e'$ to differ between $F_{[t]}$ and $F'_{[t]}$.
    However, if $e^+$ \projstagealttextthree, then there may be one edge $e_w$ incident to $w$ that appears in $F_{[t]}$ but does not appear in $F'_{[t]}$, due to having $d(w) = D$ \projstagealttextfour{$e_w$} instead of $d(w) = D-1 < D$ (as is the case when running $\projo$ on $\gs$).
    Note that $e^+$ will not affect the inclusion of other edges in the output stream: any edge incident to $w$ that \projstagealttexttwo{} after $e_w$ will appear in neither flattened graph, and any edge that \projstagealttexttwo{} prior to $e_w$ will have $d(w) < D$ when considered for inclusion in both output streams.

    We see that each edge $e=\{\vplus,w\}$ in $E^+$ causes at most one edge incident to $w$ to differ between projections (in particular, to appear in $F_{[t]}$ but not appear in $F'_{[t]}$). Therefore, if we bound $|E^+|$, we can bound the number of edges in category (2). Since all edges in $E^+$ are incident to $\vplus$ and a node in $U$, there are at most $|U|$ edges in $E^+$. By the fact that the streams $S_{[t]}$ and $S'_{[t]}$ are $(D,\ell)$-bounded, there are at most $\ell$ nodes in $U$. Therefore, the number of edges in category (2) is at most $\ell$. 

    Categories (1) and (2) contain a total of at most $D+\ell$ edges, so at most $D+\ell$ edges differ between $F_{[t]}$ and $F'_{[t]}$, which is what we wanted to show.
\end{proof}

\ifnum\oakland=0
\subsection{Proof of Node-to-Edge Stability for \texorpdfstring{\boldmath{$\projp$}}{DLL}}
\label{sec:node edge projected}

In this section, we prove item (2b) of \cref{thrm:combined-stab}, which we repeat below for convenience. This proof of stability requires a more involved analysis than other proofs, so we describe our proof approach below before moving to the formal proof.

\begin{theorem}[Item (2b) of \cref{thrm:combined-stab}]
\label{thrm:node-edge stab proj}
    Let $T\in\N$, $D\in \N$, $\ell\in \N\cup\{0\}$, and let $\projp$ be \cref{alg:time aware projection} with \inccrit{} $\style = \projection$. If $\gs\nnode \gs'$ are node-neighboring graph streams of length $T$, then for all time steps $t\in[T]$ such that $S$ and $S'$ are $(D, \ell)$-bounded through time $t$, the edge distances between the projections through time $t$ satisfy
    \[
                \dedge\Big(\projp(\gs)_{[t]}\ ,\ 
                \projp(\gs')_{[t]} \Big)
                \leq 
                \begin{cases}
                    D + 2\ell^{3/2} & \text{if } D\geq \ell,\ \text{and}\\
                    D + 2\ell \sqrt{D} & \text{if } D< \ell.\\
                \end{cases}
    \]
\end{theorem}

The proof of stability begins from the observation---as in our other stability proofs---that we only need to show stability of the flattened graphs (\cref{lem:greedy-flatten}). To show this stability of the flattened graphs, we construct a graph that contains exactly one edge for every edge that differs between projections, and we carefully label and direct each of these edges to create a \emph{difference graph} (see \cref{alg:diff graph}).

If the resulting difference graph were a DAG that satisfied the conditions of \cref{lem:dag ell 32}, then the difference graph would have at most $2\ell\sqrt{k}$ edges (where $\ell$ and $k$ are as defined in \cref{lem:dag ell 32}, and we set $k=\min\{D,\ell\}$). Unfortunately, the difference graph does not satisfy these conditions---but we can show that the difference graph is close to a graph that satisfies these conditions. More specifically, if we carefully remove at most $D$ edges from the difference graph, we can obtain a \emph{pruned difference graph} that has the special, highly structured form described in the conditions of \cref{lem:dag ell 32}. Therefore, the difference graph must contain at most $D + 2\ell \sqrt{\min\{D,\ell\}}$ edges, so the edge distance between projected graph streams is at most $D + 2\ell \sqrt{\min\{D,\ell\}}$.

Before presenting \cref{lem:dag ell 32} we define an \emph{\cut}, a method for creating cuts for a DAG.

\begin{definition}[Order-induced cut]
\label{defn:order induced cut}
    Let $G$ be a DAG on $k$ nodes, and let $v_1,\ldots,v_k$ be a topological ordering of its nodes from left to right (i.e., when the nodes are lined up horizontally in this order, all edges in $G$ go from left to right).

    Given this ordering, the \emph{$\ord{i}$ $\cut$ of $G$} is defined as the cut made by partitioning the nodes into the sets $\{v_1,\ldots,v_i\}$ and $\{v_{i+1},\ldots,v_k\}$. The \emph{$\ord{i}$ $\cut$ set of $G$} is defined as the edges in the cut set of the resulting cut. The cuts and cut sets are depicted in \cref{fig:natural-cuts}.
\end{definition}

In the proof below, we will begin by proving the result in \cref{lem:dag ell 32} about the number of edges in a highly structured DAG. We will then go about creating the difference graph, removing edges to obtain a pruned difference graph, and showing that the pruned difference graph shares the structure of the DAG whose edge count we upper bound. We conclude by upper bounding the number of edges that were removed to create the pruned difference graph.

\begin{figure*}
    \centering
    \scalebox{0.6}{\tikzset{every picture/.style={line width=0.75pt}} %

\begin{tikzpicture}[x=0.75pt,y=0.75pt,yscale=-1,xscale=1]

\draw [color={rgb, 255:red, 155; green, 155; blue, 155 }  ,draw opacity=1 ][fill={rgb, 255:red, 0; green, 0; blue, 0 }  ,fill opacity=0.09 ]   (201.07,2.02) -- (200.74,166.18) ;
\draw [color={rgb, 255:red, 155; green, 155; blue, 155 }  ,draw opacity=1 ][fill={rgb, 255:red, 0; green, 0; blue, 0 }  ,fill opacity=0.09 ]   (200.74,166.18) -- (200.41,330.34) ;

\draw [color={rgb, 255:red, 155; green, 155; blue, 155 }  ,draw opacity=1 ][fill={rgb, 255:red, 0; green, 0; blue, 0 }  ,fill opacity=0.09 ]   (100.73,1) -- (100.4,165.17) ;
\draw [color={rgb, 255:red, 155; green, 155; blue, 155 }  ,draw opacity=1 ][fill={rgb, 255:red, 0; green, 0; blue, 0 }  ,fill opacity=0.09 ]   (100.4,165.17) -- (100.07,329.33) ;

\draw [color={rgb, 255:red, 155; green, 155; blue, 155 }  ,draw opacity=1 ][fill={rgb, 255:red, 0; green, 0; blue, 0 }  ,fill opacity=0.09 ]   (301.42,2.02) -- (301.09,166.18) ;
\draw [color={rgb, 255:red, 155; green, 155; blue, 155 }  ,draw opacity=1 ][fill={rgb, 255:red, 0; green, 0; blue, 0 }  ,fill opacity=0.09 ]   (301.09,166.18) -- (300.76,330.34) ;

\draw [color={rgb, 255:red, 155; green, 155; blue, 155 }  ,draw opacity=1 ][fill={rgb, 255:red, 0; green, 0; blue, 0 }  ,fill opacity=0.09 ]   (400.8,1) -- (400.47,165.17) ;
\draw [color={rgb, 255:red, 155; green, 155; blue, 155 }  ,draw opacity=1 ][fill={rgb, 255:red, 0; green, 0; blue, 0 }  ,fill opacity=0.09 ]   (400.47,165.17) -- (400.14,329.33) ;

\draw [color={rgb, 255:red, 155; green, 155; blue, 155 }  ,draw opacity=1 ][fill={rgb, 255:red, 0; green, 0; blue, 0 }  ,fill opacity=0.09 ]   (501.14,2.02) -- (500.81,166.18) ;
\draw [color={rgb, 255:red, 155; green, 155; blue, 155 }  ,draw opacity=1 ][fill={rgb, 255:red, 0; green, 0; blue, 0 }  ,fill opacity=0.09 ]   (500.81,166.18) -- (500.48,330.34) ;

\draw [color={rgb, 255:red, 155; green, 155; blue, 155 }  ,draw opacity=1 ][fill={rgb, 255:red, 155; green, 155; blue, 155 }  ,fill opacity=0.62 ][line width=0.75]    (600.52,2.02) -- (600.19,166.18) ;
\draw [color={rgb, 255:red, 155; green, 155; blue, 155 }  ,draw opacity=1 ][fill={rgb, 255:red, 155; green, 155; blue, 155 }  ,fill opacity=0.62 ][line width=0.75]    (600.19,166.18) -- (599.86,330.34) ;

\draw  [line width=2.25]  (54.36,129.53) .. controls (68.17,129.53) and (79.36,140.73) .. (79.36,154.54) .. controls (79.36,168.35) and (68.17,179.54) .. (54.36,179.54) .. controls (40.55,179.54) and (29.35,168.35) .. (29.35,154.54) .. controls (29.35,140.73) and (40.55,129.53) .. (54.36,129.53) -- cycle ;
\draw  [line width=2.25]  (154.38,129.53) .. controls (168.19,129.53) and (179.39,140.73) .. (179.39,154.54) .. controls (179.39,168.35) and (168.19,179.54) .. (154.38,179.54) .. controls (140.57,179.54) and (129.38,168.35) .. (129.38,154.54) .. controls (129.38,140.73) and (140.57,129.53) .. (154.38,129.53) -- cycle ;
\draw  [line width=2.25]  (254.4,129.53) .. controls (268.21,129.53) and (279.41,140.73) .. (279.41,154.54) .. controls (279.41,168.35) and (268.21,179.54) .. (254.4,179.54) .. controls (240.59,179.54) and (229.4,168.35) .. (229.4,154.54) .. controls (229.4,140.73) and (240.59,129.53) .. (254.4,129.53) -- cycle ;
\draw  [line width=2.25]  (354.42,129.53) .. controls (368.23,129.53) and (379.43,140.73) .. (379.43,154.54) .. controls (379.43,168.35) and (368.23,179.54) .. (354.42,179.54) .. controls (340.61,179.54) and (329.42,168.35) .. (329.42,154.54) .. controls (329.42,140.73) and (340.61,129.53) .. (354.42,129.53) -- cycle ;
\draw  [line width=2.25]  (454.45,129.53) .. controls (468.26,129.53) and (479.45,140.73) .. (479.45,154.54) .. controls (479.45,168.35) and (468.26,179.54) .. (454.45,179.54) .. controls (440.64,179.54) and (429.44,168.35) .. (429.44,154.54) .. controls (429.44,140.73) and (440.64,129.53) .. (454.45,129.53) -- cycle ;
\draw  [line width=2.25]  (554.47,129.53) .. controls (568.28,129.53) and (579.47,140.73) .. (579.47,154.54) .. controls (579.47,168.35) and (568.28,179.54) .. (554.47,179.54) .. controls (540.66,179.54) and (529.46,168.35) .. (529.46,154.54) .. controls (529.46,140.73) and (540.66,129.53) .. (554.47,129.53) -- cycle ;
\draw  [line width=2.25]  (654.49,129.53) .. controls (668.3,129.53) and (679.5,140.73) .. (679.5,154.54) .. controls (679.5,168.35) and (668.3,179.54) .. (654.49,179.54) .. controls (640.68,179.54) and (629.48,168.35) .. (629.48,154.54) .. controls (629.48,140.73) and (640.68,129.53) .. (654.49,129.53) -- cycle ;

\draw [line width=2.25]    (154.38,179.54) .. controls (168.66,287.43) and (456.6,333.62) .. (553.03,181.85) ;
\draw [shift={(554.47,179.54)}, rotate = 121.57] [fill={rgb, 255:red, 0; green, 0; blue, 0 }  ][line width=0.08]  [draw opacity=0] (14.29,-6.86) -- (0,0) -- (14.29,6.86) -- cycle    ;
\draw [line width=2.25]    (173.92,138.14) .. controls (184.69,119.3) and (288.26,54.05) .. (351.56,126.14) ;
\draw [shift={(354.42,129.53)}, rotate = 230.89] [fill={rgb, 255:red, 0; green, 0; blue, 0 }  ][line width=0.08]  [draw opacity=0] (14.29,-6.86) -- (0,0) -- (14.29,6.86) -- cycle    ;
\draw [line width=2.25]    (154.38,129.53) .. controls (154.79,106.83) and (321.18,-25.86) .. (452.47,127.2) ;
\draw [shift={(454.45,129.53)}, rotate = 229.98] [fill={rgb, 255:red, 0; green, 0; blue, 0 }  ][line width=0.08]  [draw opacity=0] (14.29,-6.86) -- (0,0) -- (14.29,6.86) -- cycle    ;
\draw [line width=2.25]    (79.36,154.54) -- (124.38,154.54) ;
\draw [shift={(129.38,154.54)}, rotate = 180] [fill={rgb, 255:red, 0; green, 0; blue, 0 }  ][line width=0.08]  [draw opacity=0] (14.29,-6.86) -- (0,0) -- (14.29,6.86) -- cycle    ;
\draw [line width=2.25]    (279.41,154.54) -- (324.42,154.54) ;
\draw [shift={(329.42,154.54)}, rotate = 180] [fill={rgb, 255:red, 0; green, 0; blue, 0 }  ][line width=0.08]  [draw opacity=0] (14.29,-6.86) -- (0,0) -- (14.29,6.86) -- cycle    ;
\draw [line width=2.25]    (354.42,179.54) .. controls (356.98,208.75) and (447.17,241.01) .. (530.84,174.06) ;
\draw [shift={(534.65,170.94)}, rotate = 140.01] [fill={rgb, 255:red, 0; green, 0; blue, 0 }  ][line width=0.08]  [draw opacity=0] (14.29,-6.86) -- (0,0) -- (14.29,6.86) -- cycle    ;
\draw [line width=2.25]    (479.45,154.54) -- (524.46,154.54) ;
\draw [shift={(529.46,154.54)}, rotate = 180] [fill={rgb, 255:red, 0; green, 0; blue, 0 }  ][line width=0.08]  [draw opacity=0] (14.29,-6.86) -- (0,0) -- (14.29,6.86) -- cycle    ;
\draw [line width=2.25]    (579.47,154.54) -- (624.48,154.54) ;
\draw [shift={(629.48,154.54)}, rotate = 180] [fill={rgb, 255:red, 0; green, 0; blue, 0 }  ][line width=0.08]  [draw opacity=0] (14.29,-6.86) -- (0,0) -- (14.29,6.86) -- cycle    ;
\end{tikzpicture}}
    \caption{Order-induced cuts. Let the nodes in the figure represent a left-to-right topological ordering of nodes. Each vertical line defines a distinct order-induced cut, where nodes to the left are in one part of a cut, nodes to the right are in the other part of the cut, and the cut-set for the cut defined by a vertical line is the set of edges intersected by that vertical line.}
    \label{fig:natural-cuts}
\end{figure*}

\begin{lemma}[Number of edges in a structured DAG]
\label{lem:dag ell 32}
    Let $k,\ell\in\N\cup \{0\}$. A DAG $G$ with both of the following properties has at most $2\ell\sqrt{k}$ edges:
    \begin{enumerate}
        \item $G$ has at most $\ell + 1$ nodes, and
        \item every order-induced cut set of $G$ has at most $k$ edges.
    \end{enumerate}    
\end{lemma}
\begin{proof}[Proof of \cref{lem:dag ell 32}]
    Let $G$ be a DAG with the properties described in the statement of \cref{lem:dag ell 32} above, and let $k^*\leq k$ and $\ell^* \leq \ell$ be an upper bound on the number of edges in $\cut$ sets in the DAG and one less than the number of nodes in the DAG, respectively. Because $G$ is a DAG, it must have a topological ordering---that is, there is a way to order its nodes horizontally from left to right such that all edges in the graph go from left to right. We enumerate these sorted nodes from left to right as $v_1,\ldots,v_{\ell^* + 1}$.
    
    Throughout this proof,
    we define the \emph{length of edge $e = (v_j,v_{j'})$} in the DAG as $j'-j$.
    To bound the number of edges in this DAG, we bound the number of ``short'' edges and the number of ``long'' edges. More precisely, let $F$ be the set of edges in the DAG, which we separate into the following two sets according to edge length: $F_0$, the set of edges with length at most $\sqrt{k^*}$; and $F_1$, the set of edges with length greater than $\sqrt{k^*}$.

    We first bound the cardinality of $F_0$. All of the edges in the DAG go from left to right and there are $\ell^*+1$ nodes, so there are $\ell^*$ starting points for edges, and there are at most $\sqrt{k^*}$ choices of endpoints for edges of length at most $\sqrt{k^*}$ from each starting point. This gives us
    \[
        |F_0| \leq \ell^* \sqrt{k^*}.
    \]

    We next bound the cardinality of $F_1$. We use a probabilistic argument. Let $F_1^s$ denote the set of edges in $F_1$ that are in the $\ord{s}$ $\cut$ set. In other words, $F_1^s$ denotes the number of edges in the $\ord{s}$ $\cut$ set that have length greater than $\sqrt{k^*}$. We now compute the expected cardinality of $F_1^s$ for a uniformly random choice of $s\in[\ell^*]$ (since there are $\ell^*$ locations at which we could make a cut), and we relate it to the cardinality of $F_1$ to get a bound for $|F_1|$. Note that we have $|F_1^s| \leq k^*$ by condition (2) of \cref{lem:dag ell 32}. This gives us the following series of inequalities:
    \begin{align*}
        k^*
        &\geq \Ex_s \left[ |F_1^s| \right] \\
        &= \sum_{e\in F_1} \Pr_s\left[ \text{$e\in F_1^s$} \right] \tag{law of total expectation} \\
        &> |F_1| \cdot \frac{\sqrt{k^*}}{\ell^*} \tag{$|F_1|$ edges, each in the cut set w.p. greater than $\frac{\sqrt{k^*}}{\ell^*}$}.
    \end{align*}
    Solving for $|F_1|$ gives us
    \(
        |F_1| < \ell^*\sqrt{k^*}.
    \)

    We now combine our bounds on $|F_0|$ and $|F_1|$ to get
    \[
        |F| = |F_0| + |F_1| \leq 2\ell^*\sqrt{k^*} \leq 2\ell \sqrt{k},
    \]
    which is what we wanted to show.
\end{proof}

\subsubsection{(Pruned) Difference Graphs and Their Properties}
\label{sec:difference graphs}

We next show that the pruned difference graph has the structure described in \cref{lem:dag ell 32}. We begin by defining a difference graph. The algorithm for creating a difference graph is provided in \cref{alg:diff graph}. On an input of two node-neighboring graph streams $\gs\nnode \gs'$, it returns a graph that encodes the edges that differ between the projections of $\gs$ and $\gs'$. The resulting \emph{difference graph} will consist of directed edges, each colored red or blue and labeled with an integer.

\begin{figure*}[ht!]
\begin{mdframed}
    \centering
    \tikzset{every picture/.style={line width=0.75pt}}        

\begin{tikzpicture}[x=0.75pt,y=0.75pt,yscale=-1,xscale=1]
\draw [color={rgb,255:red,0;green,0;blue,0}, draw opacity=1] 
  (71.26,17.57) .. controls (75.4,17.57) and (78.76,20.93) .. (78.76,25.07) ..
  controls (78.76,29.21) and (75.4,32.57) .. (71.26,32.57) ..
  controls (67.11,32.57) and (63.76,29.21) .. (63.76,25.07) ..
  controls (63.76,20.93) and (67.11,17.57) .. (71.26,17.57) -- cycle ;

\draw [color={rgb,255:red,0;green,0;blue,0}, draw opacity=1] 
  (121.03,17.35) .. controls (125.18,17.35) and (128.53,20.71) .. (128.53,24.85) ..
  controls (128.53,28.99) and (125.18,32.35) .. (121.03,32.35) ..
  controls (116.89,32.35) and (113.53,28.99) .. (113.53,24.85) ..
  controls (113.53,20.71) and (116.89,17.35) .. (121.03,17.35) -- cycle ;

\draw [color={rgb,255:red,252;green,2;blue,2}, draw opacity=1]   
  (78.76,25.07) -- node[above,midway,yshift=2pt,color=black]{$i$} (113.53,24.85);
\draw [shift={(113.53,24.85)}, rotate=179.63] 
  [color={rgb,255:red,252;green,2;blue,2}, draw opacity=1, line width=0.75] 
  (10.93,-3.29) .. controls (6.95,-1.4) and (3.31,-0.3) .. (0,0) ..
  controls (3.31,0.3) and (6.95,1.4) .. (10.93,3.29);

\draw [color={rgb,255:red,0;green,0;blue,0}, draw opacity=1] 
  (71.26,70.57) .. controls (75.4,70.57) and (78.76,73.92) .. (78.76,78.07) ..
  controls (78.76,82.21) and (75.4,85.57) .. (71.26,85.57) ..
  controls (67.11,85.57) and (63.76,82.21) .. (63.76,78.07) ..
  controls (63.76,73.92) and (67.11,70.57) .. (71.26,70.57) -- cycle ;

\draw [color={rgb,255:red,0;green,0;blue,0}, draw opacity=1] 
  (121.03,70.34) .. controls (125.18,70.34) and (128.53,73.7) .. (128.53,77.84) ..
  controls (128.53,81.99) and (125.18,85.34) .. (121.03,85.34) ..
  controls (116.89,85.34) and (113.53,81.99) .. (113.53,77.84) ..
  controls (113.53,73.7) and (116.89,70.34) .. (121.03,70.34) -- cycle ;

\draw [color={rgb,255:red,18;green,3;blue,253}, draw opacity=1]   
  (78.76,78.07) -- node[above,midway,yshift=2pt,color=black]{$i$} (113.53,77.84);
\draw [shift={(113.53,77.84)}, rotate=179.63] 
  [color={rgb,255:red,18;green,3;blue,253}, draw opacity=1, line width=0.75] 
  (10.93,-3.29) .. controls (6.95,-1.4) and (3.31,-0.3) .. (0,0) ..
  controls (3.31,0.3) and (6.95,1.4) .. (10.93,3.29);

\draw (311,145) node[text width=400pt, align=justify] 
{The direction of $(u,v)$ indicates that, in the stream whose projection is missing the edge, $u$ saturated before $v$ (and, since $\{u,v\}$ was not included, before $\{u,v\}$ was considered).};

\draw (65.8,73.3) node[anchor=north west, inner sep=0.75pt]{$u$};
\draw (116.13,73.3) node[anchor=north west, inner sep=0.75pt]{$v$};
\draw (66,20.3) node[anchor=north west, inner sep=0.75pt]{$u$};
\draw (116.13,20.3) node[anchor=north west, inner sep=0.75pt]{$v$};

\draw (137.7,69.79) node[anchor=north west, inner sep=0.75pt, text width=250pt, align=justify]
{an edge not in the projection of $\gs$, considered at projection stage $i$, that is \textcolor[rgb]{0.05,0.08,0.96}{\textit{added (blue)}} in the projection of $\gs'$. \\$(\{u,v\}\notin\projp(\gs), \{u,v\}\in\projp(\gs'))$};

\draw (137.7,16.24) node[anchor=north west, inner sep=0.75pt, text width=250pt, align=justify]
{an edge in the projection of $\gs$, considered at projection stage $i$, that is \textcolor[rgb]{1,0,0}{\textit{deleted (red)}} in the projection of $\gs'$. \\$(\{u,v\}\in\projp(\gs), \{u,v\}\notin\projp(\gs'))$};

\end{tikzpicture}
\end{mdframed}
    \caption{Edges in a difference graph}
    \label{fig:difference-graph-edges}
\end{figure*}

As shown in \cref{fig:difference-graph-edges}, red edges in the difference graph represent edges that appear in $\Pi_D(\gs)$ and not $\Pi_D(\gs')$, while blue edges represent edges that appear in $\Pi_D(\gs')$ and not $\Pi_D(\gs)$.
The label on the edge corresponds to when the projection algorithm considers it for addition.
An edge $\{u,v\}$ is directed as $(u,v)$ if $u$ is saturated before $v$ in the projection missing the edge, and is otherwise directed as $(v,u)$; the details on directing edges are provided in \cref{alg:diff graph}.
All edges that differ between projections are contained in the difference graph (\cref{lem:dg has all}).

\begin{algorithm}[ht!]
    \caption{Algorithm $\dg$ for difference graphs.}
    \label{alg:diff graph}
    \begin{algorithmic}[1]
        \Statex \textbf{Input:} Node-neighboring graph streams $\gs\simeq\gs'$ of length $T$, degree bound $D\in\N$, time $t\in [T]$. Let $\projp$ be \cref{alg:time aware projection} with $\style = \projection$. (Without loss of generality, assume that $\gs'$ is the ``larger'' graph stream with an additional node $\vplus$.)
        \Statex \textbf{Output:} Colored, directed graph $\gdiff$ with labeled edges, where $V(\gdiff) = V(\gs')$.
        \State $\Eblue, \Ered \gets \emptyset$
        \State $F \gets \flatten{\projp(\gs)_{[t]}}$, $F' \gets \flatten{\projp(\gs')_{[t]}}$
        \State $\satstep_{\gs}(\vplus)\gets -1$ \label{line:satstep vplus}
        \yoakland{\State} \Comment{\commentstyle{Redefine $\satstep_{\gs}(\vplus)$ to evaluate to $-1$}}

        \For{projection stage $i$ of $\projp(S')$}
            \State Consider edge $e = \{u,v\}$ processed during projection stage $i$
            \If{$e\in F \cap F'$} ignore $e$\label{line:ignore edge}
            \ElsIf{$e\in F$}
            \yoakland{\State} \Comment{\commentstyle{Directing red edges}}\label{line:add edge F}
                \If{$\satstep_{\gs'}(u) < \satstep_{\gs'}(v)$}\label{line:direct red} 
                    \State add $(u,v)$ with label $i$ to $\Ered$ 
                \Else
                    \State add $(v,u)$ with label $i$ to $\Ered$ 
                \EndIf
            \ElsIf{$e\in F'$}
            \yoakland{\State} \Comment{\commentstyle{Directing blue edges}}\label{line:add edge F prime}
                \If{$\satstep_{\gs}(u) < \satstep_{\gs}(v)$} \label{line:direct blue} 
                    \State add $(u,v)$ with label $i$ to $\Eblue$
                \Else
                    \State add $(v,u)$ with label $i$ to $\Eblue$
                \EndIf
            \EndIf
        \EndFor
        \State \Return $V(F'), (\Eblue, \Ered)$
    \end{algorithmic}
\end{algorithm}

\begin{lemma}[The difference graph contains all differing edges]
\label{lem:dg has all}
    Let $\gs\nnode \gs'$ be a pair of node-neighboring graph streams of length $T$, let $t\in[T]$, and let $\projp$ be \cref{alg:time aware projection} with \inccrit{} $\style = \projection$. An edge differs between $\projp(\gs)_{[t]}$ and $\projp(\gs')_{[t]}$ if and only if the edge is in $\dg((\gs,\gs'),D,t)$. 
\end{lemma}

\begin{proof}[Proof of \cref{lem:dg has all}]
     By \cref{lem:greedy-flatten}, we only need to consider how the flattened graphs differ.

     We first show the $\Rightarrow$ direction. By Line~\ref{line:ignore edge} of \cref{alg:diff graph}, we only ignore an edge (i.e., don't add it to the difference graph) if it appears in both flattened graphs. Therefore, if an edge differs between projections of node neighbors, then it appears in the difference graph, which is what we wanted to show.
     
     We next show the $\Leftarrow$ direction. By Line~\ref{line:ignore edge} edges in both flattened graphs are ignored, and by Lines~\ref{line:add edge F} and \ref{line:add edge F prime} edges in one flattened graph but not the other are added to the difference graph.
\end{proof}

\begin{remark}[Interpreting a difference graph]
    \cref{alg:diff graph} returns three sets: a set $V(\gdiff)$ of vertices; a set $\Eblue$ of labelled, directed, blue edges; and a set $\Ered$ of labelled, directed, red edges. A difference graph is the directed graph formed on the vertices in $V(\gdiff)$, where edges in $\Eblue$ are colored blue and labelled according to their labels, and edges in $\Ered$ are colored red and labelled according to their labels.
\end{remark}

An additional property, which we prove in \cref{prop:outgoing same color}, is that all out-edges from a node $v$ in a difference graph have the same color. For a node $v$, we let $\outcolor(v)$ denote this color of its out-edges.

\begin{definition}[Out color of a node]
\label{defn:outcolor}
    Let $G$ be a difference graph, and for edge $e$ in the difference graph, let $\ecolor(e)$ denote its color. All out-edges $e$ from a node in a difference graph have the same color $\ecolor(e)$ (\cref{prop:outgoing same color}). The \emph{out color} of a node $v$ in $G$, denoted $\outcolor(v)$, is the color of the out-edges from $v$.
\end{definition}

Now that we have defined a difference graph, we can present the definition of a \emph{pruned difference graph} (\cref{alg:pruned diff graph}), which we will later show has the same structure as the DAG described in \cref{lem:dag ell 32}. A pruned difference graph removes two types of edges from a difference graph: every edge that (1) is an in-edge to a node that, in the larger graph stream $\gs'_{[t]}$, has degree at most $D$ or (2) is an in-edge to a node $v$ and itself has color $\outcolor(v)$ (all out-edges from a node in a difference graph have the same color---see \cref{prop:outgoing same color}). We also remove all nodes in the difference graph that, after removing edges to create the pruned difference graph, have degree 0. 

\begin{algorithm}[ht!]
    \caption{Algorithm $\pdg$ for pruned difference graphs.}
    \label{alg:pruned diff graph}
    \begin{algorithmic}[1]
        \Statex \textbf{Input:} Difference graph $\gdiff$ returned by \cref{alg:diff graph}, time $t$ and graph stream $\gs'$ input to \cref{alg:diff graph}.
        \Statex \textbf{Output:} Colored, directed graph with labeled edges.
        \State Parse difference graph $\gdiff$ as $V,(\Eblue,\Ered)$
        \State $F' \gets \flatten{\gs'_{[t]}}$
        \For{edge $e = (u,v)$ in $\Eblue \cup \Ered$}
            \If{$\ecolor(e) = \outcolor(v)$} \Comment{see \cref{defn:outcolor}.}
            \label{line:remove same color edge}
                \State Remove $e$ from $E_{\ecolor(e)}$ 
            \EndIf
            \If{$\deg_{v}(F')\leq D$}\label{line:remove low deg}
                \State Remove $e$ from $E_{\ecolor(e)}$
            \EndIf
        \EndFor
        \State $F_\mathit{prune} = \flatten{V, \Eblue \cup \Ered}$
        \For{$v$ in $V$}
            \If{$\deg_v(F_\mathit{prune}) = 0$} remove $v$ from $V$\label{line:remove nodes}
            \EndIf
        \EndFor
        \State Return $V, (\Eblue, \Ered)$
    \end{algorithmic}
\end{algorithm}

To show that \cref{lem:dag ell 32} applies to a pruned difference graph produced by \cref{alg:pruned diff graph}, we need to prove the first three items of \cref{lemma:pdg is dag etc} below; the fourth item shows that the pruned difference graph differs in few edges from the original difference graph.

\begin{lemma}
\label{lemma:pdg is dag etc}
    Let $T,D\in\N$. Let $\gs\nnode \gs'$ be a pair of node-neighboring graph streams of length $T$, where without loss of generality $\gs'$ has one additional node $v^+$ (and its associated edges) as compared to $\gs$, and let be $(D,\ell)$-bounded through time $t\in[T]$. Let $\gdiff$ be the difference graph output by \cref{alg:diff graph} when run with degree bound $D$ and time $t$ on graph streams $\gs,\gs'$, and let $\gdiff^*$ be the pruned difference graph returned by \cref{alg:pruned diff graph} when run on $\gdiff, \gs'$, and $t$. 
    
    The pruned difference graph $\gdiff^*$ has the following properties:
    \begin{enumerate}
        \item $\gdiff^*$ is a DAG.\label{lem-item:pdg is a dag}
        \item $\gdiff^*$ has at most $\ell + 1$ nodes.
        \item Every $\cut$ in $\gdiff^*$ contains at most $\min\{D,\ell\}$ edges.
        \item At most $D$ edges were removed from $\gdiff$ to make $\gdiff^*$.
    \end{enumerate}
\end{lemma}

Before proving \cref{lemma:pdg is dag etc}, we show how its result can be used to complete the proof of \cref{thrm:node-edge stab proj}. We then use the remainder of \cref{sec:node edge projected} to prove \cref{lemma:pdg is dag etc}.

\begin{proof}[Proof of \cref{thrm:node-edge stab proj}]
    By items (1), (2), and (3) of \cref{lemma:pdg is dag etc}, we see that the pruned difference graph satisfies the conditions of \cref{lem:dag ell 32}. Therefore, the pruned difference graph has at most $2\ell \sqrt{\min\set{D,\ell}}$ edges. Item (4) of \cref{lemma:pdg is dag etc} tells us that the pruned difference graph differs from the difference graph on at most $D$ edges. By \cref{lem:dg has all} the difference graph contains all differing edges, which means that at most $D + 2\ell \sqrt{\min\set{D,\ell}}$ edges differ between the projections, through time $t$, of node-neighboring graph streams that are $(D,\ell)$-bounded through time $t$. This is what we wanted to show.
\end{proof}

\begin{proof}[Proof of \cref{lemma:pdg is dag etc}] We prove each part of \cref{lemma:pdg is dag etc} below.

\textbf{Proof of item (1).} We will show that, in the pruned difference graph $\gdiff^*$, the labels on all in-edges to nodes precede the labels on all out-edges from nodes. This means that the graph is acyclic, which will complete the proof of item (1). We first show that all out-edges from a node $v$ have the same color.

\begin{claim}[Out-edges have the same color]
\label{prop:outgoing same color}
    Let $\gdiff$ be as defined in \cref{lemma:pdg is dag etc}. For every node $v$ in $\gdiff$, all of its out-edges must have the same color.
\end{claim}
\begin{proof}[Proof of \cref{prop:outgoing same color}]
Let $\gs$ and $\gs'$ be the node-neighboring graph streams described in \cref{lemma:pdg is dag etc}.
Assume, for the sake of contradiction, that a node in $\gdiff$ has out-edges with different colors. More precisely, assume that node $u$ in $\gdiff$ has red edge $e_\red = (u,v)$ and $e_\blue = (u,w)$.

There are two cases to consider: $u = \vplus$ and $u\neq \vplus$. If $u = \vplus$, then the claim follows immediately since all edges incident to $\vplus$ are in $S'$ and not in $S$, and since $\satstep_{\gs}(\vplus)$ is less than $\satstep_{\gs}$ for all other nodes by Line~\ref{line:satstep vplus}, so all edges incident to $\vplus$ in $\gdiff$ must be blue out-edges.

Now consider the case where $u\neq \vplus$. The existence of $e_\red$ means the edge $\{u,v\}$ is in $\gs$ and not in $\gs'$, so we know that either $\satstep_{\gs'}(u) < \satstep_{\gs}(u)$ or $\satstep_{\gs'}(v) < \satstep_{\gs}(u)$. By the direction of $e_\red$ and Line~\ref{line:direct red}, we have that $\satstep_{\gs'}(u) < \satstep_{\gs'}(v)$, which tells us $\satstep_{\gs'}(u) < \satstep_{\gs}(u)$.

By a similar argument, the existence of $e_\blue$ tells us $\satstep_{\gs}(u) < \satstep_{\gs'}(u)$. However, this pair of inequalities is a contradiction, so our assumption must be false, which completes the proof.
\end{proof}

By \cref{prop:outgoing same color}, we know that all out-edges have the same color. Additionally, by construction of the pruned difference graph (in particular, Line~\ref{line:remove same color edge} or \cref{alg:pruned diff graph}), we see that if a node has out-edges, then all its in-edges have the opposite color. This leaves us with two cases for showing that the labels on all in-edges to every node $v$ in $\gdiff^*$ precede the labels on all out-edges from $v$. We only consider the case where $v\neq \vplus$ since the claim follows trivially for $\vplus$ as there are no in-edges to $\vplus$.
\begin{enumerate}[label=(\alph*)]
    \item \textbf{All out-edges from \bm{$v$} are red.} Let $b$ be the smallest label on red out-edges from $v$, which means $\satstep_{\gs'}(v) < b$. Assume for contradiction that a blue in-edge to $v$ has label $b' > b$ (we are working with the pruned difference graph and $\outcolor(v) = \red$, so all in-edges to $v$ must be blue). The label on this blue edge tells us $\satstep_{\gs'}(v) \geq b'$, which is a contradiction since $b'>b$ and $\satstep_{\gs'}(v) < b$.

    \item \textbf{All out-edges from \bm{$v$} are blue.} The proof of this case follows very similarly, except we look at $\satstep_{\gs}(v)$. Let $b$ be the smallest label on blue out-edges from $v$, which means $\satstep_{\gs}(v) < b$. Assume for contradiction that a red in-edge to $v$ has label $b' > b$. The label on this red edge tells us $\satstep_{\gs}(v) \geq b'$, which is a contradiction since $b'>b$ and $\satstep_{\gs}(v) < b$.
\end{enumerate}
For both cases, we showed that all labels on in-edges to a node $v$ must precede labels on all out-edges from $v$, which completes the proof of item (1).

\textbf{Proof of item (2).} Item (2) follows readily from previously used ideas. To prove this item, we show that the only nodes appearing in the pruned difference graph $\gdiff^*$ are $\vplus$ and nodes $v$ such that $\deg_v(\flatt{\gs'_{[t]}}) > D$. More precisely, we let $V_\mathit{low}$ be the set of nodes $v$ in $\gs'$ that are not $\vplus$ and have $\deg_v(\flatt{\gs'_{[t]}}) \leq D$, and we show that no nodes $v\in V_\mathit{low}$ are in the pruned difference graph. There are at most $\ell$ nodes $v$ such that $\deg_v(\flatt{\gs'_{[t]}}) > D$, so the set of nodes not in $V_\mathit{low}$ has size at most $\ell + 1$. This proof relies on \cref{property:no out edge if small degree}.

\begin{claim}
\label{property:no out edge if small degree}
    For all $v\neq v^+$ in $\gdiff$, $v$ cannot have an out-edge in $\gdiff$ if $\deg_v(\flatt{\gs'_{[t]}}) \leq D$.
\end{claim}
\begin{proof}[Proof of \cref{property:no out edge if small degree}]
    Let $v\in V_\mathit{low}$, and let $b_\mathit{last}$ be the \projstagetext{} of the final edge in $S'_{[t]}$ incident to $v$. Because $\deg_v(\flatt{\gs'_{[t]}}) \leq D$, we know that the \satsteptext{} of $v$ in $\gs$ and $\gs'$ must be at least $b_\mathit{last}$, so $\satstep_{\gs}(v) \geq b_\mathit{last}$ and $\satstep_{\gs'}(v) \geq b_\mathit{last}$.
    
    Assume for contradiction that there exists an out-edge $e = (v,w)$ with label $b$ in $\gdiff$. Note that every edge incident to $v$ in $\gdiff$ must have label at most $b_\mathit{last}$, so $b\leq b_\mathit{last}$. If this edge is red, then $\satstep_{\gs'}(v) < b\leq b_\mathit{last}$. Similarly, if this edge is blue, then $\satstep_{\gs}(v) < b\leq b_\mathit{last}$. However, this contradicts the fact that $\satstep_{\gs}(v) \geq b_\mathit{last}$ and $\satstep_{\gs'}(v) \geq b_\mathit{last}$. 
\end{proof}

By \cref{property:no out edge if small degree}, every $v\in V_\mathit{low}$ has no out-edges in $\gdiff$. We next note that, in the pruned difference graph $\gdiff^*$, all in-edges to nodes in $V_\mathit{low}$ are removed by Line~\ref{line:remove low deg} of \cref{alg:pruned diff graph}. Since all nodes with in-degree and out-degree 0 in the pruned difference graph are removed in Line~\ref{line:remove nodes}, only the nodes in the set of size $\ell + 1$ (i.e., the node $\vplus$ and the high-degree nodes $v$ such that $\deg_v(\flatt{\gs'_{[t]}}) > D$) are in the pruned difference graph.

\textbf{Proof of item (3).} To prove this item, we use \cref{claim:out edges at most in edges}, which relates the in-degree of a node in the original difference graph to its out-degree in the pruned difference graph. 

\begin{claim}
\label{claim:out edges at most in edges}
    Consider a node $v\neq v^+$ in the pruned difference graph $\gdiff^*$. The number of out-edges from $v$ in the \emph{original} difference graph $\gdiff$ is at most the number of in-edges to $v$ in the \emph{pruned} difference graph $\gdiff^*$ (if $v$ does not appear in $\gdiff^*$, we say $v$ has zero in-edges). This also means that the number of out-edges from $v$ in $\gdiff^*$ is at most the number of in-edges to $v$ in $\gdiff^*$. In summary,
    \newcommand{\outdegree}{\msf{out}\text{-}\msf{degree}}
    \newcommand{\indegree}{\msf{in}\text{-}\msf{degree}}
    \[
        \Centerstack{\text{$\bm{\outdegree}$ of } \text{$v$ in $\bm{\gdiff^*}$ }}\; 
        \text{\large $\bm\leq$ }\;
        \Centerstack{\text{$\bm{\outdegree}$ of } \text{$v$ in $\bm{\gdiff}$ }}\; 
        \text{\large $\bm\leq$ }\;
        \Centerstack{\text{$\bm{\indegree}$ of } \text{$v$ in $\bm{\gdiff^*}$.}}
    \]
\end{claim}
    
\begin{proof}[Proof of \cref{claim:out edges at most in edges}]
    Consider a node $v\neq v^+$ in the graphs $\gdiff,\gdiff^*$. Let $k$ be the number of out-edges from $v$ in $\gdiff$, and let $b$ be their smallest label. Let $\outcolor(v) = \red$ (a symmetric argument follows when $\outcolor(v) = \blue$).
    
    By the construction of $\gdiff$, and since $\outcolor(v) = \red$, the out-edges of $v$ are in $\projp(\gs)_{[t]}$ and not in $\projp(\gs')_{[t]}$. Additionally, since the smallest of their labels is $b$, then after \projstagetext{} $b-1$ the following is true: (1) the degree of $v$ in the projection of $\gs'$ is $D$, and (2) the degree of $v$ in the projection of $\gs$ is at most $D-k$ (because if the degree were larger, then $v$ could not have $k$ red out-edges). This means that after \projstagetext{} $b-1$, the vertex $v$ has at least $k$ incident blue edges in the difference graph $\gdiff$.
    
    Finally, all out-edges from $v$ are red and must all have the same color (\cref{prop:outgoing same color}), so these incident blue edges must be in-edges to $v$. Additionally, because $\deg_v(\flatt{\gs'_{[t]}}) > D$ and because $\outcolor(v) = \red$, these blue edges will also be present in the pruned difference graph $\gdiff^*$.
    Therefore, we see that the number of in-edges to $v$ in $\gdiff^*$ is at least the number of out-edges from $v$ in $\gdiff$, which concludes the proof.
\end{proof}

By item (1), we know that $\gdiff^*$ is a DAG, so we can topologically order its nodes. We enumerate the nodes in the topological order as $v_1,\ldots,v_{\ell^* + 1}$, where $\ell^*\leq \ell$. By \cref{claim:out edges at most in edges} and by construction, $\vplus$ is the only node in $\gdiff^*$ with no in-edges, so we will necessary have $v_1 = \vplus$. For this ordering, let $E_i$ denote the $\ord{i}$ $\cut$ set.

To complete the proof of item (3), we need to show \cref{claim:cut set cards}.

\begin{claim}
\label{claim:cut set cards}
    Let $E_i$ denote the set of edges in the $\ord{i}$ $\cut$ set, and let $\ell^* \leq \ell$ be defined as above (i.e., one less than the number of nodes in $\gdiff^*$). We have $|E_1| \leq \min\{D,\ell\}$, and $|E_{i-1}|\geq |E_{i}|$ for all $i\in\{2,\ldots,\ell^*\}$.
\end{claim}

\begin{proof}[Proof of \cref{claim:cut set cards}]
We begin by proving $|E_1| \leq \min\{D,\ell\}$ (recall that $v_1 = \vplus$). Consider the case where $D\leq \ell$. Then $\vplus$ has at most $D$ incident edges in the projection, so it has at most $D$ out-edges in the difference graph. Now consider the case where $D > \ell$. All in-edges to nodes with degree at most $D$ in $\gs'_{[t]}$ are removed when creating the pruned difference graph (Line~\ref{line:remove low deg} of \cref{alg:pruned diff graph}), so since there are at most $\ell$ nodes with degree greater than $D$ in $\gs'_{[t]}$, we see that $\vplus$ has at most $\ell$ out-edges in the pruned difference graph. Therefore, $\vplus$ has at most $\min\{D,\ell\}$ edges, so $|E_1| \leq \min\{D,\ell\}$.

We next show $|E_{i-1}|\geq |E_{i}|$ for all $i\in\{2,\ldots,\ell^*\}$. If an edge is in $E_{i-1}$, then either it is in $E_{i}$ or it is an in-edge to $v_{i}$. Let $j$ be the number of in-edges to $v_{i}$. By \cref{claim:out edges at most in edges}, there are at most $j$ out-edges from $\vplus$. This gives us $|E_{i}|\leq |E_{i-1}| - j + j = |E_{i-1}|$.
\end{proof}

\cref{claim:cut set cards} tells us that we have $|E_{i}|\leq \min\{D,\ell\}$ for all $i\in\{1,\ldots,\ell^*\}$, which is what we wanted to prove for item (3).

\textbf{Proof of item (4).} To prove this item, we categorize the edges that were removed from the difference graph to create the pruned difference graph, and we count the number of edges in each category. We consider the following categories: (a) edges from $\vplus$ that are incident to nodes $v$ with $\deg_v(\gs'_{[t]}) \leq D$ and (b) other edges that were removed when creating $\gdiff^*$ from $\gdiff$.
Let $\ell^* \leq \ell$ be the number of out-edges from $\vplus$ to nodes $v$ such that $\deg_v(\gs'_{[t]}) > D$. Therefore, category~(a) has at most $D-\ell^*$ edges.

We next show that category (b) has at most $\ell^*$ edges.
By \cref{claim:out edges at most in edges}, all edges removed from $\gdiff$ to make the pruned difference graph $\gdiff^*$ are out-edges from nodes that appear in the pruned difference graph. (This is because, if an edge is an out-edge from a node that does not appear in the pruned difference graph $\gdiff^*$, then this node's in-degree in the pruned difference graph is zero, so its out-degree in the original difference graph must be zero by \cref{claim:out edges at most in edges}.) We use \cref{claim:bounding out edges} to establish another property about the edges that were removed to make the pruned difference graph, which will allow us to bound the number of edges that are in category (b).

\begin{claim}
\label{claim:bounding out edges}
    Let $E_i$ be defined as in the proof of item (3)---that is, $E_i$ is the set of edges in the $\ord{i}$ $\cut$. For all $i\in \{2,\ldots,\ell^*\}$, if $|E_{i-1}| = |E_i| + k_i$, then at most $k_i$ out-edges from $v_i$ were removed by \cref{alg:pruned diff graph}.
\end{claim}
\begin{proof}[Proof of \cref{claim:bounding out edges}]
    Suppose for contradiction that $v_i$ has $k_i' > k_i$ out-edges that were removed by \cref{alg:pruned diff graph} (that is, there are $k_i'$ out-edges from $v_i$ that appear in $\gdiff$ and do not appear in $\gdiff^*$). Let $j$ be the number of in-edges to $v_i$ in $\gdiff^*$.
    
    By definition, $|E_i|$ is at most $|E_{i-1}|$ minus the number of in-edges to $v_i$ in $\gdiff^*$, plus the number of out-edges from $v_i$ in $\gdiff^*$. We next bound each of these edge counts. The number of in-edges to $v_i$ in $\gdiff^*$ is $j$. The number of out-edges from $v_i$ in $\gdiff$, by \cref{claim:out edges at most in edges}, is at most the in-degree of $v_i$ in $\gdiff^*$---that is, $j$.
    Therefore, there are at most $j-k_i'$ out-edges from $v_i$ in $\gdiff^*$.
    
    Substituting these quantities into our upper bound for $|E_i|$ gives us $|E_i| \leq |E_{i-1}| - j + (j-k_i') = |E_{i-1}| - k_i '$. This simplifies to $|E_i| + k_i' \leq |E_{i-1}|$ which, since $k_i' > k_i$, contradicts the fact that $|E_i| + k_i = |E_{i-1}|$.
\end{proof}

We now observe, since $\vplus$ has $\ell^*$ edges to high degree nodes $v$ such that $\deg_v(\gs'_{[t]}) \leq D$, that $|E_1| \leq \ell^*$. Where we define $k_i$ as the number of out-edges from $v_i$ that appear in $\gdiff$ and do not appear in $\gdiff^*$, combining this observation with \cref{claim:bounding out edges,claim:cut set cards} gives us $\sum_{i\in[\ell^*]}k_i \leq \ell^*$. Therefore, there are at most $\ell^*$ out-edges from nodes in the pruned difference graph that were removed from the original graph to obtain the pruned difference graph---that is, there are at most $\ell^*$ edges in category (b). Since there are at most $D-\ell^*$ edges in category (a) and at most $\ell^*$ edges in category (b), we see that at most $D$ edges were removed to obtain the pruned difference graph, which is what we wanted to show.
\end{proof}

\fi

\ifnum\oakland=0
\subsection{Proof of Edge-to-Edge Stability for \texorpdfstring{\boldmath{$\projp$}}{DLL}}

In this section, we prove item (1b) of \cref{thrm:combined-stab}, which we repeat below for convenience. The proof uses \cref{lem:greedy-flatten}, along with some new ideas. At a high level, it follows from the observation that changing one edge in a graph (i.e., adding or removing it) can result in a path of edges that differ between projections. However---with the exception of the edge that differs between the edge-neighboring graph streams---the edges that differ between projections form a simple path. Moreover, all edges in the simple path must be incident to (at least) one node with degree greater than $D$. Because we only evaluate the stability statement through time steps for which the graph streams are $(D,\ell)$-bounded, we know there are $\ell$ such high-degree nodes, which allows us to bound the length of the resulting simple path.

\begin{theorem}[Item (1b) of \cref{thrm:combined-stab}]
\label{thrm:edge-edge stab proj}
    Let $T\in\N$, $D\in \N$, $\ell\in \N\cup\{0\}$, and let $\projp$ be \cref{alg:time aware projection} with \inccrit{} $\style = \projection$. If $\gs\nedge \gs'$ are edge-neighboring graph streams of length $T$, then for all time steps $t\in[T]$ such that $S$ and $S'$ are $(D, \ell)$-bounded through time $t$, the edge distances between the projections through time $t$ satisfy
    \[
                \dedge\Big(\projp(\gs)_{[t]}\ ,\ 
                \projp(\gs')_{[t]} \Big)
                \leq 2\ell + 1.
    \]
\end{theorem}

\begin{proof}[Proof of \cref{thrm:edge-edge stab proj}]
    To simplify notation, let $F, F'$ denote $\flatten{\projp(\gs)_{[t]}}$ and $\flatten{\projp(\gs')_{[t]}}$ respectively. Note that we only care to bound the edge distance between projections for times $t\in[T]$ where $\gs_{[t]},\gs'_{[t]}$ are $(D,\ell)$-bounded. By \cref{lem:greedy-flatten}, we only need to bound $\dedge(F, F')$.
    Without loss of generality, let $\gs'$ be the larger graph stream---that is, it contains an additional edge $\eplus = \{u,v\}$ that arrives at  \projstagetext{} $b$. This edge will increment both $d(u)$ and $d(v)$ by 1 at its  \projstagetext{}.
    
    In this proof of stability, we do a thought experiment where we first consider how the graph changes if only $d(u)$ is incremented at the \projstagetext{} of $\eplus$ (and $d(v)$ is not incremented); we call this graph $F^u$. We then consider the distance between $F^u$ and $F'$ (the graph where $d(v)$ also changes), which allows us to bound the distance between $F$ and $F'$ using the triangle inequality. Our proof builds on the proof technique used by \cite{DayLL16} to show stability of their algorithm for degree histograms: we will show that the set of edges that differ between $F$ and $F^u$ forms a path. More specifically, we show that each edge in the path must be incident to a node with degree greater than $D$ in $\gs_{[t]}$, and we show that (with the exception of the edge in the path that has the earliest \projstagetext{}) these edges form a simple path.

    We first see how many edges differ between $F$ and $F^u$. If $d(u)\geq D$ when $\eplus$ \projstagealttext{} $b$ when running $\projp$ on $\gs'$, then $F$ and $F^u$ will be identical since $\eplus$ will not appear in the output stream and it was already the case that $d(u)\geq D$ for edges that are considered after $\eplus$.
    However, if $d(u)< D$ before $\eplus$ is considered, then there may be an edge $e_1 = \{u,1\}$ incident to $u$ with \projstagetext{} $b_1 > b$ that appears in $F$ but not $F^u$ due to having $d(u) = D$ at \projstagetext{} $b_1$ in $\gs'$ instead of $d(u) = D-1 < D$.
    This removal of $e_1 = \{u,1\}$ can then result in having $d(1) = D - 1 < D$, which then allows for the addition of some edge $e_2 = \{1,2\}$ with \projstagetext{} $b_2 > b_1$. The presence of $e_2$ will then cause $d(2)$ to be incremented by 1, which can in turn cause an edge incident to node $2$, with a \projstagetext{} greater than $b_2$, to not appear in $F^u$.
    
    We now consider the size of this sequence $\mD$ of differing edges, where we order the edges in this sequence according to increasing \projstagetext{} in $\gs'$. \cref{claim:edge edge path stops} describes conditions under which this sequence will end, and we show that these conditions mean the sequence $\mD$ of differences will have size at most $\ell + 1$.

    \begin{claim}
    \label{claim:edge edge path stops}
        Consider a node $v$ that appears in two edges that occur consecutively in the sequence of differing edges. Let these edges be $e_x = \{v,x\}$ with \projstagetext{} $b_x$ in $\gs'$, and $e_y = \{v,y\}$ with \projstagetext{} $b_y$ in $\gs'$. Additionally, let $e_x$ appear in $F^u$ and not in $F$, and let $e_y$ appear in $F$ and not in $F^u$. Node $v$ has the following properties:
        \begin{enumerate}
            \item The degree of $v$ in $\gs'_{[t]}$ is bigger than $D$---that is, $\deg_v(\gs'_{[t]}) > D$.
            \item The \satsteptext{} of $v$ is as follows:
            \begin{itemize}
                \item If $b_x < b_y$, then $\satstep_{\gs'}(v) < b_y$ and $\satstep_{\gs}(v) = b_y$.
                \item If $b_x > b_y$, then $\satstep_{\gs'}(v) = b_x$ and $\satstep_{\gs}(v) < b_x$.
            \end{itemize}
        \end{enumerate}
    \end{claim}

    Note that $b_x \neq b_y$, so we do not need to consider that case.
    
    \begin{proof}[Proof of \cref{claim:edge edge path stops}]
        We first prove item (1). Assume for contradiction that there exists some $v$ incident to the described edges $e_x$ and $e_y$, such that $\deg_v(\gs'_{[t]})\leq D$. If $b_x < b_y$, the inclusion of $e_x$ will increase $d(v)$ by 1, but we will still have $d(v) < D$ for both graph streams, so this will not cause $e_y$ to be removed. Now consider the case where $b_y < b_x$. For both graph streams, we will have $d(v) < D$ for the \projstagetext{} of $e_x$, so the decremented value of $d(v)$ will not newly permit $e_x$ to be added if it couldn't be added in $F$. This contradicts our assumption and proves item (1).

        We next prove item (2). Let $e_x$ and $e_y$ be described as above, and consider the case where $b_x < b_y$. Since $e_y$ appears in $F$ but not in $F^u$, this means that the increment by 1 of $d(v)$ caused $v$  to go from non-saturated prior to the \projstagetext{} of $b_y$ in $S$ to saturated prior to the \projstagetext{} of $b_y$ in $S'$, so it must be the case that $\satstep_{\gs}(v) = b_y$; likewise, $\satstep_{\gs'}(v) < b_y$.
        Now, consider the case where $b_x > b_y$. The proof follows by a symmetric argument: the fact that $d(v)$ wasn't incremented at step $b_y$ in $\gs'$ caused $v$ to go from saturated to non-saturated prior to the \projstagetext{} of $e_x$, so it must be the case that $\satstep_{\gs'}(v) = b_x$; likewise, $\satstep_{\gs}(v) < b_x$.
    \end{proof}

    We now explain how \cref{claim:edge edge path stops} allows us to bound the number of edges that differ between $F$ and $F^u$. Recall that the added edge is $\eplus = \{u,v\}$. By item (1), either $\deg_u(\gs'_{[t]}) > D$ or $\eplus$ is the only edge that differs between $F$ and $F^u$. We now analyze the case where $\deg_u(\gs'_{[t]}) > D$. Item (1) tells us that, if the sequence $\mD$ of differing edges includes a node $w$ with $\deg_w(\gs'_{[t]}) \leq D$, then the edge containing this node will be the final node in the sequence. Consider the modified sequence $\mD'$ where we remove $\eplus$ from $\mD$. By (1), if the first edge in $\mD'$ is incident on only one node with degree greater than $D$ in $\gs'_{[t]}$, then it is the only node in $\mD'$. We focus on the case where it is incident on two nodes with degree greater than $D$ in $\gs'_{[t]}$. By item (2), we see that, if a node appears in two consecutive edges in the sequence $\mD$, it cannot appear again in the sequence because it will already be saturated for both projections at or before the \projstagetext{} of the second edge. Therefore, the edges in $\mD'$ form a simple path on the nodes in $\gs'_{[t]}$, and the first edge is between high-degree nodes with degree greater than $D$ in $\gs'_{[t]}$, and all of the remaining nodes are incident on at least one high-degree nodes with degree greater than $D$ in $\gs'_{[t]}$. Since there are at most $\ell$ such nodes, this simple path contains at most $\ell$ edges. When we also include $\eplus$, we see that at most $\ell + 1$ edges differ between $F$ and $F^u$.

    A symmetric argument can be used to show that at most $\ell$ edges differ between $F^u$ and $F'$ (with the value being $\ell$ instead of $\ell + 1$ since either $\eplus$ appears in both $F^u$ and $F'$ or appears in neither). By the triangle inequality, then, the flattened graphs $F$ and $F'$ differ on at most $2 \ell + 1$ edges, which is what we wanted to show.    
\end{proof}

\fi

\section{{\boldmath From \texorpdfstring{$D$}{D}-restricted Privacy to Node Privacy}}
\label{sec:always-private}

In this section, we present our general transformation from restricted edge- and node-DP algorithms to node-DP algorithms (\cref{alg:bb}). As input, \cref{alg:bb} takes several user-specified parameters ($\eps_\test,\failtest,\beta,D,T$), a length-$T$ graph stream $\gs$ of arbitrary degree, and black-box access to a base algorithm.
At every time step, \cref{alg:bb} returns either the result of running one more step of the base algorithm on the projection of the graph stream, or a special symbol $\bot$ denoting failure. 

The following theorem encapsulates its privacy and accuracy properties, and its efficiency. Roughly, the overall algorithm is private for all graph streams as long as the base algorithm satisfies either edge or node variants of $D$-restricted DP. It is accurate on graph streams through all time steps that satisfy the assumed degree bound, and adds only linear overhead to the runtime of the base algorithm.

\begin{restatable}[Privacy for all graph streams and restricted accuracy]{theorem}{ParamsForBB}
\label{thrm:bb priv params}
    Consider running \cref{alg:bb-new} with parameters $\eps_\test, \failtest, \beta, D, T$, and let $\ell = \ceil{8\ln\paren{\frac{T}{\beta\failtest}}/\eps_\test}$ and $D' = D + \ell$ as in lines~\ref{line:set ell}~and~\ref{line:set d-prime}.
    \begin{enumerate}[leftmargin=*]
    \item \textbf{(Node privacy from restricted edge privacy.)} Suppose $\mech_{D'}$ satisfies $D'$-restricted $(\eps',\delta')$-\textbf{edge}-DP under continual observation for graph streams of length $T$. Then \cref{alg:bb-new} satisfies (unrestricted) $(\eps,\delta)$-\textbf{node}-DP under continual observation for graph streams of length $T$, with
    \begin{eqnarray*}
        \eps &=& \eps_\test + \eps' \cdot \paren{D' + \ell} 
        \quad \text{and} \\
        \delta &=& (1 + e^{\eps_\test})e^{\eps}\failtest + \delta' \cdot e^{\eps' \cdot (D' + \ell) } \cdot\paren{D' + \ell } 
        \noakland{\quad (\text{which is}\ O\paren{\failtest + \delta' (D + \ell)} \text{\ when }\ \eps \leq 1)}.
    \end{eqnarray*}
    In particular, for $\eps\leq 1 $ and $T\geq 2$, it suffices to set $\eps_\test = \eps/2$ and $\failtest = \del/30$, and
    \begin{eqnarray*}
        \eps' = \Theta \left( \frac{\eps}{B} \right)
        \quad \text{and} \quad
        \del' = \Theta \left( \frac{\del}{B} \right),&& \\
        \quad \text{where} \quad 
        B = D + \frac{\log(T/(\beta\del))}{\eps}.&&
    \end{eqnarray*}
    
    \item \textbf{(Node privacy from restricted node privacy.)} Suppose $\mech_{D'}$ satisfies $D'$-restricted $(\eps',\delta')$-\textbf{node}-DP under continual observation for graph streams of length $T$. Then \cref{alg:bb-new} satisfies (unrestricted) $(\eps,\delta)$-\textbf{node}-DP under continual observation for graph streams of length $T$, with
    \begin{eqnarray*}
        \eps &=& \eps_\test + \eps' \cdot \paren{2\ell + 1} 
        \quad \text{and} \\
        \delta &=& (1 + e^{\eps_\test})e^{\eps}\failtest + \delta' \cdot e^{\eps' \cdot (2 \ell + 1)} \cdot\paren{2\ell + 1} 
        \noakland{\quad (\text{which is}\ O\paren{\failtest + \delta' \ell} \text{\ when }\ \eps \leq 1)}. 
    \end{eqnarray*}
    In particular, for $\eps\leq 1 $ and $T\geq 2$, it suffices to set $\eps_\test = \eps/2$ and $\failtest = \del/30$, and
    \begin{eqnarray*}
        \eps' = \Theta \left( \frac{\eps}{B} \right)
        \quad \text{and} \quad
        \del' = \Theta \left( \frac{\del}{B} \right),&& \\
        \quad \text{where} \quad 
        B = \frac{\log(T/(\beta\del))}{\eps}.&&
    \end{eqnarray*}

    \item \textbf{(Accuracy.)} If the input graph stream $\gs$ is $(D,0)$-bounded through time step $t'$, then the output from $\mech_{D'}$ is released at all time steps $t\leq t'$ with probability at least $1-\beta$.

    \item \textbf{(Time and space complexity.)} \cref{alg:bb-new} adds linear overhead to the time and space complexity of $\mech_{D'}$. More formally, let $R_{[t]}$ and $S_{[t]}$ be the runtime and space complexity of $\mech_{D'}$ through $t$ time steps, and let $n_{[t]}$ and $m_{[t]}$ be the number of nodes and edges in the graph stream through time $t$. The total time complexity of \cref{alg:bb} through time $t\in[T]$ is $R_{[t]} + O(n_{[t]} + m_{[t]} + t)$, and the total space complexity through time $t$ is $S_{[t]} + O(n_{[t]})$.

    \end{enumerate}
\end{restatable}

The basic idea of the algorithm follows the Propose-Test-Release (PTR) framework of \cite{DL09}: we use the sparse vector technique \cite{DNRRV09,RothR10,HardtR10} to continually check that the conditions of \cref{thrm:combined-stab} are met. As long as they are, we can safely run the base algorithm with parameters scaled according to the edge (or node) sensitivity of the projection so that, by group privacy, its outputs are $(\eps,\del)$-indistinguishable on all pairs of node-neighboring inputs, satisfying $(\eps,\del)$-node-DP.

\noakland{
In \cref{sec:test for bad graphs}, we present a query with node-sensitivity 1 to test
that a graph is $(D,\ell)$-bounded, the condition required in \cref{thrm:combined-stab}. In \cref{sec:ptr crt}, we show how a novel, online variant of the PTR framework can be applied to the continual release setting. In \cref{sec:black box acc priv}, we present the black-box framework (\cref{alg:bb}) for transforming restricted edge- or node-private algorithms to node-private algorithms\noakland{, and analyze its privacy, accuracy, and efficiency (proof of \Cref{thrm:bb priv params})}.
}

\subsection{Testing for Bad Graphs}
\label{sec:test for bad graphs}

To use the sparse vector technique, we need a stream of queries with (node) sensitivity 1. Below, we define a function $\DistToGraph$ with node-sensitivity 1 that returns the minimum, over all graphs, of the node distance between the input graph and a graph with at least $\ell$ nodes of degree greater than $D$. In other words, it tells how close the input graph is to a graph that is \emph{not} $(D,\ell-1)$-bounded. This function can be used to make such a stream of node-sensitivity 1 queries for continually checking whether the conditions of \cref{thrm:combined-stab} are satisfied.

\begin{definition}[$\DistToGraph_{D,\ell}$]
\label{defn:dist to graph}
    Let $\DistToGraph_{D,\ell} : \mG \to \N\cup \{0\}$ return the minimum node distance between the input graph and a graph with at least $\ell$ nodes of degree greater than $D$.
\end{definition}

We now present some properties of $\DistToGraph$, which we will use in the proof of privacy for the black-box framework described in \cref{alg:bb}.

\begin{lemma}[Properties of $\DistToGraph$]
\label{lem:dtg-properties}
\label{cor:sens of dtg}
\label{thrm:dist poly time}
\label{cor:dtg meaning}
    $\DistToGraph_{D,\ell}$ has the following properties:
    \begin{enumerate}[leftmargin=*]
        \item $\DistToGraph_{D,\ell}$ has node-sensitivity 1.
        \item  For an input graph $G$ with $|V|$ nodes and $|E|$ edges, $\DistToGraph_{D,\ell}$ can be computed in time $O(|V| + |E|)$.
        
        Furthermore, given a graph stream $\gs$, one can determine the sequence of distances for all prefixes of the stream, online. With each node or edge arrival, the distance can be updated in constant time.

        \item Let $G\in \mG$ be a $(D,\ell)$-bounded graph (i.e., with at most $\ell$ nodes of degree greater than $D$). Then  
        $\DistToGraph_{D+k,\ell+k}(G) \geq k$ for all $k\in\N\cup \{0\}$.
    \end{enumerate} 
\end{lemma}

\yoakland{We omit a full proof; however, the following algorithm computes $\DistToGraph$ and can be implemented in linear time. Given a graph $G\in \mG$, 
first check if $G$ is $(D,\ell)$-bounded. While that condition is not satisfied,  add a node to $G$ with edges incident to all existing nodes (and check $(D,\ell)$-boundedness again).
$\DistToGraph_{D,\ell}(G)$ equals the number of nodes that have been added to $G$ when the check fails.}

\ifnum\oakland=0
\begin{proof}[Proof of \cref{cor:sens of dtg}]
    \textbf{Proof of item (1).} $\DistToGraph_{D,\ell}$ is defined in terms of the minimum number of nodes to add, so changing the graph by one node will change $\DistToGraph_{D,\ell}$ by at most 1.

    \textbf{Proof of item (2).} We first describe a method for computing $\dtg_{D,\ell}$ and prove this method's correctness, and then describe an efficient implementation.
    Given a graph $G\in \mG$, $\dtg_{D,\ell}(G)$ can be computed as follows: first, check if $G$ is $(D,\ell-1)$-bounded. While that condition is  satisfied,  add a node to $G$ with edges incident to all existing nodes and check $(D,\ell-1)$-boundedness again.
    Return the number of nodes that have been added to $G$ when the check first fails.

    We  prove that the quantity returned by this method is equal to $\dtg_{D,\ell}(G)$. First, observe that the method cannot underestimate the distance since, for the number $\hat k$ it outputs, there is actually a graph within node distance $\hat k$ of $G$ that is not $(D,\ell)$-bounded. To see why the method  never overestimates the distance,
    suppose $\dtg_{D,\ell}(G) = k$, and let $H$ be a graph with at least $\ell$ nodes of degree greater than $D$ that is closest to $G$ in node distance. $H$ can be formed by removing some set of nodes $B$ (and associated edges) from $G$ and adding some set of nodes $A$ (and associated edges) such that $|A| + |B| = k$.
    Now consider the graph $H'$ which is formed by taking $G$, adding all nodes in $A$ (along with the associated edges), and adding edges from every node in $A$ to all other nodes in the graph that do not already have an edge from that node. We do not delete the nodes and edges that are in $B$. $H$ is a subgraph of $H'$, so if $H$ has at least $\ell$ nodes of degree greater than $D$, then $H'$ also has at least $\ell$ nodes of degree greater than $D$. 
    Additionally,  $H'$ is exactly the graph that will be formed by the method described above after at most $k$ steps. Finally,  $H'$ is node distance $|A|\leq k$ from $G$, and $\hat k = |A|$ is the output returned by this method.

    Naively, an algorithm for the above method can be implemented by storing a graph in memory and manipulating this graph with every pass through the while loop. However,
    storing the entire graph is not necessary.
    There is a linear-time algorithm that computes the same quantity as the method described above. 
    To see why that is, 
    let $|V|$ be the number of nodes in $G$. After adding $j$ nodes that have edges to all other nodes, the number of nodes with degree greater than $D$ will be 0 if $D\geq |V| + j - 1$ (a graph with $|V| + j$ nodes has maximum degree $|V|+j-1$); and will otherwise be the number of nodes in $G$ with degree greater than $D-j$, plus $j$.
    Therefore, it suffices to maintain a \emph{cumulative degree histogram} $\chist{G}{\cdot}$, which we define below, and compute a value based on that.
    
    Let $\chist{G}{j}$ hold the number of nodes in $G$ with degree at least $j$, and let $\chist{G}{\cdot}$ have buckets for $j\in\set{0,\ldots, D+1}$. Given this cumulative degree histogram, we return \emph{the smallest value of $j\geq \max\set{D - |V| + 2,0}$ such that}
    \begin{equation}
        j + \chist{G}{D-j+1} \geq \ell \, . 
        \label{eq:chist-degree-distance}
    \end{equation}
    
    Call this smallest such value $k$. We now show that, when a new node or edge arrives, we can update the value of $k$ in constant time. To do this, we maintain the following data structures:
    a hash table $H$ for node degrees that takes a node name as input and returns its degree (we assume unassigned values in the hash table are set to 0);
    a hash table $C$ for the cumulative histogram that takes an integer $i$ as input and returns the number of nodes with degree at least $i$;
    a counter $n$ that tracks the number of nodes in the graph, initialized to 0;
    and the current value $k$ of $\dtg_{D,\ell}(G)$, which we initialize to $k = \max\set{D+2,\ell}$ when the graph is empty (since every node in a complete graph on $D + 2$ nodes has degree greater than $D$, but $\ell$ nodes are needed for a non-$(D,\ell)$-bounded graph when $\ell > D + 2$).

    We first show that $H,C,n$ can be updated in constant time when a new node or edge is added to $G$. When a new node $v$ arrives, set $H(v) = 0$, $C(0) \pluseq 1$, and $n\pluseq 1$. When a new edge $\set{u,v}$ arrives, set $H(u) \pluseq 1$ and $H(v) \pluseq 1$, and then set $C(H(u)) \pluseq 1$ and $C(H(v)) \pluseq 1$. 

    We next show that the value of $k$ (i.e., the smallest value of $j$ such that (\ref{eq:chist-degree-distance}) is satisfied) can be updated in constant time, assuming $H$ and $C$ are up to date. Recall that $k$ is the value of $\dtg_{D,\ell}$ that we want to update, and that this function has node sensitivity 1.
    This means that adding a node can change $k$ by at most 1, and that adding an edge can change $k$ by at most 2 since every graph with an edge $e$ is node distance at most 2 from a graph without $e$ (an edge $e$ has a vertex cover of size 1, so to obtain an otherwise identical graph without that edge requires at most removing a vertex incident to $e$ and then re-adding the vertex without $e$).
    Moreover, the value of $k$ cannot increase as nodes and edges are added since the old graph is a subgraph of the updated graph, so the new graph can be no further from a non-$(D,\ell-1)$-bounded graph than the old graph. Therefore, it suffices to check, for $k' \in \set{k,k-1,k-2}$ whether $k' + C(D-k'+1) \geq \ell$ and $k' \geq \max\set{D - n + 2,0}$, and then set $k$ to the smallest value $k'$ for which these conditions hold.

    Since all of the updates to $H$, $C$, $n$, and $k$ in response to the arrival of a node or edge are constant-time operations, $\dtg_{D,\ell}(G)$ can be computed in time $O(|V| + |E|)$ for a graph $G = (V,E)$.\footnote{This linear-time algorithm assumes that the graph stream does not contain duplicate edges; to remove this assumption, a check for whether an edge is a duplicate can be performed in constant time (and that edge can then be ignored).}

    \textbf{Proof of item (3).} We know from the proof of item (2) that the shortest sequence of node-neighboring graphs, from $G$ to the closest graph with at least $\ell$ nodes of degree greater than $D$, consists of sequentially adding nodes to $G$ that have all possible edges.
    Adding one node $v^+$ to $G$ increases the degree of each node in $G$ by at most 1. Additionally, the added node $v^+$ may have degree greater than $D$. 
    Thus, if $G$ is $(D,\ell)$-bounded, then adding one node yields a graph $G^+$ that  is $(D+1,\ell+1)$-bounded. The claim follows by induction on the number of added nodes.
\end{proof}
\fi

\subsection{PTR in the Continual Release Setting}
\label{sec:ptr crt}

The high-level idea of our general transformation is to use a novel, online variant of the Propose-Test-Release framework (PTR) of \cite{DL09}.
To our knowledge, PTR has not been used previously for designing algorithms in the continual release setting. In this section, we show
that PTR can, in fact, be applied in the continual release setting when the algorithm checking the safety condition is itself private under continual observation, as is the sparse vector algorithm.

\begin{theorem}[Privacy of PTR under Continual Observation]
\label{thrm:ptr crt simple}
\label{thrm:priv of ptr under crt}
    Let $\base:\mX^T\to \mY^T$ be a streaming algorithm and $\mZ \subseteq \mX^T$  a set of streams such that, for all neighbors $\gs,\gs'\in\mZ$,
    $\base(\gs)\approx_{\eps_\base,\delta_\base}\base(\gs')$. 
    
    Let $\test:\mX^T\to \{\bot,\top\}$ be $(\eps_\test,\del_\test)$-DP under continual observation such that for every stream $\gs$, if there is a time $t$ such that $\gs_{[t]}$ does not equal the $t$-element prefix of any item in $\mZ$, then $\test(\gs)$ outputs 
    $\bot$ w.p. at least $1-\beta_\test$ at or before time $t$.

    If \cref{alg:ptr crt} is initialized with $\test$ and $\base$, it will satisfy $(\eps,\del)$-DP under continual observation on all input streams $x\in\mX^T$, for
    \[
        \eps = \eps_\base + \eps_\test
        \quad \text{and} \quad
    \]
    \[
        \delta =
        \del_\base +
        (1 + e^{\eps_\base + \eps_\test}) \del_\test + (1 + e^{\eps_\test})e^{\eps_\base + \eps_\test} \beta_\test
        .
    \]
\end{theorem}

\begin{algorithm}[ht!]
    \caption{Algorithm $\ptr$ for propose-test-release under continual observation.}
    \label{alg:ptr crt}
    \begin{algorithmic}[1]
        \Statex \textbf{Input:} Stream $x\in \mX^T$, streaming algorithms $\test:\mX^T\to \{\bot,\top\}^T$ and $\base:\mX^T\to \mY^T$.
        \Statex \textbf{Output:} Stream in $(\{\bot,\top\}\times \{\mY,\bot\})^T$.

        \State $\passed \gets \true$
        \State $b_0 \gets $ initial state of $\base$ 
        \State $s_0 \gets $ initial state of $\test$
        \ForAll{$t\in[T]$}
            \State $(\verdict,s_t)\gets \test(x_t;s_{t-1})$
            \yoakland{\State} \Comment{\commentstyle{Send $x_t$ to $\test$, and set $\verdict$ to be its output}}
            \State Output $\verdict$
            \If{$\verdict = \bot$} $\passed \gets \false$ \EndIf
            \If{$\passed = \true$}
                \State $(y_t,b_t)\gets \base(x_t;b_{t-1})$
                \yoakland{\State} \Comment{\commentstyle{Send $x_t$ to $\base$ and output the result}}
                \State Output $y_t$
            \Else \ output $\bot$
            \EndIf
        \EndFor
    \end{algorithmic}
\end{algorithm}

While it is perhaps unsurprising that PTR applies to the continual release setting, the proof of \cref{thrm:ptr crt simple} does not follow immediately from the standard PTR analysis: the standard analysis is binary (either release the output of the base algorithm, or don't), while the version we need releases a dynamically chosen prefix of the base algorithm's output.

\ifnum\oakland=0

Note that there is no completeness requirement for $\test$. Whereas completeness is useful for accuracy, our proof of privacy only requires soundness of $\test$.
We now move to the proof of \cref{thrm:priv of ptr under crt}.

\begin{proof}[Proof of \cref{thrm:priv of ptr under crt}]
    Fix two neighboring streams $x$ and $ x'$ in $ \mX^T$, and let $t^*\in [T]$ be the smallest value such that either $(x_1,\ldots,x_{t^*})$ or $(x'_1,\ldots, x'_{t^*})$ does not equal the $t^*$-element prefix of any item in $\mZ$.
    (If there is no such $t^*$, just set $t^* = T + 1$.)
    We consider the behavior of \cref{alg:ptr crt} (denoted $\ptr$) on inputs $x$ and $x'$.

    To analyze privacy, we also consider the behavior of two hypothetical algorithms
    $\ttildeptr$ and $\tildeptr$, which we define below. (They are hypothetical because they have access to the value of $t^*$, which is defined with respect to a pair of neighboring input streams.) We show indistinguishability relationships between $\ptr(x)$, $\tildeptr(x)$, $\tildeptr(x')$, and $\ptr(x')$; we then apply the weak triangle inequality for indistinguishability (\cref{lem:weak triangle ineq}) to relate
    $\ptr(x)$ and $\ptr(x')$.

    Let $\ttildeptr$ be the algorithm that releases the outputs of $\test$ and $\base$ at all time steps less than $t^*$, and releases the output of $\test$ and the symbol $\bot$ at all time steps at and after $t^*$. (In the case that $t^* = T+1$, the algorithm will release the output from $\base$ at all time steps.) We see, by the indistinguishability properties of $\test$ and $\base$, that
    \begin{equation}
    \label{eq:ttildeptr indist}
        \ttildeptr(x)\approx_{(\eps_\base + \eps_\test, \del_\base + \del_\test )} \ttildeptr(x').
    \end{equation}

    We now define another (hypothetical) algorithm, $\tildeptr$. Let $t_\bot$ be the first time step at which $\test$ returns $\bot$, setting $t_\bot = T + 1$ if there is no such time step.
    Note that, whereas $t^*$ is defined with respect to a pair of neighboring inputs, $t_\bot$ is a random variable defined with respect to an execution of $\test$ on a particular input.
    We define $\tildeptr$ as the algorithm that releases the outputs of $\test$ and $\base$ at all time steps less than $\min\{t^*,t_\bot\}$, and releases the output of $\test$ and the symbol $\bot$ at all time steps at and after $\min\{t^*,t_\bot\}$. Observe that $\tildeptr$ is a post-processed version of $\ttildeptr$. By Expression~\eqref{eq:ttildeptr indist}, we get
    \begin{equation}
    \label{eq:tildeptr indist}
        \tildeptr(x)\approx_{(\eps_\base + \eps_\test,  \del_\base + \del_\test )} \tildeptr(x').
    \end{equation}

    We now compare $\ptr(x)$ and $\tildeptr(x)$.
    Let $\mB$ be the (bad) event that
    the output from $\base$ is released at or after time step $t^*$. Without loss of generality, let $x$ be the stream for which the first $t^*$ elements are not equal to the $t^*$-element prefix of any item in $\mZ$. Conditioned on $\mB$ not occurring, $\ptr(x)$ and $\tildeptr(x)$ have the same probability distribution of outputs. The probability of $\mB$ is at most $\beta_\test$ for both algorithms $\ptr(x)$ and $\tildeptr(x)$,
    so $\ptr(x) \approx_{(0,\beta_\test)} \tildeptr(x)$. Similarly, by the DP guarantee, $\mB$ occurs with probability at most $e^{\eps_\test}\beta_\test + \del_\test$ on input $x'$, so
    $\ptr(x') \approx_{(0,e^{\eps_\test}\beta_\test)} \tildeptr(x')$. Combining this with Expression~\eqref{eq:tildeptr indist} shows that
    \[
        \ptr(x)
        \approx_{(0,\beta_\test)}
        \tildeptr(x)
        \approx_{(\eps_\base + \eps_\test, \del_\base + \del_\test)}
        \tildeptr(x')
        \approx_{(0,e^{\eps_\test}\beta_\test + \del_\test)}
        \ptr(x') \, .
    \]
    Let $\eps = \eps_\base + \eps_\test$ and 
    $\del
    =
    \del_\base + \del_\test +
    e^{\eps_\base + \eps_\test}(\beta_\test + e^{\eps_\test}\beta_\test + \del_\test)
    $.
    By the weak triangle inequality (\cref{lem:weak triangle ineq}), 
    \(
        \ptr(x)
        \approx_{(\eps,\del)}
        \ptr(x').
    \)
    Since this relationship applies for all pairs of neighboring streams $x$ and $x'$, we see that $\ptr$ is $(\eps_\base + \eps_\test, \del_\base +
    (1 + e^{\eps_\base + \eps_\test}) \del_\test + (1 + e^{\eps_\test})e^{\eps_\base + \eps_\test} \beta_\test)$-DP, as desired.
\end{proof}

\fi

\subsection{Accuracy and Privacy of \texorpdfstring{\cref{alg:bb-new}}{the black-box transformation}}
\label{sec:black box acc priv}

In \cref{alg:bb-new}, we present our method for obtaining node-DP algorithms from restricted edge-DP and restricted node-DP algorithms. 

Our algorithm works as follows. Where we set $D' = D+ \ell$ and $\ell\approx \frac{\log(T/\delta)}{\eps}$, we initialize $\ptr$ (\cref{alg:ptr crt}) with (1) a $\base$ algorithm that offers indistinguishability on neighbors from the set of $(D',\ell)$-bounded graphs and (2) a $\test$ algorithm that uses sparse vector to ensure that $\DistToGraph_{D', \ell}$ is non-negative (i.e., that the graph stream is $(D',\ell)$-bounded), where $\ell$ is an additive slack term to account for the error of the sparse vector technique. Specifically, $\base$ is the composition of $\mech_{D'}$ with $\projo[D']$, where $\mech_{D'}$ satisfies $D'$-restricted edge- or node-DP.

If $\test$ succeeds, the projection will be stable with high probability, so we can safely release the result of running $\base$ on the projected graph; if $\test$ fails, the symbol $\bot$ is released. \cref{alg:bb-new} post-processes the outputs from $\ptr$, so it inherits the privacy properties of $\ptr$. The privacy properties of \cref{alg:bb-new}, stated in \cref{thrm:bb priv params}, follow from\noakland{ \cref{thrm:priv of ptr under crt}}\yoakland{ \cref{thrm:ptr crt simple} and are proven in the full version}.

\begin{algorithm}[ht!]
    \caption{$\bben$ for transforming restricted-DP algorithms to node-DP algorithms.}
    \label{alg:bb-new}
    \label{alg:bb}
    \begin{algorithmic}[1]
        \Statex \textbf{Input:} Privacy params $\eps_\test > 0$, $\failtest \in (0,1]$; accuracy param $\beta\in (0,1]$; degree bound $D\in\N$; time horizon $T\in\N$; graph stream $\gs\in\mGS^T$; alg 
        $\mech$
        \noakland{(see \cref{thrm:bb priv params} for possible assumptions on $\mech$).}
        \Statex \textbf{Output:} A stream, where each term is $\bot$ or an estimate from $\mech$.
        \State $\tau = -8\ln(1/\failtest)/\eps_\test$
        \State $\ell = \ceil{8\ln(T/(\beta\failtest))/\eps_\test}$\label{line:set ell}
        \State $D' = D + \ell$\label{line:set d-prime}
        \State $\base = \mech_{D'} \circ \projo[D']$
        \yoakland{\State} \Comment{\commentstyle{ $\projo[D']$ is \cref{alg:time aware projection} with $\style = \original$}}
        \ForAll{$t\in [T]$}
            \State $q_t(\cdot) = -1\cdot \DistToGraph_{D',\ell}(\flatt{{\ \cdot\ }_{[t]}})$
        \EndFor
        \State $\test = \begin{cases}
            \svt \text{ with privacy param $\eps_\test$,} \\
            \text{\ \ thresh $\tau$, queries $q_1,\ldots,q_T$}
        \end{cases}$  
        \State $s_0\gets$ initial state for $\ptr$
        \State Initialize $\ptr$ with algorithms $\test$ and $\base$
        \newcommand{\testval}{\mathsf{Test\text{-}val}}
        \newcommand{\baseval}{\mathsf{Base\text{-}val}}
        \ForAll{$t\in [T]$}
            \State $(\testval,\baseval,s_t)\gets \ptr(\gs_t; s_{t-1})$
            \State Output $\baseval$
            \yoakland{\State} \Comment{\commentstyle{$\baseval$ will be $\bot$ once the test fails}}
        \EndFor
    \end{algorithmic}
\end{algorithm}

\ifnum\oakland=0
We now prove \cref{thrm:bb priv params}, which we repeat below for convenience.

\ParamsForBB*

\begin{proof}[Proof of \cref{thrm:bb priv params}]      
    We prove each part below.
    
    \textbf{Proof of items (1) and (2) in \cref{thrm:bb priv params}.} To prove items (1) and (2),
    we break our proof into the following claims about $\test$ and $\base$ as initialized in \cref{alg:bb-new}, with the parameters specified in the statement of \cref{thrm:bb priv params}.
    Once we have shown the claims, the privacy properties follow immediately from \cref{thrm:priv of ptr under crt}.

    \cref{claim:base in bb} deals with the indistinguishability properties of $\base$. \cref{claim:test in bb} deals with the privacy of $\test$, as well as its failure probability: where $t$ is the first time step for which the input is not $(D',\ell)$-bounded, $\test$ outputs $\bot$ at or before time step $t$ with probability at least $1-\failtest$. Below, we use the term \emph{$(D',\ell)$-restricted privacy}
    to refer to indistinguishability for neighbors drawn from the set of graphs with at most $\ell$ nodes of degree greater than $D'$.

    \begin{claim}[Indistinguishability of $\base$]
    \label{claim:base in bb}
    Let $\base$ be parameterized as in \cref{alg:bb}.
    If $\gs$ and $\gs'$ are node-neighboring datasets that are both $(D',\ell)$-bounded, then $\base(\gs)\approx_{\eps_\base,\del_\base} \base(\gs')$ for the following values of $\eps_\base$ and $\del_\base$:
    \begin{enumerate}
        \item If $\mech_{D'}$ satisfies $D'$-restricted $(\eps',\del')$-\textbf{edge}-DP, then
        \[
            \eps_\base = \eps' \cdot \paren{D' + \ell} 
            \quad \text{and} \quad
            \delta_\base = \delta' \cdot e^{\eps_\base} \cdot\paren{D' + \ell } .
        \]
        
        \item If $\mech_{D'}$ satisfies $D'$-restricted $(\eps',\del')$-\textbf{node}-DP, then
        \[
            \eps_\base = \eps' \cdot \paren{2\ell + 1} 
            \quad \text{and} \quad
            \delta_\base = \delta' \cdot e^{\eps_\base} \cdot\paren{2\ell + 1} .
        \]
    \end{enumerate}
    \end{claim}

    \begin{proof}[Proof of \cref{claim:base in bb}]
        Let $\gs$ and $\gs'$ be node-neighboring graph streams that are $(D',\ell)$-bounded.  Group privacy (\cref{lemma:group priv}) tells us that if a mechanism $\mM$ is $(\eps,\del)$-DP, then its outputs on a pair of datasets which differ in at most $k$ individuals are $(k\cdot \eps, k\cdot e^{k\eps}\cdot \del)$-indistinguishable.
        
        \textbf{Proof of item (1).}         
        By item (2a) of \cref{thrm:combined-stab}, the projections of $\gs$ and $\gs'$ will have maximum degree at most $D'$ and will be edge distance at most $D' + \ell$ from each other, so group privacy tells us that the outputs from $\mech_{D'}$ on the projections with maximum degree at most $D'$ will be $(\eps_\base,\del_\base)$-indistinguishable for
        \[
            \eps_\base = (D'+\ell)\cdot \eps'
            \mathrm{\quad and\quad }
            \del_\base = (D'+\ell)\cdot e^{\eps_\base}\del',
        \]
        which is what we wanted to show.
        
        \textbf{Proof of item (2).}
        By item (3a) of \cref{thrm:combined-stab}, the projections of $\gs$ and $\gs'$ will have maximum degree at most $D'$ and will be node distance at most $2\ell + 1$ from each other, so group privacy tells us that the outputs from $\mech_{D'}$ on the projections with maximum degree at most $D'$ will be $(\eps_\base,\del_\base)$-indistinguishable for
        \[
            \eps_\base = (2\ell + 1)\cdot \eps'
            \mathrm{\quad and\quad }
            \del_\base = (2\ell + 1)\cdot e^{\eps_\base}\del',
        \]
        which is what we wanted to show.
    \end{proof}

    \begin{claim}[Privacy and failure probability of $\test$]
    \label{claim:test in bb}

        Let $\test$ be parameterized as in \cref{alg:bb}. $\test$ has the following properties:
        \begin{enumerate}
            \item $\test$ is $(\eps_\test,0)$-node-DP under continual observation.
            \item For every input stream $\gs$ and time $t\in [T]$, if $\gs_{[t]}$ is not $(D',\ell)$-bounded, then $\test(\gs)$ outputs 
            $\bot$ w.p. at least $1-\failtest$ at or before time $t$.
        \end{enumerate}
    \end{claim}

    \begin{proof}[Proof of \cref{claim:test in bb}]
        We prove each item below.
        
        \textbf{Proof of item (1).} By \cref{cor:sens of dtg}, for all $t\in[T]$ the query $q_t = -1\cdot \DistToGraph_{D',\ell}(\flatt{\gs_{[t]}})$ has node-sensitivity 1, so by \cref{thrm:svt privacy} the output of $\svt$ is $(\eps_\test,0)$-node-DP under continual observation.
        
        \textbf{Proof of item (2).} Let $\gs$ be a graph stream of length $T$, and let $t^*\in[T]$ be the smallest value such that $\gs_{[t^*]}$ has at least $\ell$ nodes of degree greater than $D'$, and set $t^* = T+1$ if there is no such value. (Note that this time $t^*$ is less than or equal to the smallest value $t$ such that $S_{[t]}$ is not $(D',\ell)$-bounded.)
        By definition we have $\DistToGraph_{D',\ell}(\flatt{\gs_{[t^*]}}) = 0$, so for $q_{t^*} = -1\cdot \DistToGraph_{D',\ell}(\flatt{\gs_{[t^*]}})$ we have $q_{t^*} = 0$.
        Equivalently, where we set $\tau = -8\ln(1/\failtest)/\eps_\test$ (as in \cref{alg:bb-new}), we have
        \[
            q_{t^*} = 8\ln(1/\failtest)/\eps_\test + \tau.
        \]
        By item (1) of \cref{thrm:combined-sep-svt}, then, in \cref{alg:bb} $\svt$ outputs $\Above$ at or before time step $t^*$ with probability at least $1-\failtest$, so we also have that \cref{alg:bb} outputs $\bot$ at or before time $t^*$ (and at all subsequent time steps) with probability at least $1 - \failtest$.
    \end{proof}

    From \cref{claim:base in bb,claim:test in bb}, we see that $\base$ and $\test$ as defined in \cref{alg:bb} satisfy the conditions on $\base$ and $\test$ that are described in \cref{thrm:ptr crt simple}. \cref{alg:bb-new} is a post-processed version of $\ptr$ initialized with $\base$ and $\test$, so it shares the privacy properties of $\ptr$. Therefore, by applying \cref{thrm:priv of ptr under crt}, we see that \cref{alg:bb-new} has the following privacy guarantees:
    \begin{enumerate}
        \item If $\mech_{D'}$ satisfies $D'$-restricted $(\eps',\del')$-\textbf{edge}-DP, then \cref{alg:bb-new} is $(\eps,\del)$-DP for        
        \[
            \eps = \eps_\test + \eps' \cdot \paren{D' + \ell} 
            \quad \text{and} \quad
            \delta = (1 + e^{\eps_\test})e^{\eps}\failtest + \delta' \cdot e^{\eps' \cdot (D' + \ell) } \cdot\paren{D' + \ell } .
        \]

        \item If $\mech_{D'}$ satisfies $D'$-restricted $(\eps',\del')$-\textbf{node}-DP, then \cref{alg:bb-new} is $(\eps,\del)$-DP for
        \[
            \eps = \eps_\test + \eps' \cdot \paren{2 \ell + 1}
            \quad \text{and} \quad
            \delta = (1 + e^{\eps_\test})e^{\eps}\failtest + \delta' \cdot e^{\eps' \cdot (2 \ell + 1)} \cdot\paren{2\ell + 1} .
        \]
    \end{enumerate}
    These are the values for $\eps$ and $\del$ that we wanted to show. We now solve for $\eps'$ and $\del'$ to complete the proof for items (1) and (2) of \cref{thrm:bb priv params}. The asymptotic bounds that follow use the assumptions $\eps\leq 1$, $\eps_\test = \eps/2$, and $\failtest = \del/30$; we work out these bounds in detail below.

    To solve for $\eps'$ and $\del'$, we begin with the exact expressions for $\eps$ and $\del$ in item (1). Using our assumed bounds on $\eps$, $\eps_\test$, and $\failtest$, we obtain
    \[
        \eps = \Theta\left( \eps' \cdot (D'+\ell) \right),
    \]
    so we have
    \begin{equation}
    \label{eq:eps prime val ne}
        \eps' = \Theta\left( \frac{\eps}{D'+\ell} \right).
    \end{equation}
    We also have
    \begin{align*}
        \del &= (1 + e^{\eps_\test})e^{\eps}\failtest + \delta' \cdot e^{\eps' \cdot (D' + \ell) } \cdot\paren{D' + \ell } \\
        &\leq 2e^{2\eps}\failtest + \delta' \cdot e^{\eps' \cdot (D' + \ell) } \cdot\paren{D' + \ell } \\
        &= \Theta\left( \del'\cdot (D'+\ell) \right),
    \end{align*}
    with the final line following from $\eps\leq 1$ and $\failtest = \del/30$, so we have
    \begin{equation}
    \label{eq:del prime val ne}
        \del' = \Theta\left( \frac{\del}{D'+\ell} \right).
    \end{equation}
    
    We can expand $D' + \ell$ as
    \begin{align*}
        D' + \ell
        &= D + \Theta(\ell) \\
        &= D + \Theta \left( \frac{\log(T/(\beta\failtest))}{\eps_\test} \right).
    \end{align*}
    By combining this with Expressions~\ref{eq:eps prime val ne}~and~\ref{eq:del prime val ne}, we indeed obtain
    \[
        \eps' = \Theta \left( \frac{\eps}{B} \right)
        \mathit{\quad and\quad }
        \del' = \Theta \left( \frac{\del}{B} \right),
        \mathit{\quad where \quad }
        B = D + \frac{\log(T/(\beta\del))}{\eps}.
    \]

    From the exact expressions in item (2), we can similarly obtain $\eps' = \Theta \left( \eps/\ell \right)$ and $\del' = \Theta \left( \del/\ell \right)$. We can then write $\ell = \Theta \left( \log(T/(\beta\failtest)) / \eps_\test \right)$, which gives us
    \[
        \eps' = \Theta \left( \frac{\eps}{B} \right)
        \mathit{\quad and\quad }
        \del' = \Theta \left( \frac{\del}{B} \right),
        \mathit{\quad where \quad }
        B = \frac{\log(T/(\beta\del))}{\eps}.
    \]
    This completes the proof of items (1) and (2).

    \textbf{Proof of item (3) in \cref{thrm:bb priv params}.}
    Let $\gs$ be a graph stream of length $T$, and let $t'\in[T]$ be the largest value such that $\gs_{[t']}$ has maximum degree at most $D$, setting $t' = 0$ if there is no such value.
    This proof mostly follows from item (2) of \cref{thrm:combined-sep-svt}. Let $\ell$ and $D'$ be defined as in \cref{alg:bb}, so $\ell = \ceil{8\ln(T/(\beta \failtest  ))/\eps_\test}$ and $D' = D + \ell$. By \cref{cor:dtg meaning}, we have $\DistToGraph_{D',\ell}(\flatt{\gs_{[t]}}) \geq \ell$, so for $q_t = -1\cdot \DistToGraph_{D',\ell}(\flatt{\gs_{[t]}})$ we have $q_t \leq -\ell$.

    Equivalently, where we set $\tau = -8\ln(1/\failtest)/\eps_\test$ (as in \cref{alg:bb}), we have
    \begin{align*}
        q_t
        &\leq -\ell \\
        &\leq -8\ln(T/(\beta \failtest))/\eps_\test \\
        &= -8\ln(T/\beta)/\eps_\test + \tau \\
        &= 8\ln(\beta/T)/\eps_\test + \tau.
    \end{align*}

    By item (2) of \cref{thrm:combined-sep-svt}, then, in \cref{alg:bb} $\svt$ outputs $\Below$ on $q_1,\ldots, q_{t'}$ with probability at least $1-\beta$, so we also have that \cref{alg:bb} outputs the result from $\mech_{D'}$ for all $t\leq t'$ with probability at least $1-\beta$.

    \textbf{Proof of item (4) in \cref{thrm:bb priv params}.} On a given input stream $\gs$, let $R_{[t]}$ be the runtime of $\mech_{D'}$ through time $t$ (on $\gs_{[t]}$), and let $n_{[t]}$ and $m_{[t]}$ be the number of nodes and edges in the graph stream through time step $t$.
    At each time step, \cref{alg:bb} does several constant-time operations, an update to the projection, a call to $\mech_{D'}$, and a call to $\test$. 
    The projection algorithm in \cref{alg:time aware projection} has runtime linear in the number of (new) nodes and edges, so the runtime of this through $t$ time steps is $O(n_{[t]} + m_{[t]})$. The $t$ calls to $\mech_{D'}$ have runtime $R_{[t]}$. By \cref{lem:dtg-properties}, $\DistToGraph$ can be updated in time linear in the number of new nodes and edges, so each call to $\svt$ takes either constant time or time linear in the number of new nodes and edges. Therefore, the overall runtime of \cref{alg:bb} through time $t\in[T]$ is $R_{[t]} + O(n_{[t]} + m_{[t]} + t)$.

    Regarding space complexity, the projection algorithm need only track the degree of each node, which takes space $O(n_{[t]})$; the new nodes and edges that are added to the projection can be passed to the base algorithm $\mech_{D'}$ at each time step, so this does not take additional space. By the proof of \cref{lem:dtg-properties}, we see that the algorithm for computing $\test$ only requires a constant number of counters, along with two hash tables that each have at most $O(n_{[t]})$ entries. Therefore, our algorithm requires $O(n_{[t]})$ additional space as compared to $\mech_{D'}$.
\end{proof}

\fi

\section{\texorpdfstring{{\boldmath Optimal Algorithms for $\edges$, $\triangles$, $\connectedcomps$, $\kstars$}}{Optimal algorithms for edges, triangles, connected components, k-stars}}\label{sec:opt algs}
\label{sec:optimal algs}

We now use our transformation in \cref{alg:bb} to convert restricted edge- and node-DP algorithms for several fundamental problems into node-DP algorithms that achieve the same asymptotic error as the analogous restricted node-DP algorithms given by \cite{SLMVC18,FHO21}, up to lower-order terms.
\cref{table:tpdp} shows the additive error for privately counting edges ($\edges$), triangles ($\triangles$), $k$-stars ($\kstars$), and connected components ($\connectedcomps$), and privately releasing degree histograms\footnote{The degree histogram $\deghist(G)$ for a graph $G$ with maximum degree at most $D$ is the $(D+1)$-element vector $(a_0,\ldots,a_D) \in \R^{D+1}$, where $a_i$ is the number of nodes with degree $i$ in $G$.}
($\deghist$) of the input graph stream.
Moreover, for $\edges$, $\triangles$, $\kstars$, and $\connectedcomps$, the accuracy of our algorithms is asymptotically optimal, up to lower order terms and polylogarithmic factors.
\noakland{We prove the upper bounds on error in \cref{sec:upper bds} and the lower bounds in \cref{sec:lower-bounds-error}.}

For $\edges, \triangles, \kstars,$ and $\deghist$, the errors for our transformation follow from substituting the $\eps'$ term from \cref{thrm:bb priv params} into the error bounds for the restricted edge-DP algorithms of \cite{FHO21} (algorithms with slightly worse lower-order terms follow from using restricted node-DP algorithms). The bound for $\connectedcomps$ follows from a new edge-DP algorithm based on the binary tree mechanism. \yoakland{This algorithm and the accompanying proof appear in the full version.}

By item (1a) of \cref{thrm:combined-stab}, we can also immediately obtain $(\eps,0)$-edge-DP algorithms from restricted edge-private algorithms by using the projection $\projo$ and running the corresponding restricted edge-private algorithm with $\eps'=\eps/3$ on the projection. This technique gives the first algorithms for $\triangles$, $\kstars$, and $\connectedcomps$ that are edge-private under continual observation for all graphs and, for graph streams with maximum degree at most $D$, have asymptotically optimal accuracy (up to polylogarithmic factors).

The lower bounds in \yoakland{the table }\noakland{\cref{sec:lower-bounds-error} }follow by reductions from a version of the $\Omega(\log T/\eps)$ lower bound for binary counting from \cite{DNPR10}, modified by us for the approximate-DP setting\noakland{~(\cref{thrm:bin-count-lower-bd})}. \yoakland{Proofs of these lower bounds appear in the full version.}

\ifnum\oakland=0
\subsection{Upper Bounds on Error}
\label{sec:upper bds}

We first present our upper bounds on additive error.
For $\edges$, $\triangles$, $\kstars$, and $\deghist$, the errors for our transformation roughly follow by substituting the $\eps'$ term from \cref{thrm:bb priv params} into the error bounds for restricted edge-DP algorithms of \cite{FHO21} (with slight modifications for the $\linf$ error setting). For $\connectedcomps$, the error follows from a new edge-DP algorithm based on the binary tree mechanism; this algorithm and the accompanying proof are described in \cref{lemma:edge-sens,thrm:bin tree for diff seq}.

The errors for our edge-private algorithms follow from item (1a) of \cref{thrm:combined-stab}, which says we can achieve $(\eps,0)$-edge-DP by roughly running the algorithms of \cite{FHO21} for $\triangles,\kstars,\deghist$ and $\triangles$ (modified for the $\linf$ error setting), and our algorithm for $\connectedcomps$, with privacy parameter $\eps' = \eps/3$ on the output of $\projo$.

\begin{theorem}[Accuracy of node-private algorithms]
\label{thrm:acc of applications}
    Let $\eps\in(0,1]$, $\del > 0$, $D\in\N$, $T\in\N$ where $T\geq 2$,
    and $\gs$ be a length-$T$, insertion-only graph stream. 
    There exist $(\eps,\del)$-node-DP algorithms for the following problems whose error is at most $\alpha$, with probability $0.99$, for all times steps $t\in[T]$ where $\gs_{[t]}$ has maximum degree at most $D$:
    \begin{enumerate}[leftmargin=*]
        \item $\edges$, 
        $
            \alpha = O\left(\left(D + \tfrac{1}{\eps}\log \frac {T}{\delta} \right) \frac{\log^{5/2} T}{\eps}\right).
        $
        \item $\triangles$, 
        $
            \alpha = O\left(\left(D^2 + \tfrac{1}{\eps^2}\log^2 \frac {T}{\delta} \right) \frac{\log^{5/2} T}{\eps}\right).
        $
        \item $\connectedcomps$, 
        $
            \alpha = O\left(\left(D + \tfrac{1}{\eps}\log \frac {T}{\delta} \right) \frac{\log^{5/2} T}{\eps}\right).
        $
        \item $\kstars$, 
        $
            \alpha = O\left(\left(D^k + \tfrac{1}{\eps^k}\log^k \frac {T}{\delta} \right) \frac{\log^{5/2} T}{\eps}\right).
        $
        
        \item $\deghist$, 
        $
            \alpha = \Otilde\left(\left(D^2 + \tfrac{1}{\eps^2}\log^2 \frac {T}{\delta} \right) \frac{\log^{5/2} T}{\eps}\right)
        $
        (maximum error over histogram entries and time steps).
    \end{enumerate}
    As is standard, the $\Otilde$ notation suppresses terms that are logarithmic in the argument---in this case, it suppresses $\log D$ and $\log(\frac{1}{\eps}\log \frac{T}{\del})$.
\end{theorem}

\begin{theorem}[Accuracy of edge-private algorithms]
\label{thrm:acc of applications edge}
    Let $\eps>0$, $D\in\N$, $T\in\N$ where $T\geq 2$,
    and $\gs$ be a length-$T$, insertion-only graph stream. 
    There exist $(\eps,0)$-edge-DP algorithms for the following problems whose error is at most $\alpha$, with probability $0.99$, for all times steps $t\in[T]$ where $\gs_{[t]}$ has maximum degree at most $D$:
    \begin{enumerate}[leftmargin=*]

        \item $\triangles$, 
        $
            \alpha = O\left(\dfrac{D\log^{5/2} T}{\eps}\right).
        $
        
        \item $\connectedcomps$, 
        $
            \alpha = O\left(\dfrac{\log^{5/2} T}{\eps}\right).
        $
        
        \item $\kstars$, 
        $
            \alpha = O\left( \dfrac{D^{k-1}\log^{5/2}T}{\eps} \right).
        $
        
        \item $\deghist$, 
        $
            \alpha = \Otilde\left(\dfrac{D\log^{5/2} T}{\eps} \right)
        $
        (maximum error over histogram entries and time steps).
    \end{enumerate}
    As is standard, the $\Otilde$ notation suppresses terms that are logarithmic in the argument---in this case, it suppresses $\log D$.
\end{theorem}

\newcommand{\incseq}[1]{\text{\sf inc}_{#1}}
\newcommand{\incsens}{\text{\sf IncEdgeSens}}

As described above, these upper bounds rely on the accuracy of restricted edge-private algorithms for the associated problems, all of which are based on the binary tree mechanism of \cite{DNPR10,CSS11} and result from \cref{lemma:edge-sens,thrm:bin tree for diff seq}.

These binary tree-based algorithms all follow from bounding the $\ell_1$ sensitivity of the sequence of \emph{increments} of a given function $f$ over the stream $\gs$, first considered explicitly by \cite{SLMVC18}.
Let $x\in \mX^T$.
Where $f:\mX^T \to (\R^d)^T$ returns an output in $\R^d$ at each time step, let $\incseq{f}(\gs):\mX^T \to (\R^d)^T$ return, for each time step $t\in[T]$, the increment $f(x)_t - f(x)_{t-1}$, where we define $f(x)_0 = 0^d$.
We define $\incsens_{D}(f)$ as the $\ell_1$ Lipschitz constant of $\incseq{f}$ when restricted to $D$-bounded, edge-neighboring graph streams:
\[
    \incsens_{D}(f) = \sup_{\substack{\gs \nedge \gs' \, \text{with} \\ \text{max degree }\leq D}} \| \incseq{f}(\gs) - \incseq{f}(\gs')\|_1 \, .
\]
We let $\incsens(f)$ (without a value for $D$) denote the sensitivity of $\incseq{f}$ over all neighboring graphs streams.

The tree mechanism of \cite{DNPR10,CSS11} produces an edge-DP approximation to the sequence of  values of $f$  when the noise is scaled to $\incsens_{D}(f)$, and produces an algorithm with  $D$-restricted edge-privacy when its noise is scaled to $\incsens_{D}(f)$. The following lemma is folklore; our formulation is a variation on statements in \cite{DNPR10,CSS11,SLMVC18,FHO21}.

\begin{lemma}[following \cite{DNPR10,CSS11,SLMVC18,FHO21}]
\label{thrm:bin tree for diff seq}
    Let $\eps>0$, $D\in \N$, $T\in \N$ where $T\geq 2$, and $d\in \N$.
    Let $f:\mGS^T\to (\R^d)^T$ return a ($d$-dimensional) output in $\R^d$ at each time step $t\in [T]$.

    There exists a $D$-restricted $(\eps,0)$-edge-DP algorithm that, with probability at least $1-\beta$, has $\linf$ error 
    \[
        \incsens_D(f)
        \cdot
        O\paren{\frac{\log^{3/2} T}{\eps}
        \cdot \log(dT/\beta)}
    \]
    on every dimension of the output.
    
    The analogous statement holds for (unrestricted) differential privacy when its noise scaled to $\incsens(f)$.
    Additionally, analogous statements hold for restricted and unrestricted node sensitivity and node-DP.
\end{lemma}
\begin{proof}[Proof of \cref{thrm:bin tree for diff seq}]
    By \cite[Corollary 14]{FHO21}, which uses the binary tree mechanism of \cite{CSS11,DNPR10} and the technique of difference sequences introduced by \cite{SLMVC18}, there exists a $D$-restricted $(\eps,0)$-edge-DP algorithm that, with probability at least $1-\beta'$, has error at most
    \(
        \incsens_D(f)
        \cdot
        O\paren{\frac{\log^{3/2} T}{\eps}
        \cdot \log(1/\beta')}
    \)
    per dimension at each time step.

    By substituting $\beta' = \frac{\beta}{dT}$ and taking a union bound over all $d$ coordinates and $T$ time steps, with probability at least $1-\beta$ the $\linf$ error over all $d$ coordinates and $T$ time steps is at most
    \(
        \incsens_D(f)
        \cdot
        O\paren{\frac{\log^{3/2} T}{\eps}
        \cdot \log(dT/\beta)}
    \),
    which is what we wanted to show (with analogous statements for $\incsens(f)$ and unrestricted $(\eps,0)$-edge-DP; and for node-DP).
\end{proof}

\begin{lemma}[Edge sensitivity bounds]
\label{lemma:edge-sens}
    For all $D\in \N$, the following bounds on incremental edge-sensitivity $\incsens_D$ and $\incsens$ hold (with some due to \cite{SLMVC18,FHO21}):
    \begin{enumerate}
        \item $\incsens(\edges) = 1 = O(1)$.
        \item $\incsens_D(\triangles) = D - 1 = O(D)$.
        \item $\incsens(\connectedcomps) = 2 = O(1)$.
        \item $\incsens_D(\kstars) = 2\binom{D-1}{k-1} = O\paren{D^{k-1}}$.
        \item $\incsens_D(\deghist) 
        = 8D - 4 = O(D)$.
    \end{enumerate}
    Except for $\deghist$, these equalities hold even if we restrict our attention to streams of length $T=1$. For $\deghist$, the lower bound requires $T \geq 2D$ (but the upper bound holds regardless of $T$).
\end{lemma}

\begin{proof}
    With the exception of item (3), proofs of items similar to those listed above can be found in \cite[Lemma 15]{FHO21}, though we tighten some of the sensitivities
    (analyses for $D$-restricted node sensitivity first appear in \cite{SLMVC18}). 
    Specifically, for items (2) and (5) we tighten the exact sensitivities and provide strict equalities (the asymptotics are unchanged), and for item (4) we tighten the asymptotic expression by a factor of $D$ (the exact expression in fact remains unchanged). We prove item (3) last since its proof is new and the most involved.

    \textbf{Item (1).} The addition or removal of an edge changes the increment sequence of edge counts $\incseq{\edges}$ by at most one \cite{FHO21}.

    \textbf{Item (2).} Each edge is part of at most $D-1$ triangles (and can, in fact, be part of $D-1$ triangles), so adding or removing an edge will change the sequence of increments by at most $D-1$. 

    \textbf{Item (4).} Each edge is part of at most $2 \binom{D}{k}$ different $k$-stars
    (and, if the edge is between two nodes with degree $D$, will be part of exactly $2 \binom{D}{k}$ different $k$-stars)
    so removing the edge will instead result in $2 \binom{D-1}{k}$ different $k$-stars. Therefore, the sequence of increments will change by at most $2 \paren{ \binom{D}{k} - \binom{D-1}{k} }$ \cite{FHO21}. We additionally observe $2 \paren{ \binom{D}{k} - \binom{D-1}{k} } = 2\binom{D-1}{k-1} = O\paren{D^{k-1}}$.

    \textbf{Item (5).} We first consider how an additional edge $\eplus = \set{u,v}$ in $\gs'$ (as compared to an edge-neighboring stream $\gs$) affects the degree of $u$; we then apply a similar argument for the degree of $v$. At each time step, the degree of $u$ is higher by at most one in $\gs'$ as compared to $\gs$. In general, this means that the increment to each bucket due to a change in the degree of $u$ will be earlier, and the decrement to each bucket due to a change in the degree of $u$ will also be earlier.
    There are several exceptions to this idea: the increment to the bucket for degree 0 cannot be earlier, and there cannot be a decrement to the bucket for degree $D$.
    Additionally, there cannot be an increment in $\gs$ for the bucket for degree $D$ (since otherwise $u$ would have degree $D+1$ in $\gs'$); similarly, there cannot be a decrement for the bucket for degree $D-1$ in $\gs$.

    Therefore, considering buckets not equal to $0$, $D-1$, or $D$, there are $(D-2)$ differently positioned increments and decrements---that is, the $\lone$ sensitivity of the sequence of increments for these buckets is $4(D-2)$; and for $D$ there is $1$ change, for 0 there are $2$ changes (when $D-1 > 0$), and for $D-1$ there is one less change than it would otherwise have (i.e., $1$ change for $D-1 = 0$ and $3$ changes otherwise). Therefore, the $\lone$ sensitivity of the sequence of increments is $4D - 2$ due to the change in the degree of $u$. A similar argument applies for changes due to the degree of $v$, which gives us a sensitivity of $2 \cdot (4D - 2) = 8D-4$.
    
    Note that this argument is tight when $\eplus$ arrives at a time step that is both (1) after $u$ and $v$ and (2) before all other edges, and when at most one edge incident to $u$ arrives at each time step (likewise for $v$).

    \textbf{Item (3).} Let $\connectedcomps^+$ denote $\incseq{\connectedcomps}$.
    Without loss of generality, let $\gs'$ be a graph stream that, as compared to $\gs$, has one additional edge $\eplus = \{u,v\}$ that is inserted at time $t_1$.

    When bounding $|| \connectedcomps^+(\gs) - \connectedcomps^+(\gs') ||_1$, we have to consider the following three cases. In $\gs$:
    \begin{enumerate}
        \item \textbf{(Case 1)} there is never a path between $u$ and $v$.
        \item \textbf{(Case 2)} $t_0 \leq t_1$ is the first time step at which there exists a path between $u$ and $v$.
        \item \textbf{(Case 3)} $t_3 > t_1$ is the first time step at which there exists a path between $u$ and $v$.
    \end{enumerate}
    We prove each case below.
    
    \textbf{Case 2:} The outputs from $\connectedcomps(\gs)$ and $\connectedcomps(\gs')$ are identical, so the outputs from $\connectedcomps^+(\gs)$ and $\connectedcomps^+(\gs')$ are also identical, so we have $||\connectedcomps^+(\gs) - \connectedcomps^+(\gs') ||_1 = 0$. We now consider the remaining cases.

    \textbf{Case 1:} For all $t_0 < t_1$ (that is, prior to the time step at which the graph streams differ) we have
    $\connectedcomps^+(\gs)_{t_0} = \connectedcomps^+(\gs')_{t_0}$. When the difference in graph streams occurs at $t_1$, we have $\connectedcomps^+(\gs)_{t_1} = \connectedcomps^+(\gs')_{t_1} + 1$. Only the components containing $u$ and $v$ are affected by this edge. By the assumption that there is never a path between $u$ and $v$ in $\gs$, no future edges affect the components containing $u$ and $v$, so for all $t_2 > t_1$ we have $\connectedcomps^+(\gs)_{t_2} = \connectedcomps^+(\gs')_{t_2}$. Therefore, in case 1 we have
    \[
        || \connectedcomps^+(\gs) - \connectedcomps^+(\gs') ||_1 \leq 1.
    \]

    \textbf{Case 3:} As above, for all $t_0 < t_1$ we have
    $\connectedcomps^+(\gs)_{t_0} = \connectedcomps^+(\gs')_{t_0}$. When the difference in graph streams occurs at $t_1$, we have $\connectedcomps^+(\gs)_{t_1} = \connectedcomps^+(\gs')_{t_1} + 1$. Only the components containing $u$ and $v$ are affected by the inclusion of edge $\eplus$. By the assumption that, for all $t_2\in [t_1,t_3)$, there is never a path between $u$ and $v$ in $\gs$, no edges arriving at such times $t_2$ affect the components containing $u$ and $v$, so we have $\connectedcomps^+(\gs)_{t_2} = \connectedcomps^+(\gs')_{t_2}$.
    
    At time $t_3$, the edge forming a path between $u$ and $v$ arrives, which causes the number of connected components in $\gs$ to decrease by 1 as compared to $\gs'$, so $\connectedcomps^+(\gs)_{t_3} = \connectedcomps^+(\gs')_{t_3}$. For all $t_4 > t_3$, the changes in counts will be the same for both graph streams, so $\connectedcomps^+(\gs)_{t_4} = \connectedcomps^+(\gs')_{t_4}$. Therefore, in case 3 we have
    \[
        || \connectedcomps^+(\gs) - \connectedcomps^+(\gs') ||_1 \leq 2,
    \]
    which completes the proof.  
\end{proof}

We now move to the proofs of \cref{thrm:acc of applications,thrm:acc of applications edge}. \cref{thrm:acc of applications edge} quickly follows from running the restricted edge-private algorithms that result from the sensitivities in \cref{lemma:edge-sens} and the binary tree mechanism in \cref{thrm:bin tree for diff seq} on the projection returned by $\projo$, and either running the algorithm directly on the input for the problems with unrestricted sensitivity bounds, or running the algorithm with $\eps' = \eps/3$ on the projection of $\projo$ for problems with restricted sensitivity bounds. 
\cref{thrm:acc of applications} follows from running \cref{alg:bb} with access to the relevant restricted edge-private algorithm that results from \cref{lemma:edge-sens,thrm:bin tree for diff seq}, and then determining the associated accuracy of that algorithm given the value of $\eps'$ specified by \cref{thrm:bb priv params}.

\begin{proof}[Proof of \cref{thrm:acc of applications edge}]
    This proof follows immediately from combining item (1a) of \cref{thrm:combined-stab}, which describes the edge-to-edge stability of \cref{alg:time aware projection} that includes edges in the projection according to original degrees (we use $\projo$ to denote this projection), with the bounds on accuracy for edge-private algorithms that follow from the sensitivities in \cref{lemma:edge-sens} and the binary tree mechanism in \cref{thrm:bin tree for diff seq}. More specifically, to obtain $(\eps,0)$-edge-DP under continual observation for the problems with $D$-restricted edge sensitivity, we run the $D$-restricted edge-private algorithms yielding these accuracy bounds with $\eps/3$ on the graph stream output by the projection algorithm $\projo$; and to obtain $(\eps,0)$-edge-DP under continual observation for the problems with unrestricted edge sensitivity, we run the resulting algorithm with $\eps$.
\end{proof}

\begin{proof}[Proof of \cref{thrm:acc of applications}]
    We begin by calculating the value of $\eps'$ with which \cref{alg:bb} will run the (restricted) edge-private algorithm. For these calculations, set $\eps_\test = \eps/2$ and $\failtest = \del / 30$. We also assume constant $\beta$. To have $(\eps,\del)$-node-DP, \cref{thrm:bb priv params} tells us that \cref{alg:bb} will set
    \begin{align*}
        \eps'
        &= \Theta\left( \frac{\eps}{D + \frac{\log(T/(\beta \del))}{\eps}} \right) \\
        &= \Theta \left( \frac{\eps}{D} + \frac{\eps^2}{\log(T/(\beta\del))} \right).
    \end{align*}
    For constant $\beta$, we can simplify this further to
    \[
        \eps' = \Theta \left( \frac{\eps}{D} + \frac{\eps^2}{\log(T/\del)} \right).
    \]

    To compute accuracy, we will also need to know the value of $D'$. Using the conditions stated above on our variables, we can simplify $D'$ as follows:
    \begin{align*}
        D'
        &= D + \ell \\
        &= D + \Theta\paren{ \dfrac{\log (T/(\beta\failtest))}{\eps} } \\
        &= D + \Theta\paren{ \frac{\log (T/\del) }{\eps} }.
    \end{align*}

    \textbf{Proof of items (1) and (3).} We substitute the above value of $\eps'$ into the accuracy bounds that follow from \cref{lemma:edge-sens,thrm:bin tree for diff seq}.
    We note that the accuracy expressions for $\edges$ and $\connectedcomps$ are the same, so we begin by proving the upper bounds on accuracy for these functions.
    By \cref{lemma:edge-sens,thrm:bin tree for diff seq}, with probability $1-\beta$, the edge-private algorithms for $\edges$ and $\connectedcomps$ have additive error at most $O((\log^{3/2}T)\cdot (\log T + \log(1/\beta))/\eps')$.
    Since $\beta$ is constant we can ignore the $\log(1/\beta)$ term, so substituting for $\eps'$ gives us
    \begin{align*}
        \frac{\log^{5/2}T}{\eps'}
        &= \Theta\left(  \left( \frac{D}{\eps} + \frac{\log(T/\del)}{\eps^2} \right) \cdot \log^{5/2}T \right) \\
        &= \Theta\left(  \left( D + \frac{\log(T/\del)}{\eps} \right) \cdot \frac{\log^{5/2}T}{\eps} \right),
    \end{align*}
    which is the additive error we wanted to show.

    \textbf{Proof of item (2).} By \cref{lemma:edge-sens,thrm:bin tree for diff seq}, with probability $1-\beta$, the edge-private algorithm for $\triangles$ has additive error at most $O(D'\cdot(\log^{3/2}T)\cdot (\log T + \log(1/\beta))/\eps')$. 
    Since $\beta$ is constant we can ignore the $\log(1/\beta)$ term. Substituting for $\eps'$ and $D'$ gives us
    \begin{align*}
        \frac{D' \log^{5/2} T}{\eps'}
        &=
        \Theta\paren{  \paren{  D + \frac{\log(T/\del)}{\eps} }\cdot \log^{5/2} T \cdot \paren{\frac{D}{\eps} + \frac{\log T /\del}{\eps^2}}  } \\
        &=
        \Theta\paren{  \paren{  D + \frac{\log(T/\del)}{\eps} }^2 \cdot \frac{\log^{5/2} T}{\eps} } \\
        &=
        \Theta\paren{  \paren{  D^2 + \frac{\log^2(T/\del)}{\eps^2} } \cdot \frac{\log^{5/2} T}{\eps} },
    \end{align*}
    which is the additive error we wanted to show.

    \textbf{Proof of item (4).} By \cref{lemma:edge-sens,thrm:bin tree for diff seq}, with probability $1-\beta$, the edge-private algorithm for $\kstars$ has additive error at most $O(D'^{k-1}\cdot(\log^{3/2}T)\cdot (\log T + \log(1/\beta))/\eps')$. 
    Since $\beta$ is constant we can ignore the $\log(1/\beta)$ term. Substituting for $\eps'$ and $D'$ gives us
    \begin{align*}
        \frac{D'^{k-1} \log^{5/2} T}{\eps'}
        &=
        \Theta\paren{  \paren{  D + \frac{\log(T/\del)}{\eps} }^{k-1}\cdot \log^{5/2} T \cdot \paren{\frac{D}{\eps} + \frac{\log T /\del}{\eps^2}}  } \\
        &=
        \Theta\paren{  \paren{  D + \frac{\log(T/\del)}{\eps} }^k \cdot \frac{\log^{5/2} T}{\eps} } \\
        &=
        \Theta\paren{  \paren{  D^k + \frac{\log^k(T/\del)}{\eps^k} } \cdot \frac{\log^{5/2} T}{\eps} },
    \end{align*}
    which is the additive error we wanted to show.

    \textbf{Proof of item (5).} By \cref{lemma:edge-sens,thrm:bin tree for diff seq}, with probability $1-\beta$, the edge-private algorithm for $\deghist$ has additive error at most $O(D'\cdot(\log^{3/2}T)\cdot (\log T + \log(1/\beta) + \log D')/\eps')$ on every bin of the histogram. Since $\beta$ is constant we can ignore the $\log(1/\beta)$ term. Substituting for $\eps'$ and $D'$ gives us
    \begin{align*}
        \frac{D'\cdot(\log^{3/2}T)\cdot (\log T + \log D')}{\eps'}
        &= 
        \Thetatilde\paren{  \paren{  D + \frac{\log(T/\del)}{\eps} }\cdot \log^{5/2} T \cdot \paren{\frac{D}{\eps} + \frac{\log T /\del}{\eps^2}}  } \\
        &=
        \Thetatilde\paren{  \paren{  D + \frac{\log(T/\del)}{\eps} }^2 \cdot \frac{\log^{5/2} T}{\eps} } \\
        &=
        \Thetatilde\paren{  \paren{  D^2 + \frac{\log^2(T/\del)}{\eps^2} } \cdot \frac{\log^{5/2} T}{\eps} },
    \end{align*}
    where, as is standard, the $\Thetatilde$ notation suppresses terms that are logarithmic in the argument---in this case, it suppresses $\log D$ and $\log(\frac{1}{\eps}\log \frac{T}{\del})$. This is the additive error we wanted to show.
\end{proof}

\subsection{Lower Bounds on Error}
\label{sec:lower-bounds-error}

We now present several lower bounds on the error necessary for privately releasing graph statistics in the approximate DP setting, for $\del = O(1/T)$. The proofs of these bounds are based on a reduction from binary counting. For $\edges$ and $\deghist$, our reductions from binary counting are similar to the constructions used by \cite{FHO21}. Our reductions for $\triangles,\connectedcomps$, and $\kstars$ are different from those of \cite{FHO21} and yield stronger lower bounds.

\begin{theorem}[Lower bounds for node-private algorithms]
\label{thrm:lower bd node}
    For sufficiently large $T\in\N$, and all $\eps > 0,\del = O(1/T)$, and $D\in\N$ 
    (with $D \geq 2$ for $\triangles$),\footnote{A graph with maximum degree $D = 1$ contains no triangles.}
    consider a mechanism $\mM$ for each of the problems below that runs on length-$T$ graph streams with maximum degree at most $D$. If $\mM$ satisfies $(\eps,\del)$-node-DP and $(\alpha,T)$-accuracy for the specified task, then its additive error must be lower bounded in the following way:
    \begin{enumerate}[leftmargin=*]
        \item $\edges$,
        $
            \alpha = \Omega \paren{ \min \set{ \dfrac{D \log T}{\eps}, D T } }.
        $
        
        \item $\triangles$,
        $
            \alpha = \Omega \paren{ \min \set{ \dfrac{D^2 \log T}{\eps}, D^2 T } }.
        $
        
        \item $\connectedcomps$,
        $
            \alpha = \Omega \paren{ \min \set{ \dfrac{D \log T}{\eps}, D T } }.
        $

        \item $\kstars$,
        $
            \alpha = \Omega \paren{ \min \set{ \dfrac{D^k \log T}{\eps}, D^k T } }.
        $
        
        \item $\deghist$,
        $
            \alpha = \Omega \paren{ \min \set{ \dfrac{D \log T}{\eps}, D T } }
        $
        (maximum error over histogram entries and time steps).
    \end{enumerate}
\end{theorem}

\begin{theorem}[Lower bounds for edge-private algorithms]
\label{thrm:lower bd edge}
    For sufficiently large $T\in\N$, and all $\eps > 0,\del = O(1/T)$, and $D\in\N$
    (with $D \geq 2$ for $\triangles$),
    consider a mechanism $\mM$ for each of the problems below that runs on length-$T$ graph streams with maximum degree at most $D$. If $\mM$ satisfies $(\eps,\del)$-edge-DP and $(\alpha,T)$-accuracy for the specified task, then its additive error must be lower bounded in the following way:
    \begin{enumerate}[leftmargin=*]
        \item $\edges$,
        $
            \alpha = \Omega \paren{ \min \set{ \dfrac{\log T}{\eps}, T } }.
        $
        
        \item $\triangles$,
        $
            \alpha = \Omega \paren{ \min \set{ \dfrac{D\log T}{\eps}, D T } }.
        $
        
        \item $\connectedcomps$,
        $
            \alpha = \Omega \paren{ \min \set{ \dfrac{\log T}{\eps}, T } }.
        $

        \item $\kstars$,
        $
            \alpha = \Omega \paren{ \min \set{ \dfrac{D^{k-1} \log T}{\eps}, D^{k-1} T } }.
        $
        
        \item $\deghist$,
        $
            \alpha = \Omega \paren{ \min \set{ \dfrac{\log T}{\eps}, T } }
        $
        (maximum error over histogram entries and time steps).
    \end{enumerate}
\end{theorem}

We now present a lower bound on the error needed to privately release a binary count under continual observation. Note that, although this lower bound has previously been stated only for pure DP \cite{DNPR10}, we show that it also holds for the approximate-DP setting where $\del = O(1/T)$.
To prove \cref{thrm:bin-count-lower-bd}, we use a packing argument, introduced by \cite{HT10,BeimelBKN14}. We use the version given by \cite{Vadhan17}.

\begin{theorem}[Approximate-DP lower bound for binary counting]
\label{thrm:bin-count-lower-bd}
    Let $f:\{0,1\}^T\to \R^T$ take a $T$-element binary stream $\gs = \gs_1,\ldots, \gs_T$ as input and release, at each time step $t\in[T]$, the sum $\sum_{i\in[t]} \gs_i$. Let $\mM$ be $(\alpha,T)$-accurate for $f$ and $(\eps,\del)$-DP under continual observation. If $\eps>0$, $\del = O(1/T)$, and $T$ is sufficiently large, then $\mM$ must have additive $\linf$ error at least
    \[
        \alpha = \Omega \paren{ \min \set{ \frac{\log T}{\eps}, T } }.
    \]
\end{theorem}

\begin{lemma}[Packing lower bound \cite{HT10,BeimelBKN14,Vadhan17}]
\label{lemma:pack-arg}
    Let $\mC\subseteq \mX^n$ be a collection of datasets all at Hamming distance at most $m$ from some fixed dataset $x_0\in\mX^n$, and let $\set{\mG_x}_{x\in\mC}$ be a collection of disjoint subsets of $\mY$. If there is an $(\eps,\del)$-DP mechanism $\mM:\mX^n\to \mY$ such that $\Pr[\mM(x)\in \mG_x]\geq p$ for every $x\in \mC$, then
    \[
        \frac{1}{|\mC|} \geq p\cdot e^{-m\cdot \eps} - m \del.
    \]
\end{lemma}

\begin{proof}[Proof of \cref{thrm:bin-count-lower-bd}]
    For the packing argument, we construct a collection of datasets similar to those used by \cite{DNPR10}. Let $f:\set{0,1}^T \to \R^T$ be the function for binary counting described in \cref{thrm:bin-count-lower-bd}.

    We first construct a collection $\mC$ of $k = \floor{T/m}$ datasets in $\set{0,1}^T$. We will construct all of these datasets to be at Hamming distance $m$ from the all 0s dataset $x_0 = 0^T$.

    The collection $\mC$ contains the following $k$ datasets. For each $i\in [k]$, construct the $\ord{i}$ dataset $x_i = 0^{m\cdot (i-1)}\circ 1^m \circ 0^{T-m\cdot i}$. That is, each dataset is a set of $k-1$ blocks of $m$ consecutive 0s,\footnote{The end of the stream may be padded with additional 0s to ensure it has length $T$.} and one block of $m$ consecutive 1s, where dataset $i$ has its $\ord{i}$ block contain the consecutive 1s.
    
    We see that all of these datasets are Hamming distance $m$ from the all 0s dataset $x_0$. We also see that, for all $x\neq x'\in \mC$,
    \begin{equation}
    \label{eq:disjoint-sets}
        || f(x) - f(x') ||_\infty > m/2.
    \end{equation}
    Let $\mG_{x_i}$ be the closed $\linf$ ball of radius $\alpha = m/2$ around $f(x_i)$. By (\ref{eq:disjoint-sets}) we see that the collection of sets $\set{\mG_{x_i}}_{x_i \in \mC}$ is disjoint.

    If $\mM$ is $(\eps,\del)$-DP and $(\alpha, T)$-accurate for the binary counting function $f$, then on input $x_i$ it must give an answer in $\mG_{x_i}$ with probability $p\geq 0.99$. We can now use a packing argument to solve for $\alpha$. \cref{lemma:pack-arg} tells us
    \begin{align*}
        \frac{1}{|\mC|} &\geq p \cdot e^{-m\eps} - m\del \\
        \frac{1}{k} &\geq \frac{e^{-m\eps}}{2} - m\del \tag{$|\mC| = k$, $p\geq 0.5$} \\
        \frac{e^{-m\eps}}{2} &\leq \frac{1}{k} + m\del. 
    \end{align*}
    Recall $k = \floor{T/m}$, so $k\geq (T-m)/m$. Additionally, recall $\alpha = m/2$ and by assumption $\del = O(1/T)$. This gives us
    \begin{align*}
        e^{-\alpha\eps/2}
        &= O\paren{ \frac{m}{T - m} + m\cdot \frac{1}{T} } \\
        &= O\paren{\frac{\alpha}{T - \alpha}}.
    \end{align*}
    Taking the reciprocal of each side gives us
    \begin{align*}
        e^{\alpha\eps/2}
        &= \Omega\paren{\frac{T - \alpha}{\alpha}},
    \end{align*}
    and taking the log gives us
    \[
        \alpha \eps = \Omega\paren{\log T - \log \alpha}.
    \]
    We want to solve for $\alpha$ such that $\alpha \eps + \log \alpha = \Omega(\log T)$.
    This leaves us with two cases: $\alpha \eps = \Omega(\log \alpha)$ and $\alpha \eps = O(\log \alpha)$. In the first case, $\alpha \eps$ dominates, so we need $\alpha = \Omega(\log T / \eps)$. In the second case, $\log \alpha$ dominates, so we need $\alpha = \Omega(T)$. Therefore,
    \[
        \alpha = \Omega \paren{ \min \set{ \frac{\log T}{\eps}, T } },
    \]
    which is what we wanted to show.
\end{proof}

We now prove \cref{thrm:lower bd node,thrm:lower bd edge} through reductions from binary counting.

\begin{proof}[Proof of \cref{thrm:lower bd node}]
    Let $x\in \set{0,1}^T$ be a binary stream. Below, we describe functions that take $x$ as input and return a graph stream where the count of some feature (e.g., triangles) at each time step $t$ can be used to compute the corresponding prefix sum of the binary stream through time step $t$. Crucially, the functions described below map neighboring binary streams to node-neighboring $(D,0)$-bounded graph streams (i.e., graph streams with maximum degree at most $D$).

    We use the following notation: let $x_i$ and $\gs_i$ denote the $\ord{i}$ index of the binary stream and graph stream, respectively. Additionally, the statements below assume $\del = O(1/T)$. Note that our mappings from binary streams to $\edges$ and $\deghist$ are similar to the mappings used by \cite{FHO21}, though our other mappings are different and yield stronger lower bounds. 
    \begin{enumerate}
        \item \textbf{(Edges.)} For each time step $t\in[T]$, let $\gs_t$ contain $D$ isolated nodes. If $x_t = 1$, also include a node $v_t$ with edges to all other nodes that arrived in that time step. Therefore, each time step contains either $0$ new edges or $D$ new edges.

        Assume for contradiction that there exists an $(\eps,\del)$-node-DP mechanism $\mM$ for solving $\edges$ with $(\alpha,T)$-accuracy where $\alpha = o \paren{ \min \set{ \frac{D \log T}{\eps}, D T } }$. Then applying the transformation described above to a binary stream, running $\mM$ on the resulting stream, and dividing the output at each time step by $D$ will solve binary counting with error $\alpha = o \paren{ \min \set{ \frac{\log T}{\eps}, T } }$, which contradicts \cref{thrm:bin-count-lower-bd}.

        \item \textbf{(Triangles.)} For each time step $t\in[T]$, let $\gs_t$ contain a complete graph on $D$ nodes. If $x_t = 1$, also include a node $v_t$ with edges to all other nodes. Therefore, each time step contains either $\binom{D}{3}$ new triangles or $\binom{D+1}{3} = \binom{D}{3} + \binom{D}{2}$ new triangles.

        Assume for contradiction that there exists an $(\eps,\del)$-node-DP mechanism $\mM$ for solving $\triangles$ with $(\alpha,T)$-accuracy where $\alpha = o \paren{ \min \set{ \frac{D^2 \log T}{\eps}, D^2 T } }$. Then applying the transformation described above to a binary stream, running $\mM$ on the resulting stream, subtracting $\binom{D}{3}$ from the output at each time step and dividing this result by $\binom{D}{2} = \Theta(D^2)$ will solve binary counting with error $\alpha = o \paren{ \min \set{ \frac{\log T}{\eps}, T } }$, which contradicts \cref{thrm:bin-count-lower-bd}.

        \item \textbf{(Connected components.)} For each time step $t\in[T]$, let $\gs_t$ contain $D$ isolated nodes. If $x_t = 1$, also include a node $v_t$ with edges to all other nodes that arrived in that time step. Therefore, each time step contains either $1$ new connected component or $D$ new connected components.

        Assume for contradiction that there exists an $(\eps,\del)$-node-DP mechanism $\mM$ for solving $\connectedcomps$ with $(\alpha,T)$-accuracy where $\alpha = o \paren{ \min \set{ \frac{D \log T}{\eps}, D T } }$. Then applying the transformation described above to a binary stream, running $\mM$ on the resulting stream, and dividing the output at each time step by $D-1$ will solve binary counting with error $\alpha = o \paren{ \min \set{ \frac{\log T}{\eps}, T } }$, which contradicts \cref{thrm:bin-count-lower-bd}.

        \item \textbf{($\bm{k}$-stars.)} For each time step $t\in[T]$, let $\gs_t$ contain $D$ isolated nodes. If $x_t = 1$, also include a node $v_t$ with edges to all other nodes that arrived in that time step. Therefore, each time step contains either $0$ new $k$-stars or $\binom{D}{k}$ new $k$-stars.

        Assume for contradiction that there exists an $(\eps,\del)$-node-DP mechanism $\mM$ for solving $\kstars$ with $(\alpha,T)$-accuracy where $\alpha = o \paren{ \min \set{ \frac{D^k \log T}{\eps}, D^k T } }$. Then applying the transformation described above to a binary stream, running $\mM$ on the resulting stream, and dividing the output at each time step by $\binom{D}{k} = \Theta(D^k)$ will solve binary counting with error $\alpha = o \paren{ \min \set{ \frac{\log T}{\eps}, T } }$, which contradicts \cref{thrm:bin-count-lower-bd}.

        \item \textbf{(Degree histograms.)} For each time step $t\in[T]$, let $\gs_t$ contain $D$ isolated nodes. If $x_t = 1$, also include a node $v_t$ with edges to all other nodes that arrived in that time step. Therefore, the histogram bin for nodes of degree $1$ increases each time step by either 0 new nodes or $D$ new nodes.

        Assume for contradiction that there exists an $(\eps,\del)$-node-DP mechanism $\mM$ for solving $\deghist$ that has $(\alpha,T)$-accuracy for some bin with $\alpha = o \paren{ \min \set{ \frac{D \log T}{\eps}, D T } }$. Without loss of generality, let it be the bin for nodes of degree $1$. Then applying the transformation described above to a binary stream, running $\mM$ on the resulting stream, and dividing the output for the bin for degree 1 at each time step by $D$ will solve binary counting with error $\alpha = o \paren{ \min \set{ \frac{\log T}{\eps}, T } }$, which contradicts \cref{thrm:bin-count-lower-bd}. \qedhere
    \end{enumerate}
\end{proof}

The reductions from binary counting that we use to prove \cref{thrm:lower bd edge} are similar to the reductions above for proving \cref{thrm:lower bd node}, though they instead map neighboring binary streams to edge-neighboring graph streams.

\begin{proof}[Proof of \cref{thrm:lower bd edge}]
    As above, let $x\in \set{0,1}^T$ be a binary stream. We describe functions that take $x$ as input and return a graph stream where the count of some feature (e.g., triangles) at each time step $t$ can be used to compute the corresponding prefix sum of the binary stream through time step $t$. Crucially, the functions described below map neighboring binary streams to edge-neighboring $(D,0)$-bounded graph streams (i.e., graph streams with maximum degree at most $D$).

    We use the following notation: let $x_i$ and $\gs_i$ denote the $\ord{i}$ index of the binary stream and graph stream, respectively. Additionally, the statements below assume $\del = O(1/T)$. Note that our mappings from binary streams to $\edges$ and $\deghist$ are similar to the mappings used by \cite{FHO21}, though our other mappings are different and yield stronger lower bounds. 
    \begin{enumerate}[leftmargin = *]
        \item \textbf{(Edges.)} For each time step $t\in[T]$, let $\gs_t$ contain two isolated nodes $u_t,v_t$. If $x_t = 1$, also include an edge $\set{u_t,v_t}$. Therefore, each time step contains either $0$ new edges or $1$ new edge.

        Assume for contradiction that there exists an $(\eps,\del)$-node-DP mechanism $\mM$ for solving $\edges$ with $(\alpha,T)$-accuracy on each bin where $\alpha = o \paren{ \min \set{ \frac{\log T}{\eps}, T } }$. Then applying the transformation described above to a binary stream, running $\mM$ on the resulting stream, and releasing the result will solve binary counting with error $\alpha = o \paren{ \min \set{ \frac{\log T}{\eps}, T } }$, which contradicts \cref{thrm:bin-count-lower-bd}.

        \item \textbf{(Triangles.)} For each time step $t\in[T]$, let $\gs_t$ contain the following structure: $D-1$ nodes $w^1_t,\ldots w^{D-1}_t$, and two nodes $u_t,v_t$, each with an edge to every node of the form $w^i_t$. If $x_t = 1$, also include an edge $\set{u_t,v_t}$. Therefore, each time step contains either $0$ new triangles or $D-1$ new triangles.

        Assume for contradiction that there exists an $(\eps,\del)$-edge-DP mechanism $\mM$ for solving $\triangles$ with $(\alpha,T)$-accuracy where $\alpha = o \paren{ \min \set{ \frac{D \log T}{\eps}, D T } }$. Then applying the transformation described above to a binary stream, running $\mM$ on the resulting stream, and dividing the output at each time step by $D-1$ will solve binary counting with error $\alpha = o \paren{ \min \set{ \frac{\log T}{\eps}, T } }$, which contradicts \cref{thrm:bin-count-lower-bd}.

        \item \textbf{(Connected components.)} For each time step $t\in[T]$, let $\gs_t$ contain two isolated nodes $u_t,v_t$. If $x_t = 1$, also include an edge $\set{u_t,v_t}$. Therefore, each time step contains either $1$ new connected component or $2$ new connected components.

        Assume for contradiction that there exists an $(\eps,\del)$-edge-DP mechanism $\mM$ for solving $\connectedcomps$ with $(\alpha,T)$-accuracy where $\alpha = o \paren{ \min \set{ \frac{ \log T}{\eps}, T } }$. Then applying the transformation described above to a binary stream, running $\mM$ on the resulting stream, and subtracting 1 from the output at each time step will solve binary counting with error $\alpha = o \paren{ \min \set{ \frac{\log T}{\eps}, T } }$, which contradicts \cref{thrm:bin-count-lower-bd}.

        \item \textbf{($\bm{k}$-stars.)} For each time step $t\in[T]$, let $\gs_t$ contain a $(D-1)$-star with center node $c_t$, plus one isolated node $v_t$. If $x_t = 1$, also include the edge $\set{c_t,v_t}$. Therefore, each time step contains either $\binom{D-1}{k}$ new $k$-stars or $\binom{D}{k} = \binom{D-1}{k} + \binom{D-1}{k-1}$ new $k$-stars.

        Assume for contradiction that there exists an $(\eps,\del)$-edge-DP mechanism $\mM$ for solving $\kstars$ with $(\alpha,T)$-accuracy where $\alpha = o \paren{ \min \set{ \frac{D^{k-1} \log T}{\eps}, D^{k-1} T } }$. Then applying the transformation described above to a binary stream, running $\mM$ on the resulting stream, subtracting $\binom{D-1}{k}$ from the output at each time step and dividing this result by $\binom{D-1}{k-1} = \Theta(D^{k-1})$ will solve binary counting with error $\alpha = o \paren{ \min \set{ \frac{\log T}{\eps}, T } }$, which contradicts \cref{thrm:bin-count-lower-bd}.

        \item \textbf{(Degree histograms.)} For each time step $t\in[T]$, let $\gs_t$ contain two isolated nodes $u_t,v_t$. If $x_t = 1$, also include an edge $\set{u_t,v_t}$. Therefore, the histogram bin for nodes of degree $1$ increases each time step by either $0$ new nodes or $2$ new nodes.

        Assume for contradiction that there exists an $(\eps,\del)$-node-DP mechanism $\mM$ for solving $\deghist$ that has $(\alpha,T)$-accuracy for some bin with $\alpha = o \paren{ \min \set{ \frac{\log T}{\eps}, T } }$. Without loss of generality, let it be the bin for nodes of degree $1$. Then applying the transformation described above to a binary stream, running $\mM$ on the resulting stream, and dividing the output for the bin for degree 1 at each time step by $2$ will solve binary counting with error $\alpha = o \paren{ \min \set{ \frac{\log T}{\eps}, T } }$, which contradicts \cref{thrm:bin-count-lower-bd}.\qedhere
    \end{enumerate}
\end{proof}

\fi

\section{Experiments}
\label{sec:experiments}

We implemented \cref{alg:bb} for the task of estimating the number of edges.
We ran our algorithm on several synthetic graph streams (which we describe below) and compared the accuracy of our algorithm for counting edges to the accuracy of a direct application  of batch model algorithms for counting edges with advanced composition. We used the implementation of the (standard) binary tree mechanism from \cite{AnderssonP23}.

The Python code for the algorithm and synthetic graph generation is  on GitHub: \\
\href{https://github.com/cwagaman/time-aware-proj}{https://github.com/cwagaman/time-aware-proj}.

\myparagraph{Parameters for node privacy.} All of our experiments use $\varepsilon =1$ and $\delta = 10^{-10}$ (about $2^{-33}$), and are $(\varepsilon,\delta)$-node-DP. 

\myparagraph{Input streams.} We look at (1) \emph{random graphs} with $n = 10^6$ nodes and $m = 2\cdot 10^8$ edges (with the edges drawn uniformly without replacement from the set of possible edges), and 
(2) \emph{two-block graphs} with $n = 10^6$ nodes and $m = 2\cdot 10^8$ edges (with the edges drawn uniformly without replacement from the set of possible edges), except for $5{,}000$ randomly selected nodes that have degree $10{,}000$ (with these edges drawn uniformly at random from the set of edges incident to the randomly chosen nodes). In both cases, the average degree in the graph is $\frac {2m}{n}=400$, but in the two-block cases there are nodes with 25 times the average degree.

In both cases, we consider a stream with $T=10^6$ time steps. A uniformly random subset of $\frac m T =200$ edges arrives at each time step.

\myparagraph{Degree cutoffs.} Recall that the algorithms we consider require an analyst-specific degree cutoff $D$. We conduct three experiments overall. For the random graphs (with maximum degree about 400), we conduct experiments with $D=400$ and $D=1{,}000$---these reflect settings where the  maximum degree estimate is tight or  conservatively large. For the two-block stream with maximum degree 10,000, we use $D=15{,}000$, which corresponds to a slight overestimate of the maximum degree.

One could get around the need to specify $D$ by running several parallel copies of the algorithm with different degree cutoffs (say, using powers of 2 up to some reasonable limit, and choosing the output corresponding to the smallest value of $D$ for which the PTR test has not yet rejected). We did not explore this approach in our experiments.

\newcommand{\mywidth}{0.30\textwidth}

\begin{figure*}[tb]
  \centering
    \includegraphics[width=\mywidth]{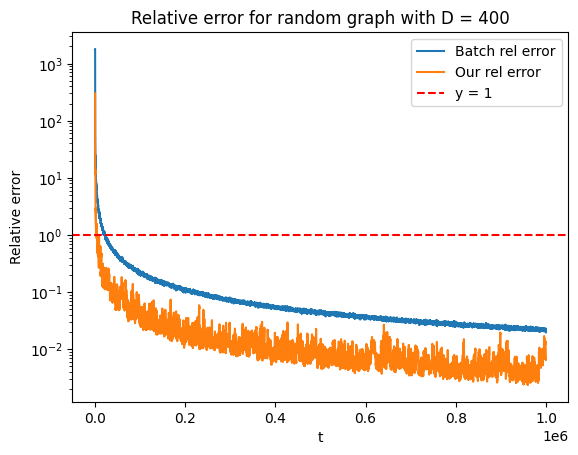}
    \hfill
    \includegraphics[width=\mywidth]{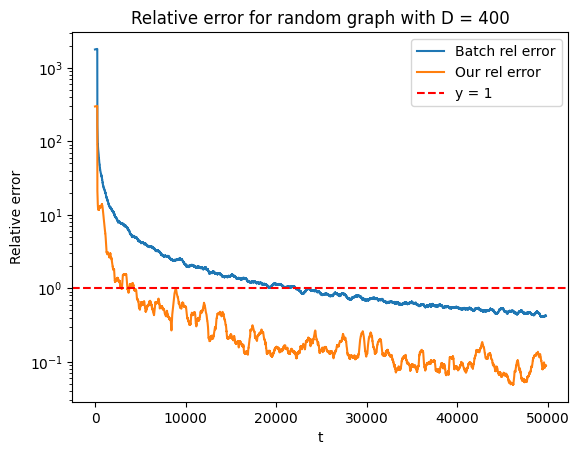}
    \hfill 
    \includegraphics[width=\mywidth]{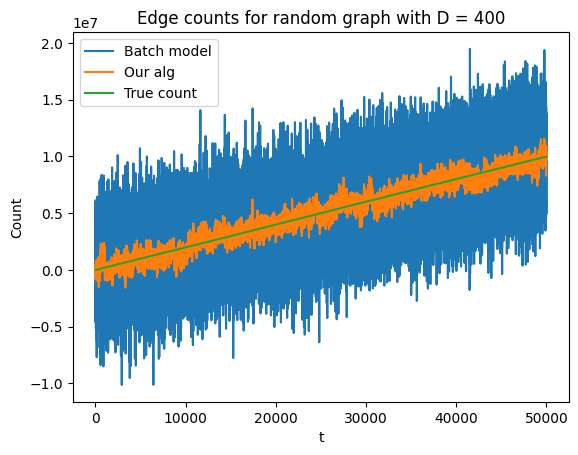}

    \includegraphics[width=\mywidth]{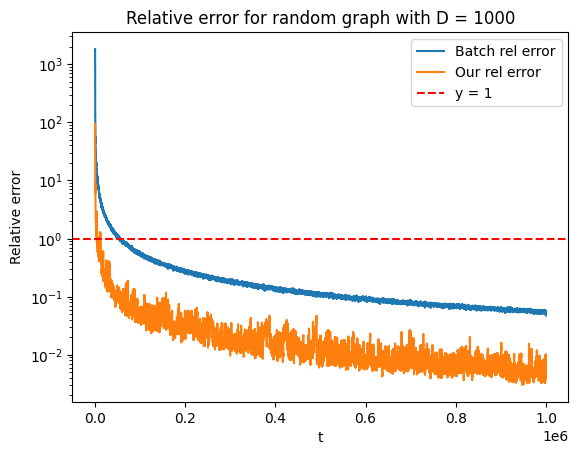}
    \hfill
    \includegraphics[width=\mywidth]{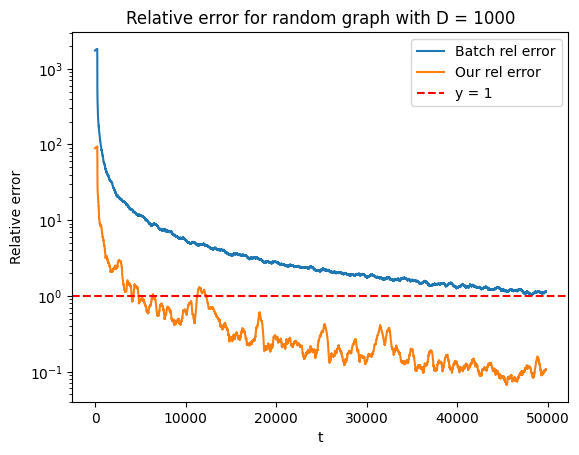}
    \hfill 
    \includegraphics[width=\mywidth]{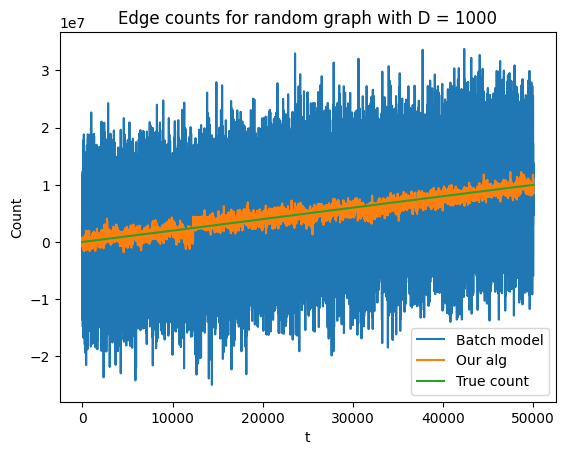}

    \includegraphics[width=\mywidth]{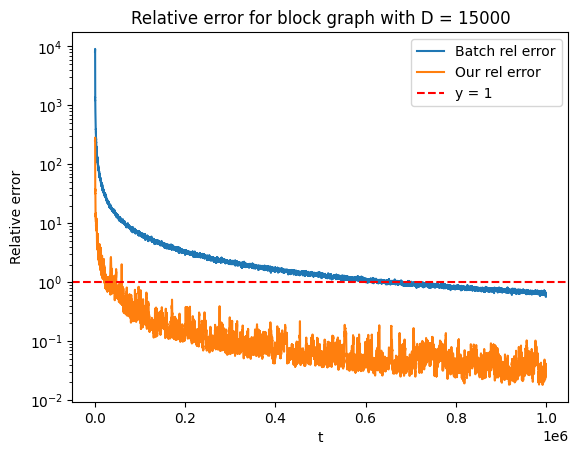}
    \hfill
    \includegraphics[width=\mywidth]{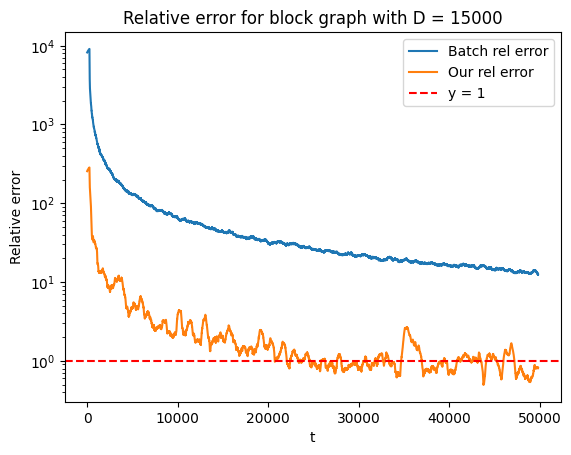}
    \hfill 
    \includegraphics[width=\mywidth]{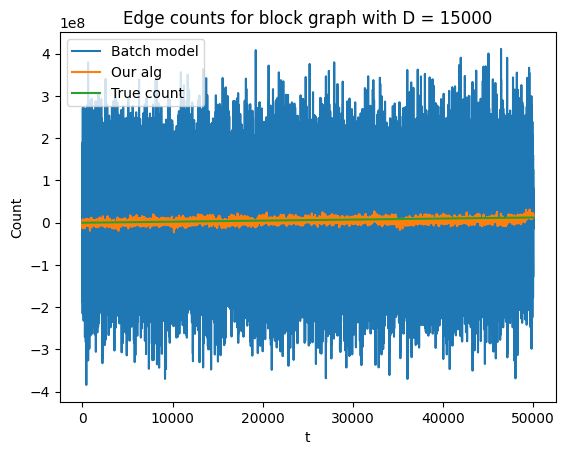}

    \caption{Results from one run of our algorithm (orange) and the batch-model baseline (blue) on three stream/cutoff pairs: a random graph stream  with $D=400$ (top row); a random graph stream with $D=1{,}000$; middle row); a two-block graph stream with $D=15{,}000$. For each graph, we provide three plots. 
    Left: Relative error across all $10^6$ time steps. Center: Relative error across first $50{,}000$ steps. 
    Right: True edge counts (green) together with reported counts from each algorithm across first $50{,}000$ steps. The left and center plots include a red dashed line at relative error 1 for visual reference. 
    }
    \label{fig:all-plots}
\end{figure*}

\myparagraph{Results.} We present results via two different types of plots (\cref{fig:all-plots}). The first type shows the actual edge count, the value reported by our algorithm, and the value reported by algorithms that follow from prior work. The second type shows the relative error of our algorithm and that of  prior work. The denominator in the relative error is the current edge count, which increases linearly over time.
We show these plots for all $10^6$ time steps, along with a ``zoomed-in'' version of these plots for the first 50,000 time steps. These latter plots allow one to see the variability of the error over time.

\myparagraph{Running time.} As stated in %
\cref{thrm:bb priv params}, our algorithm requires $O(n + m)$ additional time (total) and $O(n)$ additional space on a stream with $n$ nodes and $m$ edges. 
On a Microsoft Surface laptop with 8GB of RAM, our (unoptimized) implementation for edge counting runs in about $6\cdot 10^{-5}$ seconds per edge on the random graph described above with $n=10^6$ nodes and $m\approx 2 \cdot 10^8$ edges. For comparison, merely reading the graph and tracking the degree of each node (with the same machine and programming environment) takes $1.5\cdot 10^{-5}$ seconds per edge. Compared to this elementary processing, our algorithm takes more time by a factor of 4; optimized implementations would presumably see a similarly small run time increase.

\myparagraph{Discussion.} Our experimental results, displayed in \Cref{fig:all-plots}, show natural regimes where our new algorithm offers much better accuracy than the batch-model baseline. 
Furthermore, 
in all three experiments, the relative error of our algorithm drops below $1$ early in the stream ($10{,}000$ time steps for the random graphs, and $50{,}000$ for the two-block model). The error of our algorithm largely reflects the asymptotic error bounds, showing that the advantage of our algorithm holds even for modest settings of parameters. Additionally, implementing our algorithm was straightforward, and it is lightweight and runs quickly, even without optimizations.

Our algorithm's error is driven by the specified degree cutoff $D$ (assuming it produces any output at all, which requires roughly that $D$ be not much smaller than the actual maximum degree). The experiments reflect this: the relative error is lower in the random graph stream, where the maximum and average degree are basically the same, than in the two-block graph. Additionally, the relative error is lower when the degree cutoff is close to the maximum degree ($D=400$ as opposed to $D=1{,}000$, for the random graph stream). Choosing the value of $D$ automatically (as discussed under ``Degree cutoffs'', above) would help address overestimation of the maximum degree. However, variance among the degrees is a more fundamental obstacle, since the sensitivity of the edge count to the removal of any single vertex is equal to the maximum degree, but relative error compares to the total number of edges.

These experiments also reveal some limitations of our algorithm. Our algorithm's advantage over the baseline comes the fact that its error scales with $\polylog T$, instead of the $\sqrt{T}$ scaling that appears when composing the batch-model algorithms. As a result, $T$ must be large to obtain a significant accuracy advantage. Our algorithm's advantage over the baseline also improves for larger values of $D$, since then the fixed, additive slack term from the PTR framework becomes less significant. %
Designing an algorithm with low additive error in practice for all ranges of $D$ (including, say, with $D$ a small constant) remains an open question.

\section*{Acknowledgements}
\addcontentsline{toc}{section}{Acknowledgements} 
The authors would like to thank the anonymous Oakland reviewers as well as colleagues---Ephraim Linder, Sofya Raskhodnikova, Satchit Sivakumar, Theresa Steiner, and Marika Swanberg---for many helpful comments. The authors were supported in part by NSF awards CCF-1763786 and CNS-2120667, Faculty Awards from Google and Apple, and Cooperative Agreement CB16ADR0160001 with the Census Bureau. The views expressed in this paper are those of the authors and not those of the U.S. Census Bureau or any other sponsor.

\ifnum\oakland=1
    \begin{appendices}
    
        \section{}
        
        \yoakland{\subsection{Properties of Differential Privacy}}
\noakland{\myparagraph{Properties of Differential Privacy.}}
The definition of differential privacy extends to groups of individuals. If an algorithm is DP for one individual, it also offers a (differently parameterized) privacy guarantee for a collection of individuals. We borrow the formulation of this property from \cite{Vadhan17}.

\begin{lemma}[DP offers group privacy \cite{DworkMNS16}]
\label{lemma:group priv}
Let $\mM:\mX\to \mY$ be a randomized algorithm that is $(\eps,\del)$-DP. Then, where $x,x'\in \mX$ differ in the data of $k$ individuals, $\mA(x)$ and $\mA(x')$ are $(k\cdot \eps, k\cdot e^{k\eps} \cdot \del)$-indistinguishable. That is,
\[
    \mA(x) \approx_{k\cdot \eps, k\cdot e^{k\eps} \cdot \del} \mA(x').
\]
\end{lemma}

This follows from a well-known ``weak triangle inequality'' for $(\eps,\delta)$-indistinguishability: 

\begin{lemma}[Weak triangle inequality]
\label{lem:weak triangle ineq}
    For all $\eps_1, \eps_2, \delta_1 ,\delta_2 \geq 0$: If random variables $A,B,C$ satisfy $A\approx_{\eps_1,\delta_1} B$ and $B \approx_{\eps_2,\delta_2} C$, then $A\approx_{\eps',\delta'} C$ for  $\eps' = \eps_1 + \eps_2$ and $\delta' = \max(\delta_1  + e^{\eps_1}\delta_2,\ \delta_2+ e^{\eps_2}\delta_1) \leq  e^{\eps_2}\delta_1+ e^{\eps_1}\delta_2$.
\end{lemma}

Differential privacy is robust to post-processing.

\begin{lemma}[DP is robust to post-processing \cite{DworkMNS16}]
\label{lemma:post proc}
Let $\mM:\mX\to \mY$ be a randomized algorithm that is $(\eps,\del)$-DP. Let $f:\mY\to \mZ$ be an arbitrary, randomized mapping. Then $f\circ \mM:\mX\to \mZ$ is $(\eps,\del)$-DP.
\end{lemma}

        \subsection{Edge-DP under Continual Observation}

We can use the edge-to-edge stability of $\projo$ to transform algorithms that satisfy $D$-restricted edge-DP under continual observation that satisfy (pure or approximate) edge-DP under continual observation. Up to logarithmic factors, the algorithms we obtain are asymptotically optimal for $\triangles,\connectedcomps,\kstars$. 

The errors for our edge-private algorithms follow from item (1a) of \cref{thrm:combined-stab}, which says that we can achieve $(\eps,0)$-edge-DP by running the algorithms of \cite{FHO21} for $\triangles, \kstars, \deghist$, and our algorithm \noakland{described in \cref{thrm:priv alg for cc}} for $\connectedcomps$, with privacy parameter $\eps' = \eps/3$. (Because item (1a) of \cref{thrm:combined-stab} holds unconditionally, there is no need to use the PTR framework.) In contrast with previous work \cite{FHO21,SLMVC18}, this allows us to offer privacy on all graph streams, and we maintain the accuracy guarantees of the restricted private algorithm (up to constant factors) for graph streams with maximum degree at most $D$.

\begin{theorem}[Accuracy of edge-private algorithms]
\label{thrm:acc of applications edge}
    Let $\eps>0$, $D\in\N$, $T\in\N$,
    and $\gs$ be a length-$T$ graph stream. 
    There exist $(\eps,0)$-edge-DP algorithms for the following problems whose error is at most $\alpha$, with probability $0.99$, for all times steps $t\in[T]$ where $\gs_{[t]}$ has maximum degree at most $D$:
    \begin{enumerate}[leftmargin=*]

        \item $\triangles$, 
        $
            \alpha = O\left(\dfrac{D\log^{3/2} T}{\eps}\right).
        $
        
        \item $\connectedcomps$, 
        $
            \alpha = O\left(\dfrac{\log^{3/2} T}{\eps}\right).
        $
        
        \item $\kstars$, 
        $
            \alpha = O\left( \dfrac{D^{k-1}\log^{3/2}T}{\eps} \right).
        $
        
        \item $\deghist$, 
        $
            \alpha = \widetilde{O}\left(\dfrac{D\log^{3/2} T}{\eps} \right),
        $
        where $\widetilde{O}$ hides $\log D$ and $\log(\log T / \eps)$ factors.
        ~\connornote{uses edge priv}~\connornote{I think bounds from \cite{FHO21} should be increased by a factor of $\poly \log D$, so that would increase our bounds by $\poly\log \max\{D,\ell\}$. If they there were $D$ laplace r.v.s, then we'd increase things by a factor of $\log D$; since it's weird binary tree things, maybe $\poly \log D$ is sufficient?}
    \end{enumerate}
\end{theorem}

\begin{proof}[Proof of \cref{thrm:acc of applications edge}]
    This proof follows immediately from combining item (1a) of \cref{thrm:combined-stab}, which describes the edge-to-edge stability of \cref{alg:time aware projection} that includes edges in the projection according to original degrees (which we call $\projo$), with the edge-private algorithms of \cite{FHO21} and with our edge-private algorithm for $\connectedcomps$. In particular, for $\edges$ and $\triangles$ we can run the restricted edge-private algorithms of \cite{FHO21} with $\eps'=\eps/3$ on outputs from the projection $\projo$; and for $\connectedcomps$ we do the same, except with our algorithm\noakland{~described in \cref{thrm:priv alg for cc}}.
\end{proof}

\ifnum\oakland=0

\subsection{\texorpdfstring{\boldmath{$\connectedcomps$}}{Connected Components} with Restricted Edge-DP under Continual Observation}

We now present our algorithm, based on the binary tree mechanism, for $\connectedcomps$ that offers $D$-restricted edge-privacy under continual observation. Per \cref{thrm:acc of applications edge}, this algorithm can then be run on the output of $\projo$ to offer an algorithm for $\connectedcomps$ that is $(\eps,0)$-edge-DP under continual observation.

\begin{theorem}[\cite{FHO21}]
\label{thrm:fho edge acc}
    For all $\eps' > 0$, $\beta' \in (0,1]$, $D'\in \N$, there exists
    \begin{enumerate}
        \item an $(\eps',0)$-edge-DP algorithm that approximates $\edges$ and, with probability at least $1-\beta'$, has additive error at most $O(\log^{3/2}T \cdot \log(1/\beta')/\eps')$.
        \item a $D'$-restricted $(\eps',0)$-edge-DP algorithm that approximates $\triangles$ and, with probability at least $1-\beta'$, has additive error at most $O(D' \log^{3/2}T \cdot \log(1/\beta')/\eps')$.
    \end{enumerate}
\end{theorem}

\begin{theorem}[Private algorithm for $\connectedcomps$]
\label{thrm:priv alg for cc}
    For all $\eps' > 0$, $\beta' \in (0,1]$, there exists an $(\eps',0)$-edge-DP algorithm that approximates $\connectedcomps$ and, with probability at least $1-\beta'$, has additive error at most $O(\log^{3/2}T \cdot \log(1/\beta')/\eps')$.
\end{theorem}

We now provide some statements (\cref{thrm:bin tree for diff seq,lemma:cc edge sens}) that we will use in proving \cref{thrm:priv alg for cc}. To make an edge-private algorithm for releasing the number of connected components in a graph stream, we leverage \cite[Corollary 13]{FHO21}. Here, we present a more general version of that corollary.

\begin{theorem}
\label{thrm:bin tree for diff seq}
    Let $f:\mX^T\to \R^T$ return an output in $\R$ at each time step $t\in [T]$. For an input $x\in\mX^T$, let $f^+:\mX^T \to \R^T$ return, at each time step $t\in[T]$, the increment $f(x)_t - f(x)_{t-1}$, where we define $f(x)_0 = 0$.

    If for all neighbors $x\simeq x'\in \mX^T$, the $\ell_1$ distance between outputs is bounded by
    \[
        ||f^+(x) - f^+(x') ||_1 \leq \Gamma,
    \]
    there exists an $(\eps,0)$-DP algorithm that with probability at least $1-\beta$ has error $O(\Gamma \log^{3/2} T \cdot \log(1/\beta) / \eps)$.
\end{theorem}
\begin{proof}[Proof of \cref{thrm:bin tree for diff seq}]
    This proof follows from the proof of \cite[Corollary 13]{FHO21}, which uses the binary tree mechanism of \cite{CSS11,DNPR10} and the technique of difference sequences introduced by \cite{SLMVC18}.
\end{proof}

In \cref{lemma:cc edge sens}, we show that for all pairs of neighboring graph streams, the $\ell_1$ distance between the sequences of increments for evaluating $\connectedcomps$ on these neighbors can be bounded in the following way. We can then substitute this result into \cref{thrm:bin tree for diff seq} to make an edge-private approximation for $\connectedcomps$.

\begin{lemma}[Edge-sensitivity of $\connectedcomps$ under continual observation]
\label{lemma:cc edge sens}
    Let $\connectedcomps:\mGS^T\to \R^T$ return, at each time step $t\in[T]$, the number of connected components in the graph stream so far. For an input graph stream $\gs$ of length $T$, let $\connectedcomps^+:\mGS^T \to \R^T$ output, for all $t\in[T]$, the increment $\connectedcomps(\gs)_t - \connectedcomps(\gs)_{t-1}$.
    Let $\gs, \gs'\in\mGS^T$ be edge-neighboring graph streams. The $\lone$ distance between $\connectedcomps^+(\gs)$ and $\connectedcomps^+(\gs')$ is bounded by the following:
    \[
        || \connectedcomps^+(\gs) - \connectedcomps^+(\gs') ||_1 \leq 2.
    \]
\end{lemma}

\begin{proof}[Proof of \cref{lemma:cc edge sens}]
    Without loss of generality, let $\gs'$ be a graph stream that, as compared to $\gs$, has one additional edge $e' = \{u,v\}$ that is inserted at time $t_1$.

    When bounding $|| \connectedcomps^+(\gs) - \connectedcomps^+(\gs') ||_1$, we have to consider the following three cases. In $\gs$:
    \begin{enumerate}
        \item \textbf{(Case 1)} there is never a path between $u$ and $v$.
        \item \textbf{(Case 2)} $t_0 \leq t_1$ is the first time step at which there exists a path between $u$ and $v$.
        \item \textbf{(Case 3)} $t_3 > t_1$ is the first time step at which there exists a path between $u$ and $v$.
    \end{enumerate}
    \textbf{Case 2}: The outputs from $\connectedcomps(\gs)$ and $\connectedcomps(\gs')$ are identical, so the outputs from $\connectedcomps^+(\gs)$ and $\connectedcomps^+(\gs')$ are also identical, so we have $||\connectedcomps^+(\gs) - \connectedcomps^+(\gs') ||_1 = 0$. We now consider the remaining cases.

    \textbf{Case 1:} For all $t_0 < t_1$ (that is, prior to the time step at which the graph streams differ) we have
    $\connectedcomps^+(\gs)_{t_0} = \connectedcomps^+(\gs')_{t_0}$. When the difference in graph streams occurs at $t_1$, we have $\connectedcomps^+(\gs)_{t_1} = \connectedcomps^+(\gs')_{t_1} + 1$. Only the components containing $u$ and $v$ are affected by this edge. By the assumption that there is never a path between $u$ and $v$ in $\gs$, no future edges affect the components containing $u$ and $v$, so for all $t_2 > t_1$ we have $\connectedcomps^+(\gs)_{t_2} = \connectedcomps^+(\gs')_{t_2}$. Therefore, in case 1 we have
    \[
        || \connectedcomps^+(\gs) - \connectedcomps^+(\gs') ||_1 \leq 1.
    \]

    \textbf{Case 3:} As above, for all $t_0 < t_1$ we have
    $\connectedcomps^+(\gs)_{t_0} = \connectedcomps^+(\gs')_{t_0}$. When the difference in graph streams occurs at $t_1$, we have $\connectedcomps^+(\gs)_{t_1} = \connectedcomps^+(\gs')_{t_1} + 1$. Only the components containing $u$ and $v$ are affected by the inclusion of edge $e'$. By the assumption that, for all $t_2\in [t_1,t_3)$, there is never a path between $u$ and $v$ in $\gs$, no edges arriving at such times $t_2$ affect the components containing $u$ and $v$, so we have $\connectedcomps^+(\gs)_{t_2} = \connectedcomps^+(\gs')_{t_2}$.
    
    At time $t_3$, the edge forming a path between $u$ and $v$ arrives, which causes the number of connected components in $\gs$ to decrease by 1 as compared to $\gs'$, so $\connectedcomps^+(\gs)_{t_3} = \connectedcomps^+(\gs')_{t_3}$. For all $t_4 > t_3$, the changes in counts will be the same for both graph streams, so $\connectedcomps^+(\gs)_{t_4} = \connectedcomps^+(\gs')_{t_4}$. Therefore, in case 3 we have
    \[
        || \connectedcomps^+(\gs) - \connectedcomps^+(\gs') ||_1 \leq 2,
    \]
    which completes the proof.    
\end{proof}

These results allow us to prove \cref{thrm:priv alg for cc}, which bounds the accuracy of our edge-private approximation for $\connectedcomps$.

\begin{proof}[Proof of \cref{thrm:priv alg for cc}]
    This proof follows immediately from \cref{thrm:bin tree for diff seq} and the sensitivity bound in \cref{lemma:cc edge sens}.
\end{proof}

\fi

        \section{The Sparse Vector Technique}

We use the sparse vector technique (SVT), introduced by \cite{DNRRV09} and refined by  \cite{RothR10,HardtR10,LSL17}, to continually check that the input graph satisfies the conditions of \cref{thrm:combined-stab}.
In \cref{alg:svt}, we provide a version of the sparse vector technique described in \cite[Algorithm 1]{LSL17}.

\begin{algorithm}[ht!]
    \caption{Mechanism $\svt$ for answering threshold queries with the sparse vector technique.}
    \label{alg:svt}
    \begin{algorithmic}[1]
        \Statex \textbf{Input:} Stream $\gs\in\mGS^T$; queries $q_1,q_2,\ldots$ of sensitivity 1; cutoff $c \in \N$; privacy parameter $\eps > 0$; threshold $\tau\in\R$.
        \Statex \textbf{Output:} Stream of answers in $\{\Above, \Below\}$.

        \State $\eps_1 = \eps_2 = \eps/2$
        \State $\countsf = 0$
        \State Draw $Z\sim \lap(1/\eps_1)$ 
        \For{each time $t\in 1,2,\dots,T$}
            \State Draw $Z_t\sim \lap(2c/\eps_2)$
            \If{$q_t(\gs) + Z_t \geq \tau + Z$ and $\countsf < c$}
                \State Output $\Above$
                \State $\countsf \mathrel{+}= 1$
            \Else \ output $\Below$
            \EndIf
        \EndFor
    \end{algorithmic}
\end{algorithm}

\begin{theorem}[Privacy of $\svt$ \cite{LSL17}]
\label{thrm:svt privacy}
    \cref{alg:svt} is $(\eps,0)$-DP under continual observation.
\end{theorem}

    \end{appendices}

    \newpage %

\section{Meta-Review}

The following meta-review was prepared by the program committee for the 2024
IEEE Symposium on Security and Privacy (S\&P) as part of the review process as
detailed in the call for papers.

\subsection{Summary}

The paper develops privacy preserving algorithms to continually release graph statistics (edge count, degree histogram, triangle count) as the graph is updated over time. There are two notions of differential privacy (DP) for graphs: edge privacy and node privacy. While node-DP provides a much more compelling privacy guarantee to individuals, it is quite challenging to develop node-DP algorithms as most graph statistics are highly sensitive to the addition or removal of a single node with high degree. Because of this challenge prior work had only achieved a weaker notion of ``degree restricted'' node-DP where the DP privacy guarantees only hold if one makes the (implausible) assumption that every node in the graph/stream has degree at most $D$. The paper uses graph projections to transform a ``degree restricted'' node-DP algorithm (or a ``degree restricted'' edge-DP algorithm) into an algorithm that satisfies node-DP unconditionally. The node-DP algorithms are shown to be accurate as long as the original graph/stream was ``degree restricted'' i.e., the maximum degree of any node in the graph is at most $D$ for some suitable parameter $D<n$.

\subsection{Scientific Contributions}
\begin{itemize}
    \item Creates a New Tool to Enable Future Science
    \item Provides a Valuable Step Forward in an Established Field
\end{itemize}

\subsection{Reasons for Acceptance}
\begin{enumerate}
    \item Creates a New Tool to Enable Future Science: The paper introduces the first node-DP algorithm to continually release fundamental graph statistics (edge count, degree histogram, triangle count) as the graph is updated over time.
    \item Provides a Valuable Step Forward in an Established Field: Prior work (e.g., \cite{BlockiBDS13}) used graph projections as a tool to preserve node-DP in the static setting. The paper shows how these graph projections can be extended to work in the continual release setting where the graph is updated over time.
\end{enumerate}

\subsection{Noteworthy Concerns} %
\begin{enumerate} %
    \item The node-DP algorithm in the paper is only shown to release accurate statistics if the underlying graph has maximum degree $D$. It is not clear whether the statistics will be accurate in the more plausible setting that the graph is $(D,\ell)$-bounded i.e., there are at most $\ell$ nodes whose degree is larger than $D$.
    \item The empirical evaluation only considers synthetic graph streams.
\end{enumerate}

\else
    \appendix
    
\section{The Sparse Vector Technique}
\label{sec:svt}

In this section, we describe the sparse vector technique, introduced by \cite{DNRRV09} and refined by  \cite{RothR10,HardtR10,LSL17}, and some of its standard properties. We use the sparse vector technique to continually check that the input graph satisfies the conditions of \cref{thrm:combined-stab}.
In \cref{alg:svt}, we provide a version of the sparse vector technique described in \cite[Algorithm 1]{LSL17}, and throughout the rest of \cref{sec:svt}, we prove some useful properties of this algorithm.
\cref{thrm:combined-sep-svt} presents accuracy properties of \cref{alg:svt} that we use in our construction of the general transformation described in \cref{sec:black box acc priv} from (restricted) private algorithms to node-private algorithms.
To prove this statement, we use \cref{thrm:output bad when bad,thrm:output good when good stream}, which present slight variations on standard theorems about accuracy guarantees for the sparse vector technique.

\begin{algorithm}[ht!]
    \caption{Mechanism $\svt$ for answering threshold queries with the sparse vector technique.}
    \label{alg:svt}
    \begin{algorithmic}[1]
        \Statex \textbf{Input:} Stream $\gs\in\mGS^T$; queries $q_1,q_2,\ldots$ of sensitivity 1; cutoff $c \in \N$; privacy parameter $\eps > 0$; threshold $\tau\in\R$.
        \Statex \textbf{Output:} Stream of answers in $\{\Above, \Below\}$.

        \State $\eps_1 = \eps_2 = \eps/2$
        \State $\countsf = 0$
        \State Draw $Z\sim \lap(1/\eps_1)$ 
        \For{each time $t\in 1,2,\dots,T$}
            \State Draw $Z_t\sim \lap(2c/\eps_2)$
            \If{$q_t(\gs) + Z_t \geq \tau + Z$ and $\countsf < c$}
                \State Output $\Above$
                \State $\countsf \mathrel{+}= 1$
            \Else \ output $\Below$
            \EndIf
        \EndFor
    \end{algorithmic}
\end{algorithm}

We now give some properties of \cref{alg:svt}. We first provide a theorem on the privacy of \cref{alg:svt}.

\begin{theorem}[Privacy of $\svt$ \cite{LSL17}]
\label{thrm:svt privacy}
    \cref{alg:svt} is $(\eps,0)$-DP under continual observation.
\end{theorem}

We next provide a statement on the accuracy of \cref{alg:svt}.

\begin{theorem}[Separation needed for accurate answers]
\label{thrm:combined-sep-svt}
    Consider \cref{alg:svt}, and let $\gs\in\mGS^T$, $c\in \N$, $\eps > 0$, $\tau\in\R$.
    \begin{enumerate}
        \item \textbf{(True answer is above the threshold.)} Fix $\del\in (0,1]$.
        To ensure that \cref{alg:svt} outputs $\Above$ at or before time step $t$ with probability at least $1-\delta$, it suffices to have
        \[
            q_t(\gs) \geq 8c \ln(1/\delta)/\eps + \tau.
        \]

        \item \textbf{(True answer is below the threshold.)} Fix $t'\in[T]$ and $\beta\in (0,1]$.
        To ensure that \cref{alg:svt} outputs $\Below$ on all queries $q_1,\ldots, q_{t'}$ with probability at least $1-\beta$, it suffices to have
        \[
            q_t(\gs)\leq 8c \ln(\beta / T)/\eps + \tau
        \]
        for all $t\leq t'$.
    \end{enumerate}
\end{theorem}

To prove \cref{thrm:combined-sep-svt}, we use several lemmas presented below. The proof of \cref{thrm:combined-sep-svt} appears at the end of this section.

\begin{lemma}
\label{thrm:lap sum tails}
    Let $\gs\in\mGS^T$, $c\in\N$, $\eps > 0$, $\tau\in\R$, and $\eps_1 = \eps_2 = \eps/2$. Additionally, let $Z\sim \lap(1/\eps_1)$ and $Z_t\sim \lap(2c/\eps_2)$.
    For all $x\in\R$ and queries $q_t$,
    \begin{enumerate}
        \item if $x\geq 0$ and $q_t(\gs)\geq x + \tau$, then
        \[
            \pr{q_t(\gs) + Z_t \leq \tau + Z} \leq \exp\left(\frac{-|x|\cdot \eps}{8c} \right).
        \]
        \item if $x\leq 0$ and $q_t(\gs)\leq x + \tau$, then
        \[
            \pr{q_t(\gs) + Z_t \geq \tau + Z} \leq \exp\left(\frac{-|x|\cdot \eps}{8c} \right).
        \]

    \end{enumerate}
\end{lemma}

\begin{proof}[Proof of \cref{thrm:lap sum tails}]
    We begin by proving item (1). Assume that $x\geq 0$ and $q_t(\gs)\geq x + \tau$. We have the following expressions:
    \begin{align*}
        \pr{q_t(\gs) + Z_t \leq \tau + Z}
        &\leq \pr{x + \tau + Z_t \leq \tau + Z} \tag{given $q_t(\gs)\geq x + \tau$} \\
        &= \pr{x\leq Z - Z_t} \\
        &= \pr{x \leq Z + Z_t} \tag{$Z_t$ and $Z$ are independent, symmetric r.v.s} \\
        &\leq \pr{x/2 \leq Z \cup x/2 \leq Z_t} \tag{at least one must occur for the line above} \\
        &\leq \pr{x/2\leq Z} + \pr{x/2\leq Z_t} \tag{union bound} \\
        &= \frac{1}{2} \exp\left( \frac{-x}{2} \cdot \eps/2 \right) + \frac{1}{2} \exp\left( \frac{-x}{2} \cdot \frac{\eps/2}{2c} \right)  \tag{Laplace CDF, r.v.s are continuous} \\
        &\leq \exp\left( \frac{-x\eps}{8c} \right), \tag{$\exp\left( \dfrac{-x}{2} \cdot \dfrac{\eps/2}{2c} \right) \geq \exp\left( \dfrac{-x}{2} \cdot \eps/2 \right)$, for $x \geq 0$}
    \end{align*}
    which complete the proof of item (1).

    The proof of item (2) follows from item (1), and from the symmetry of Laplace r.v.s and the fact that $Z$ and $Z_t$ are independent.
\end{proof}

\cref{thrm:output bad when bad} uses \cref{thrm:lap sum tails} to show a lower bound on the probability that \cref{alg:svt} returns $\Above$ when the true answer to the query $q_t$ is (at least) some value $x\geq 0$ above the threshold $\tau$.
Similarly, \cref{thrm:output good when good stream} uses \cref{thrm:lap sum tails} to show a lower bound on the probability that \cref{alg:svt} exclusively outputs values of $\Below$ when the  threshold $\tau$ is (at least) some value $-x\geq 0$ above the true answer to each query $q_1,\ldots, q_{t'}$.

\begin{lemma}[Probability of $\Above$ for $q_t(\gs) \geq x + \tau$]
\label{thrm:output bad when bad}
    Consider \cref{alg:svt}, and let $\gs\in\mGS^T$, $c\in \N$, $\eps > 0$, $\tau\in\R$.
    For all $x \geq 0$ and queries $q_1,q_2,\ldots$, if $q_t(\gs) \geq x + \tau$, then \cref{alg:svt} outputs $\Above$ at or before time step $t$ with probability at least
    \[
        1 - \exp\left( \frac{-x\eps}{8c} \right).
    \]
\end{lemma}
\begin{proof}[Proof of \cref{thrm:output bad when bad}]
    Consider the event that we have $q_t(\gs) + Z_t \geq \tau + Z$. If this event occurs, either $\Above$ will be output, or $\Above$ was output at some earlier time step. When the conditions in the theorem statement above hold, by item (1) of \cref{thrm:lap sum tails} this event occurs with probability at least
    \[
        1 - \exp\left( \frac{-x\eps}{8c} \right),
    \]
    which completes the proof.
\end{proof}

\begin{lemma}[Probability of $\Below$ for stream $q_1(\gs),\ldots, q_{t'}(\gs) \leq x + \tau$]
\label{thrm:output good when good stream}
    Consider \cref{alg:svt}, and let $\gs\in\mGS^T$, $c\in \N$, $\eps > 0$, $\tau \in \R$. Fix $t'\in [T]$.
    For all $x \leq 0$ and query streams $q_1,q_2,\ldots$, if for all $t\leq t'$ we have $q_t(\gs) \leq x + \tau$, then \cref{alg:svt} outputs $\Below$ for all time steps $t\leq t'$ with probability at least
    \[
        1 - T \cdot \exp\left( \frac{x\eps}{8c} \right).
    \]
\end{lemma}

\begin{proof}[Proof of \cref{thrm:output good when good stream}]
    Consider the event that, for all $t\leq t'$, we have $q_t(\gs) + Z_t < \tau + Z$. If this event occurs, the output at each time step $t\leq t'$ is $\Below$. We now consider the complement of this event. When the conditions in the theorem statement above hold, by item (2) of \cref{thrm:lap sum tails} and the union bound, this complement occurs with probability at most
    \[
        T \cdot \exp\left( \frac{x\eps}{8c} \right).
    \]
    Therefore, the event in which we're interested occurs with probability at least
    \[
        1 - T \cdot \exp\left( \frac{x\eps}{8c} \right),
    \]
    which completes the proof.
\end{proof}

Parts (1) and (2) of \cref{thrm:combined-sep-svt} follow from \cref{thrm:output bad when bad,thrm:output good when good stream}, respectively, and show how far above or below the threshold $\tau$ it suffices to have the true query answers to ensure that \cref{alg:svt} outputs $\Above$ or $\Below$ with some user-specified probability.

\begin{proof}[Proof of \cref{thrm:combined-sep-svt}]
We prove each item below.

    \textbf{(Item 1.)}
    Let $A$ be the event that \cref{alg:svt} outputs $\Above$ at or before time step $t$. We observe that we want $\pr{A} \geq 1-\delta$.
    Consider the case where we have $q_t(\gs)\geq 8c\ln(1/\delta)/\eps + \tau$. By \cref{thrm:output bad when bad}, where we set $x = 8c\ln(1/\delta)/\eps$, we have 
    \begin{equation}
    \label{eq:dist for bad}
        \pr{A} \geq 1 - \exp\left( \frac{-x\eps}{8c} \right).
    \end{equation}
    Substituting $x = 8c\ln(1/\delta)/\eps$ into Expression~\ref{eq:dist for bad} gives us
    \begin{align*}
        \pr{A}
        &\geq 1- \exp\left( \frac{-\eps\cdot 8c\ln(1/\delta)/\eps}{8c} \right) \\
        &= 1 - \exp(-\ln(1/\delta)) \\
        &= 1 - \delta.
    \end{align*}
    Therefore, we have $\pr{A} \geq 1-\delta$ when we have $q_t(\gs)\geq 8c\ln(1/\delta)/\eps + \tau$, which is what we wanted to show.

    \textbf{(Item 2.)}
    Let $B$ be the event that \cref{alg:svt} outputs $\Below$ on all queries $q_1,\ldots, q_{t'}$. We observe that we want $\pr{B} \geq 1-\beta$.
    Consider the case where, for all $t\leq t'$, we have $q_t(\gs)\leq 8c\ln(\beta/T)/\eps + \tau$. By \cref{thrm:output good when good stream},
    \begin{equation}
    \label{eq:dist for good}
        \pr{B} \geq 1 - T \cdot \exp\left( \frac{x\eps}{8c} \right).
    \end{equation}
    Substituting $x = 8c\ln(\beta/T)/\eps$ into Expression~\ref{eq:dist for good} gives us
    \begin{align*}
        \pr{B}
        &\geq 1 - T \cdot \exp\left( \frac{\eps\cdot 8c\ln(\beta/T)/\eps}{8c} \right) \\
        &= 1 - T \cdot \exp(\ln(\beta/T) \\
        &= 1 - T\cdot \beta/T \\
        &= 1 - \beta.
    \end{align*}
    Therefore, we have $\pr{B} \geq 1-\beta$ when we have $q_t(\gs)\leq 8c\ln(\beta/T)/\eps + \tau$ for all $t\in[t']$, which is what we wanted to show.
\end{proof}

\section{Tightness of Stabilities \texorpdfstring{in \cref{thrm:combined-stab}}{}}
\label{sec:stab-tight}

In \cref{thrm:tightness-stab}, we show that the upper bounds on stability in \cref{thrm:combined-stab} are tight: the worst-case lower bounds on stability are tight up to small additive constants for $\projo$, and are tight up to constant multiplicative factors for $\projp$. 
The lower bounds given in \cref{thrm:tightness-stab} are for input graphs. Recall from \cref{rmk:tap-graph-input} that \cref{alg:time aware projection} treats an input graph as a length-$1$ graph stream, so graphs can be viewed as a special case of graph streams.

\begin{lemma}[Tightness of \cref{thrm:combined-stab}]
\label{thrm:tightness-stab} 
    Let $D\in \N$, $\ell\in \N\cup\{0\}$, and let $\projo,\projp$ be \cref{alg:time aware projection} with \inccrit{} $\style = \original$ and $\style=\projection$, respectively.
    \begin{enumerate}[leftmargin=*]
        \item \textbf{(Edge-to-edge stability.)} There exist
        \begin{enumerate}
            \item edge-neighboring graphs $G,G'$ such that
            \( 
                \dedge\Big(\projo(G)\ ,\ 
                \projo(G') \Big)
                = 3.
            \)
        
            \item edge-neighboring, $(D,\ell)$-bounded graphs $G,G'$ such that
            \(
                \dedge\Big(\projp(G)\ ,\ 
                \projp(G') \Big) = \ell + 1.
            \)
        \end{enumerate}
        \item \textbf{(Node-to-edge stability.)} There exist
        \begin{enumerate}
            \item node-neighboring, $(D,\ell)$-bounded graphs $G,G'$ such that
            \( 
                \dedge\Big(\projo(G)\ ,\ 
                \projo(G') \Big)
                =
                D + \ell - 1.
            \)
                
            \item node-neighboring, $(D,\ell)$-bounded graphs $G,G'$ such that, for $D$ and $\ell$ sufficiently large,
            
            \(
                \dedge\Big(\projp(G)\ ,\ 
                \projp(G') \Big)
                = 
                \begin{cases}
                    D + \Omega(\ell^{3/2}) & \text{if } D\geq \ell,\ \text{and}\\
                    D + \Omega(\ell \sqrt{D}) & \text{if } D< \ell.\\
                \end{cases}
            \)
            \end{enumerate}

        \item \textbf{(Node-to-node stability.)} There exist
        \begin{enumerate}
            \item node-neighboring, $(D,\ell)$-bounded graphs $G,G'$ such that
            \( 
                \dnode\Big(\projo(G)\ ,\ 
                \projo(G') \Big)
                =
                2 \ell - 1.
            \)
            
            \item node-neighboring, $(D,\ell)$-bounded graphs $G,G'$ such that
            \(
                \dnode\Big(\projp(G)\ ,\ 
                \projp(G') \Big)
                =
                \ell + 1.
            \)
        \end{enumerate}
    \end{enumerate}
\end{lemma}

For the proof below, we construct pairs of graphs (each of which can be viewed as a length-$1$ graph stream) with the above-specified distances between their projections. Note that, for all $T\in\N$, the examples can be expanded to length-$T$ graph streams by having no nodes or edges arrive in the remaining $T-1$ time steps.

\begin{proof}[Proof of \cref{thrm:tightness-stab}]
    We prove each item below. We begin by proving the statements for $\projo$ (i.e., the ``a'' items) and then prove the statements for $\projp$ (i.e., the ``b'' items, with (2b) proved last since its proof is the most involved).
    
    \textbf{(Item 1a.)} Consider a pair of edge-neighboring graphs $G,G'$, where $G$ consists of two separate $D$-stars with centers $u$ and $v$, and where $G'$ is the same plus an additional edge $\eplus$ between $u$ and $v$. Additionally, let $\eplus$ be the first edge in the consistent ordering. All edges in $\gs$ are in $\projo(G)$. However, $\eplus$ is in $\projo(G')$ but is not in $\projo(G)$. Additionally, one edge incident to $u$ is not in $\projo(G')$, and one edge incident to $v$ is not in $\projo(G')$---otherwise,
    the projection would have added new edges to $u$ and $v$ despite having counter values $d(u) \geq D$ and $d(v)\geq D$.

    \textbf{(Item 2a.)} We construct a pair of node-neighboring graphs $G,G'$. Graph $G$ consists of a set of $D$ nodes with no edges, and $\ell - 1$ separate $D$-stars. $G'$ is the same, plus an additional node $\vplus$ with an edge to each of the $D$ empty nodes and each center of the $D$-stars. Additionally, let all of the edges incident to $\vplus$ appear before all other edges in the consistent ordering, and let the $D$ edges incident to empty nodes appear before all the other edges incident to $\vplus$.
    
    All edges in $G$ are in $\projo(G)$. However, all of the first $D$ edges incident to $\vplus$ are in $\projo(G')$ and are not in $\projo(G)$. Additionally, one edge incident to the center of each $D$-star is not in the projection because their counters are each at $D$ prior to the projection stage of the final edge incident to each $D$-star's center. There are $D$ edges incident to $\vplus$ that appear in $\projo(G')$ and not $\projo(G)$; and there are $\ell - 1$ centers of $D$-stars, each of which is missing an edge in $\projo(G')$ as compared to $\projo(G)$. We see, then, that the projections differ on $D + \ell - 1$ edges.

    \textbf{(Item 3a.)} Consider the same example used for the proof of item (2a). To obtain $\projo(G')$ from $\projo(G)$, we see that it is necessary to add $\vplus$. We also see that it is necessary to change the edges incident to one node in each of the $D$-stars. To change a node, it must be removed and then re-added with different edges. Since there are $\ell - 1$ separate $D$-stars, there are $2\ell - 2$ additions and removals of nodes in the $D$-stars, plus the addition of $\vplus$. Therefore, we see that $\projo(G)$ and $\projo(G')$ are at node distance $2\ell - 1$.

    \textbf{(Item 1b.)} Consider the following pair of edge-neighboring graphs $G,G'$. Graph $G$ is constructed by taking two empty nodes $v$ and $w$, and a set of $\ell$ separate $(D-1)$-stars, enumerating their centers $u_1,\ldots,u_\ell$, and then adding edges between $u_i$ and $u_{i+1}$ for all $i\in[\ell-1]$. Additionally, these edges should appear last in the consistent ordering, and they should appear such that, for all $i\in[\ell-2]$, edge $\{u_i,u_{i+1}\}$ precedes $\{u_{i+1},u_{i+2}\}$. There should also be an edge $\{u_\ell, w\}$ that appears last in the consistent ordering. Graph $G'$ is constructed in the same way, except there is also an edge $\{v,u_1\}$ that appears first in the consistent ordering.

    In $\projo(G)$, the edge $\{u_1,u_2\}$ appears because $d(u_1) = D-1$ and $d(u_2) = D-1$ at the projection stage of the edge. By contrast, $\{u_2,u_3\}$ does not appear because $d(u_2) = D$ at the projection stage of the edge. By similar logic $\{u_3,u_4\}$ appears, $\{u_4,u_5\}$ does not, and so on in this alternating fashion, with $\{u_\ell,w\}$ appearing if and only if $\ell$ is odd. On the other hand, for $\projo(G')$ the edge $\{v,u_1\}$ appears since $d(v) = D-1$ and $d(u_1) = D-1$ at the projection stage of this edge. By contrast, $\{u_1,u_2\}$ does not appear since $d(u_1) = D$ at the projection stage of this edge. Similarly, $\{u_2,u_3\}$ appears, $\{u_3,u_4\}$ does not, and so on, with $\{u_\ell,w\}$ appearing if and only if $\ell$ is even. We see that $\ell + 1$ edges differ between $\projo(G)$ and $\projo(G')$, which is what we wanted to show.

    \textbf{(Item 3b.)} Consider the same example used for the proof of item (1b), with the modification that $v$ does not appear in $G$. We see that to obtain $\projo(G')$ from $\projo(G)$ we need to---starting with $v$---remove every other node in the path $v,u_1,\ldots, u_\ell,w$, and then add a modified version of each of these nodes (except for $v$). There are $\ell + 2$ nodes in the path, so we need to remove $\lceil (\ell + 2)/2 \rceil \geq \ell/2 + 1$ nodes and add $\lceil (\ell + 1)/2 \rceil \geq \ell/2$ nodes, so the node distance between $\projo(G)$ and $\projo(G')$ is at least $\ell + 1$.

    \textbf{(Item 2b.)} Let $k = \min\{D,\ell\}$.
    Before describing $G$ and $G'$, we construct a graph $P$ that is a collection of paths and contains $\Omega(\ell\sqrt{k})$ edges, and we then show how to construct $G$ and $G'$ such that all edges in $P$ differ between the projections of $G$ and $G'$.
    The high-level idea for $P$ is to construct $m = \lfloor k/4 \rfloor$ simple paths on $\ell$ nodes, such that this collection of paths contains $\Omega(\ell\sqrt{k})$ edges.

    For the sake of exposition, we use directed edges to construct these paths; they can be replaced with undirected edges. Let $(u_1,\ldots,u_\ell)$ be some left-to-right ordering of an arbitrary set of $\ell$ nodes, and let $\vplus$ be some additional node that is ordered to the left of $u_1$. All of the edges in our constructed paths will go from left to right. Let the \emph{hop length} of an edge $(u_i,u_{i'})$, where $i'>i$, refer to the value $i'-i$.
    
    We build $m = \lfloor k/4 \rfloor$ simple paths on these nodes as follows; we denote by $P$ the resulting graph of simple paths. Let $p\in \set{1,\ldots,m}$ enumerate these paths. For all paths $p$, the first edge is from $\vplus$ to some node $u_{s_p}$, where we define $s_p = 2p-1$. We now describe the remaining edges in each path. For each odd integer $h\in\set{1,3,5,7,\ldots}$, construct $(h+1)/2$ paths, where we start with path $p = 1$, then path $p = 2$, and so on, up to and including path $p = m$. For a path $p$ being made with value $h$, the first edge is from $\vplus$ to $u_{2p-1}$; and the remaining edges are edges of hop length $h$, except for the final edge which goes to a special node $v_p$. That is, where $s_p = 2p-1$, path $p$ with hop length $h$ is the path $\vplus \to u_{s_p} \to u_{s_p + h} \to u_{s_p + 2h} \to \cdots \to u_{s_p + h\cdot \lfloor (\ell - s_p)/h \rfloor} \to v_p$.

    We now color the edges of $P$ in the following way. On each path, make the edges alternate between red and blue edges: color all edges leaving $\vplus$ with blue, color the next edge on each path red, color the following edge blue, and so on. Because hops and start indices are odd, we maintain the following invariants:
    \begin{enumerate}
        \item all blue edges $(u_i, u_{i'})$ have even $i$ and odd $i'$, and
        \item all red edges $(u_i, u_{i'})$ have odd $i$ and even $i'$.
    \end{enumerate}
    These invariants mean that, for a given node $u_i$, all of its in-edges have the same color, and all of its out-edges have the same color (which is the opposite of the in-edges' color). We also note that, for all nodes $u_i$, the number of in-edges to $u_i$ is equal to the number of out-edges from $u_i$. Additionally, all nodes $v_p$ have one in-edge and zero out-edges.

    We now show that $G$ and $G'$ can be constructed such that all edges in $P$ will differ between projections of $G$ and $G'$. More specifically, we describe a construction where all blue edges and no red edges will be in $G'$; and all red edges and no blue edges will be in $G$. Below, we describe how to construct $G'$; the graph $G$ is constructed identically, except $\vplus$ and all edges from $\vplus$ are removed from the stream.
    
    We construct $G'$ by taking $P$ and adding some additional nodes and edges, and imposing a consistent ordering on these edges. We first describe the additional nodes and edges to add. Let $\deg_{u}(P)$ denote the degree of $u$ in $P$. Add a set of $D - \deg_{\vplus}(P)$ isolated nodes, and add an edge between $\vplus$ and each of these nodes. For all $i\in[\ell]$, add a set of $D - \deg_{u_i}(P)/2$ isolated nodes, and add an edge between $u_i$ and each node in this set. Let all of these edges appear before all other edges in the consistent ordering.
    
    We now describe how to order the remaining edges. For all $i\in[\ell]$, let all in-edges to $u_i$ appear in the consistent ordering appear before all in-edges to $u_{i'}$ for $i' > i$, and after all in-edges to $u_{i''}$ for $i'' < i$. Additionally, let all in-edges to nodes of the form $v_p$ appear in the consistent ordering after all of the edges already described. 

    We next use the following claim; we prove it at the end of this proof.

    \begin{claim}
    \label{claim:edges-in-proj}
        For all $u_i$ in $P$, $\projp(G')$ contains the following edges in $P$:
        \begin{enumerate}
            \item if $i$ is odd, then the projection contains all in-edges to $u_i$ in $P$, and no out-edges from $u_i$ in $P$.
            \item if $i$ is even, then the projection contains no in-edges to $u_i$ in $P$, and all out-edges from $u_i$ in $P$.
        \end{enumerate}
    \end{claim}

    By \cref{claim:edges-in-proj}, we see that all edges from $u_i$ with even $i$, and all edges to $u_i$ with odd $i$ are in $\projp(G')$, and that no other edges in $P$ are in $\projp(G')$. This means that, by the two invariants provided above, all blue edges are in $\projp(G')$ and no red edges are in $\projp(G')$. A symmetric claim and argument can be used to show that all red edges are in $\projp(G)$ and no blue edges are in $\projp(G)$. This means that all edges in $P$ differ between the projections.

    To prove the theorem, we use the following claim about the number of edges in $P$; we prove it at the end of this proof.

    \begin{claim}
    \label{claim:num-edges-p}
        The graph $P$ contains $\Omega(\ell\sqrt{k})$ edges.
    \end{claim}

    Recall from above that all edges in $P$ differ between the projections. Note that, in addition to the edges in $P$, all of the additional $D - \deg_{\vplus}(P)$ edges from $\vplus$ appear in the projection of $G'$ and do not appear in the projection of $G$. By this observation and \cref{claim:num-edges-p}, $G$ and $G'$ project to graph streams which differ in at least $D - \deg_{\vplus}(P) + \Omega(\ell\sqrt{k})$ edges.
    
    We now show $D - \deg_{\vplus}(P) + \Omega(\ell\sqrt{k}) = D + \Omega(\ell\sqrt{k})$.
    Recall that $\deg_{\vplus}(P) = m$, where $m = \floor{k/4}$. We have two cases: $D \geq \ell$ and $D < \ell$. For $D < \ell$, we have $k = D$, so $m = \lfloor D/4 \rfloor = o(\ell\sqrt{D})$, which means $D - \deg_{\vplus}(P) + \Omega(\ell\sqrt{k}) = D + \Omega(\ell\sqrt{k})$. For $D \geq \ell$, we have $m = \lfloor \ell/4 \rfloor$, so $m = o(\ell^{3/2})$, which means $D - \deg_{\vplus}(P) + \Omega(\ell\sqrt{k}) = D + \Omega(\ell\sqrt{k})$. 
    
    Therefore, the projections of the graph streams differ in at least $D + \Omega(\ell\sqrt{k}) = D + \Omega(\ell\sqrt{\min\{D,\ell\}})$ edges, which is what we wanted to show and completes the proof.

    We now provide proofs of the two claims used in the proof.

    \begin{proof}[Proof of \cref{claim:edges-in-proj}]
    To prove the claim, we induct on $i$. We prove the base case for $i = 1$ and $i = 2$. For the base case, we see that the in-edge to $u_1$, namely $(\vplus,u_1)$, is in the projection. Additionally, we see that the out-edge from $u_1$, namely $(u_1,u_2)$, is not in the projection since $u_1$ already has degree $D$ prior to considering $(u_1,u_2)$ for inclusion in the projection. For $u_2$, we see that the in-edge to $u_2$, namely $(u_1,u_2)$, is not in the projection. Additionally, we see that the out-edge from $u_2$, namely $(u_2,u_3)$, is in the projection since $u_2$ has degree $D-1$ prior to considering $(u_2,u_3)$ for inclusion in the projection and $u_3$ has at most $D-1$ edges that are considered for addition prior to $(u_2,u_3)$.
    
    Assume the claim is true for all $i < j$. We now show the claim is true for $j$. One useful fact from the construction of $G'$ is that, for all $u_j$, the number of in-edges to $u_j$ in $P$ is equal to the number of out-edges from $u_j$ in $P$, and there are $D - \deg_{u_j}(P)/2$ additional edges from $u_j$ in $G'$.
    
    We first consider the case where $j$ is odd. The only in-edges to $u_j$ are from $\vplus$ and from nodes $u_i$ where $i<j$ and $i$ is even. By assumption, all out-edges from nodes of the form $u_i$ for $i<j$ where $i$ is even are in the projection. Additionally, $\vplus$ has degree at most $D$ by construction, so all of its edges can be in the projection; since there are $D - \deg_{u_j}(P)/2$ edges from $u_j$ to isolated nodes, and there are $\deg_{u_j}(P)/2-1$ in-edges to $u_j$ that are already included in the projection, the edge from $\vplus$ can also be included in the projection without exceeding the degree bound. However, once these edges are included in the projection, adding any more edges would cause the degree of $u_j$ to $D$ in the projection, so none of the remaining out-edges are included in the projection.

    We next consider the case where $j$ is even. The only in-edges to $u_j$ are from nodes $u_i$ where $i<j$ and $i$ is odd. By assumption, all out-edges from nodes of the form $u_i$ for $i<j$ where $i$ is odd are \emph{not} in the projection. Therefore, none of the in-edges to $u_j$ are in the projection. We next show that all of the out-edges from $u_j$ are in the projection. First, note that the $D - \deg_{u_j}(P)/2$ edges from $u_j$ to isolated nodes are all in the projection. Next, note that there are $\deg_{u_j}(P)/2$ remaining out-edges from $u_j$, so they can all be included in the projection without exceeding the degree bound. We now consider the nodes to which they are in-edges. There are two cases for these edges. If the edge is an in-edge to some node of the form $v_p$, including this edge will not cause $v_p$ to exceed its degree bound, so it will be included in the projection. Otherwise, it is an in-edge to a node of the form $u_{i'}$ where $i' > i$. At most $D - \deg_{u_{i'}}(P)/2$ edges are edges to the isolated nodes (all of these edges appear in the consistent ordering prior to the edge we are considering). At most $\deg_{u_{i'}}(P)/2 - 1$ other edges (i.e., the other in-edges to $u_{i'}$) appear in the consistent ordering prior to this edge. Therefore, including this edge will not cause $u_{i'}$ to exceed its degree bound, so it will be included in the projection.
    \end{proof}

    \begin{proof}[Proof of \cref{claim:num-edges-p}]
    We now show that there are $\Omega(\ell\sqrt{k})$ edges in $P$. To do this, we show the set of paths can be divided into $\Omega(\sqrt{m})$ disjoint sets of $\Omega(\ell)$ edges. We first show that, for all values of $h$ where we make $(h+1)/2$ paths, there are at least $\ell/4$ edges in the set of paths with length $h$. There are $m = \lfloor k/4 \rfloor$ paths, so for all $p$ we have $s_p = 2p - 1 \leq \ell/2$. This means a path $p$ with hop length $h$ contains at least $\ell/(2h)$ edges, since there are at least $\ell/(2h) - 2$ edges between nodes of the form $u_i,u_j$ for $i,j\geq s_p$, and there is one edge $\vplus \to u_{s_p}$ and one edge $u_{s_p + h\cdot \lfloor (\ell - s_p)/h \rfloor} \to v_p$. Since there are $(h+1)/2$ paths with hop length $h$, there are at least $\ell/4$ edges in the set of paths with length $h$. We next show, roughly, that we create $(h+1)/2$ paths for all odd natural numbers $h \leq \sqrt{m}$. In other words, we determine the biggest value of $h$ that is used in the procedure for making $m = \floor{k/4}$ paths. Since there are at most $m$ paths, we solve for $x$ in
    \[
         \sum_{h\in\set{1,3,5,\ldots}} \frac{h+1}{2} =  \sum_{i\in[x]} i = m
    \]
    and, using the identity $\sum_{i\in[n]} i = n(n+1)/2$, find $x \approx \sqrt{2 m}$. Therefore, the largest value of $h$ is (roughly) $\sqrt{8 m}$. Since there are $h = \Omega(\sqrt{m})$ disjoint sets of $\Omega(\ell)$ edges, we see that there are $\Omega(\ell\sqrt{m}) = \Omega(\ell\sqrt{k})$ edges in the graph. (In this case, there are roughly $\ell/4 \cdot \sqrt{2k}$ edges in the described graph.)
    \end{proof}
\end{proof}

\fi 

\printbibliography[heading=bibintoc]

\end{document}